\documentclass[12pt, letterpaper]{article}

\RequirePackage[T1]{fontenc}
\usepackage{amsmath}
\numberwithin{equation}{section}
\usepackage{amssymb}
\RequirePackage{amsthm}
\usepackage{authblk}
\usepackage{caption}
\usepackage{comment}
\usepackage{dsfont}
\usepackage{eurosym}
\usepackage[top=1in, bottom=2in, left=1in, right=1in]{geometry}
\usepackage[colorlinks=true, allcolors=blue]{hyperref}
\usepackage{natbib}
\bibpunct{(}{)}{,}{a}{}{,}
\usepackage{mathrsfs}
\usepackage{mathtools}
\usepackage{upquote}
\usepackage{subfig}
\usepackage{booktabs}
\usepackage{lscape}
\usepackage[dvipsnames]{xcolor}
\usepackage[normalem]{ulem}
\usepackage{makecell}

\theoremstyle{plain}
\newtheorem{assumption}{Assumption}[section]
\newtheorem{lemma}{Lemma}[section]
\newtheorem{proposition}{Proposition}[section]
\newtheorem{remark}{Remark}[section]
\newtheorem{theorem}{Theorem}[section]

\theoremstyle{remark}
\newtheorem{definition}{Definition}[section]
\newtheorem{example}{Example}[section]

\DeclareMathOperator{\Exp}{\mathds E}

\DeclareMathOperator{\Prob}{\mathds P}

\DeclareMathOperator{\diag}{\mathrm diag}

\begin{document}
\title{An operator-level ARCH Model}

\author[1]{Alexander Aue}
\author[2]{Sebastian K\"uhnert}
\author[3]{Gregory Rice}
\author[4]{Jeremy VanderDoes}

\affil[1]{Department of Statistics, University of California, Davis, CA 95616 Davis, USA, \textsuperscript{\ddag}\href{mailto:aaue@ucdavis.edu}{aaue@ucdavis.edu}}
\affil[2]{Department of Mathematics, Ruhr University Bochum, 44780 Bochum, Germany, \textsuperscript{*}\href{mailto:sebastian.kuehnert@ruhr-uni-bochum.de}{sebastian.kuehnert@ruhr-uni-bochum.de}}
\affil[3,4]{Department of Statistics and Actuarial Science, University of Waterloo, ON N2L 0A4 Waterloo, Canada, 
\textsuperscript{\dag}\href{mailto:grice@uwaterloo.ca}{grice@uwaterloo.ca}, 
\href{mailto:Jeremy.VanderDoes@uwaterloo.ca}{Jeremy.VanderDoes@uwaterloo.ca}}

\date{March 15, 2026}

\maketitle

\begin{abstract} AutoRegressive Conditional Heteroscedasticity (ARCH) models are standard for modeling time series exhibiting volatility, with a rich literature in univariate and multivariate settings. In recent years, these models have been extended to function spaces. However, functional ARCH and generalized ARCH (GARCH) processes established in the literature have thus far been restricted to model ``pointwise'' variances. In this paper, we propose a new ARCH framework for data residing in general separable Hilbert spaces that accounts for the full evolution of the conditional covariance operator. We define a general operator-level ARCH model. For a simplified  Constant Conditional Correlation version of the model, we establish conditions under which such models admit strictly and weakly stationary solutions, finite moments, and weak serial dependence. Additionally, we derive consistent Yule--Walker-type estimators of the infinite-dimensional model parameters. The practical relevance of the model is illustrated through simulations and a data application to high-frequency cumulative intraday returns.
\end{abstract}



\section{Introduction}\label{Introduction}

Conditionally heteroscedastic processes are fundamental in financial modeling. To capture their dynamics, \cite{Engle1982} introduced the \textit{Autoregressive Conditional Heteroscedastic} (ARCH) model, later extended by \cite{Bollerslev1986} to the \textit{Generalized ARCH} (GARCH) model. Extensive reviews of uni- and multivariate (G)ARCH as well as other related volatility models can be found in \cite{AndersenEtAl2009}, \cite{mvgarch2}, \cite{mvgarch1}, \cite{FrancqZakoian2019}, and \cite{Gourieroux1997}.

With advances in data processing and the growing availability of high-frequency financial data, functional volatility and (G)ARCH (f(G)ARCH) models have gained prominence. As a motivating example, consider observing on consecutive days $k\in\{1, \dots,N\}$ the price of the S\&P 500 index at a resolution of one minute. Conveniently, these observations can be viewed as a sample of curves or functions $\{P_k(t)$, $i \in \{1,\dots,N\},\;t\in[0,1]\}$, where intraday time is normalized to the unit interval; see the left panel of Figure \ref{fig:sp500}. A transformation of particular interest are the overnight intraday log-returns $X_k(t)= \log P_k(t) - \log P_{k-1}(1)$; see the right panel of Figure \ref{fig:sp500}. Representation of the data as functions constructed from noisy high-dimensional intraday observations allows for the use of techniques from functional data analysis. We refer the reader to \cite{HorvathKokoszka2012} and \cite{HsingEubank2015} for introductions to functional data analysis, and to  \cite{Bosq2000} and \cite{BosqBlanke2007} for introductions to linear time series processes in function spaces.

\begin{figure}
    \centering
    \subfloat[S\&P 500 price curves]{
        \includegraphics[width=0.4\textwidth, trim=6.5cm 1.5cm 3cm 5cm, clip]{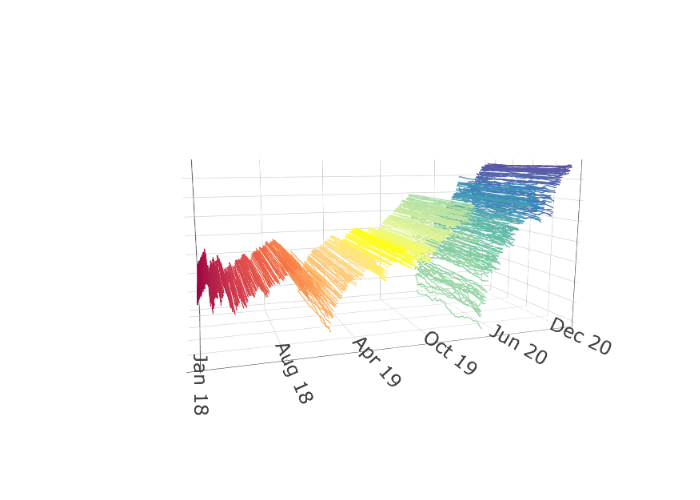}
    }\hfill
    \subfloat[S\&P 500 OCIDR]{
        \includegraphics[width=0.4\textwidth, trim=6.5cm 1.5cm 3cm 5cm, clip]{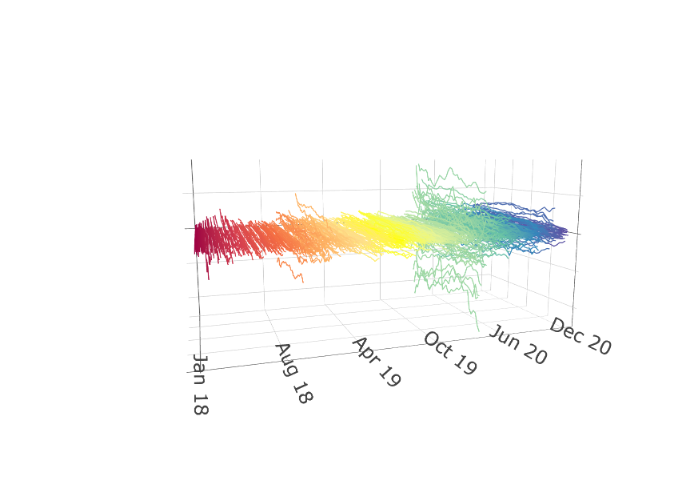}
    }
    \caption{\textbf{S\&P500 Data.} Visualization of price (left) and overnight cumulative intraday return (right) curves based on 15-minute resolution S\&P 500 stock market prices from 2018 to 2020}
    \label{fig:sp500}
\end{figure}

In seminal work, \cite{Hoermannetal2013} introduced the fARCH$(1)$ process to model conditional heteroscedasticity in functional time series data. Their model was extended to fGARCH$(1,1)$ by \cite{Aueetal2017}, and to general fGARCH$(p,q)$ by \cite{Ceroveckietal2019}. For further developments and modified functional GARCH and volatility models, see \cite{Andersen02042024}, \cite{KearneyShangZhao2023}, \cite{Kuehnert2019, Kuehnert2020}, \cite{Laksacietal2025}, \cite{LiSunLiu2025}, \cite{RiceEtAL2023}, \cite{SunYu2020}. 

Henceforth, we refer to these existing models as ``pointwise'' f(G)ARCH (pw-f(G)ARCH) models, since they are formulated as follows: For the function spaces $C[0,1]$ and $L^2[0,1]$ of continuous and square-integrable functions on $[0,1],$ the pw-fARCH$(p)$ model as introduced in \cite{Ceroveckietal2019} takes the form
\begin{align}\label{eq:pw-fARCH def}
    X_{k}(t) &= \sigma_{k}(t)\varepsilon_{k}(t), \quad
    \sigma^2_{k}(t) = \delta(t) + \sum_{i=1}^{p} \,\alpha_i(X^2_{k-i})(t), \quad k\in\mathbb{Z}; \qquad \Exp(\varepsilon^2_0(t))=1  . 
\end{align}
Here, the model parameters are $\delta$, a positive function, and $\alpha_i$, bounded linear operators mapping non-negative functions to non-negative functions, with the $\varepsilon_k$ being i.i.d.~and centered innovations. When there exists a stationary and causal solution to \eqref{eq:pw-fARCH def} with finite second-order moments, the condition $\Exp(\varepsilon^2_0(t))=1$ identifies $\sigma_k^2(t)= \mbox{Var}(X_k(t)| \mathcal{F}_{k-1})$ as the pointwise conditional variance, where $\mathcal{F}_{k} =\sigma( \varepsilon_i, \; i \le k)$ is the information filtration generated by the innovations. 

Despite its success in numerous applications, this model has some limitations. First, it only explicitly models the serial dependence structure of the pointwise conditional variance $\mbox{Var}(X_k(t)| \mathcal{F}_{k-1})$. For many tasks, for example the construction of a prediction set for $X_k$, one wishes to estimate the full conditional covariance kernel/operator of the sequence. Under model \eqref{eq:pw-fARCH def}, so long as it is well-defined, the conditional covariance kernel takes the form 
$$
    \mbox{Cov}\big(X_k(t),X_k(s) \big| \mathcal{F}_{k-1}\big) \;=\; \Exp\!\big[ X_k(t) X_k(s) \big| \mathcal{F}_{k-1}\big] \;=\; \sigma_k(t) \sigma_k(s)C_\varepsilon(t,s),
$$
where $C_\varepsilon$ is the covariance kernel of the innovations. As such, the conditional covariance forecast implied by model \eqref{eq:pw-fARCH def} constitutes a ``rank-one update'' to the covariance of the innovations, which may be overly simplistic. 

Similar critiques of the multivariate diagonal GARCH model have led to a plethora of alternative models of varying complexity, including the VEC, CCC, DCC, and BEKK GARCH; see \citet[][Ch.~10]{FrancqZakoian2019}. Analogues of these models for functional time series have not been considered, to the best of our knowledge. One potential reason for this is the fact that the identity map is not compact in general, separable Hilbert spaces; an issue that is circumvented in the formulation of pw-fARCH models. The non-compactness of the identity map prevents the user from specifying in a useful way an innovation process with identity covariance so that the conditional covariance of the model, and the model parameters, can be identified. This identifiability issue can be  overcome, however, by specifying that the innovations possess some known, injective covariance operator, which is a critical idea underlying our model.   

In this paper, we propose a new ARCH model, termed ``operator-level ARCH'', in general separable Hilbert spaces. It provides a more direct generalization of ARCH processes to the functional setting in that the model is formulated directly for the conditional covariance operator, rather than the pointwise variance as with pw-fARCH models. Since the general model is, from a theoretical point of view, quite challenging, we focus on an operator-level counterpart of the \emph{Constant Conditional Correlation} (CCC) ARCH model in the multivariate setting, which is more feasible. We establish conditions on the model parameters and innovations such that it admits strictly and weakly stationary solutions, refining the classical notion of the (functional) \emph{top Lyapunov exponent} for strict stationarity to suit our general setting. Additionally, we provide sufficient conditions for the existence of finite moments and weak dependence. Our estimation approach targets the ARCH operators under known orders and relies on pseudo Yule--Walker (YW) equations. We derive consistent YW estimators for the intercept term $\Delta,$ which is a self-adjoint and positive definite operator in our setting, and the operators $\alpha_i,$ which map between operator spaces. The estimates are developed under a specific diagonal representation in finite- and infinite-dimensional settings. The identifiability issue that arises when directly deriving the YW-type equations are addressed using specific transformations of the data. In the infinite-dimensional case, we introduce a Sobolev-type condition  inspired by \cite{hall:meister:2007} to quantify the approximability of infinite-dimensional operators by their finite-dimensional approximations. The estimation is conducted via Tikhonov-type pseudoinverses in the YW framework. While the theoretical properties of the estimators are well established, their complexity may present practical challenges. The applicability of the proposed methodology is illustrated through a simulation study an application to intraday log-returns.

The structure of the paper is as follows. Section \ref{Sec:General model} introduces the general model and outlines the main assumptions. Section \ref{Sec:CCC} introduces the corresponding CCC model and derives stationarity conditions and probabilistic properties. Section \ref{Sec:estimation} addresses parameter estimation. Section \ref{Section Simulation} presents a simulation study and Section \ref{Section Applications} provides applications. Section \ref{Section Discussion} concludes. Additionally, the appendix contains preliminaries (Sections \ref{Appendix: Section preliminaries} and \ref{Appendix: Section Notes}), and proofs of key results (Section \ref{Section proofs}).

We adopt the following notation. The identity map is denoted  by $\mathbb{I}$. On a Cartesian product space $V^n$, $n \in \mathbb{N}$, we define inner product and norm by $\langle x, y \rangle = \sum_{i=1}^n \langle x_i, y_i \rangle$ and $\|x\|^2= \sum_{i=1}^n \|x_i\|^2$ for $x = (x_1, \dots, x_n)^{\!\top},~y = (y_1, \dots, y_n)^{\!\top} \in V^n$, assuming $V$ has inner product $\langle \cdot, \cdot \rangle$ and norm $\|\cdot\|$. The spaces of bounded linear, Hilbert-Schmidt (H-S), and nuclear/trace-class operators from $\mathcal{H}$ to $\mathcal{H}_\star$ are respectively denoted by $\mathcal{L}_{\mathcal{H},\mathcal{H}_\star}$, $\mathcal{S}_{\mathcal{H},\mathcal{H}_\star}$, and $\mathcal{N}_{\mathcal{H},\mathcal{H}_\star}$, with norms $\|\cdot\|_{\mathcal{L}}, \|\cdot\|_{\mathcal{S}}, \|\cdot\|_{\mathcal{N}}$ and H-S inner product $\langle \cdot, \cdot\rangle_{\mathcal{S}}$. When $\mathcal{H}, \mathcal{H}_\star$ are clear, we write $\mathcal{T}$ for $\mathcal{T}_{\mathcal{H},\mathcal{H}_\star}$ and $\mathcal{T_H} \coloneqq \mathcal{T}_{\mathcal{H},\mathcal{H}},$ where $\mathcal{T} \in \{\mathcal{L,S,N}\}.$ For $A \in \mathcal{L}$, $A^\ast$ denotes the adjoint. Write $\mathcal{T}_{\geq 0}$ ($\mathcal{T}_{> 0}$) for the self-adjoint non-negative (positive) definite elements of $\mathcal{T}$. For $x \in \mathcal{H}, y \in \mathcal{H}_\star$, the \emph{tensor product operator} is $x \otimes y \coloneqq \langle x, \cdot\rangle y$, with $x^{\otimes 2} \coloneqq x \otimes x,$ and $x\!\otimes_{\mathcal{S}}\!y$ is used when $x, y$ are H-S operators. For $p \in [1,\infty)$, $L^p_\mathcal{H} = L^p_\mathcal{H}(\Omega, \mathfrak{A},\mathbb{P})$ is the space of $X \in \mathcal{H}$ with $\mathbb{E}\|X\|^p < \infty$. The \emph{cross-covariance operator} is
$$
    \mathscr{C}_{\!X,Y} \coloneqq \mathbb{E}[X-\mathbb{E}X ]\otimes [Y-\mathbb{E}Y], \quad X\in L^2_{\mathcal{H}},~ Y \in L^2_{\mathcal{H}_\star},
$$ 
and the \emph{covariance operator} is $\mathscr{C}_{\!X} = \mathscr{C}_{\!X,X}$, with expectation in the Bochner sense.  For weakly stationary $\boldsymbol{X}=(X_k)$ and jointly weakly stationary $\boldsymbol{X}, \boldsymbol{Y}$, the \emph{lag-$h$-covariance} and \emph{lag-$h$-cross-covariance} operators are $\mathscr{C}^h_{\!\boldsymbol{X}} \coloneqq \mathscr{C}_{\!X_0,X_h}$ and $\mathscr{C}^h_{\!\boldsymbol{X},\boldsymbol{Y}} \coloneqq \mathscr{C}_{\!X_0,Y_h}$, with $\mathscr{C}_{\!\boldsymbol{X}} \coloneqq \mathscr{C}^0_{\!\boldsymbol{X}}$ and $\mathscr{C}_{\!\boldsymbol{X},\boldsymbol{Y}} \coloneqq \mathscr{C}^0_{\!\boldsymbol{X},\boldsymbol{Y}}$. Further, $\boldsymbol{X} = (X_k)$ is a \emph{weak white noise} (WWN) if it is weakly stationary, centered, with $\mathbb{E}\|X_0\|^2 > 0$, and $\mathscr{C}^h_{\!\boldsymbol{X}} = 0$ for all $h \neq 0$; and a \emph{strong white noise} (SWN) is an i.i.d.~WWN. For sequences $(b_n), (c_n) \subset \mathbb{R}$, $b_n \asymp c_n$ (as $n \to \infty$) denotes asymptotic equivalence up to constants. Unless stated otherwise, all limits are taken as $N \to \infty$.

\section{General model and assumptions}\label{Sec:General model}
 
Let $\mathcal{H}$ be a separable Hilbert space equipped with the inner product $\langle\cdot, \cdot \rangle$.

\begin{definition}\label{Definition of fvec-GARCH} We call $(X_k)_{k\in\mathbb{Z}}\subset \mathcal{H}$ an \emph{operator-level ARCH$(p)$ process} (op-ARCH$(p)$) if 
\begin{align}\label{eq:definition op-ARCH}
	X_{k} = \Sigma^{1/2}_{k}(\varepsilon_{k}), \qquad \Sigma_{k} = \Delta + \sum_{i=1}^{p}\,\alpha_i(X^{\otimes 2}_{k-i}), \quad k\in\mathbb{Z},\allowdisplaybreaks
\end{align}
where $\boldsymbol{\varepsilon} = (\varepsilon_k)$ is a SWN with covariance operator $\mathscr{C}_{\boldsymbol{\varepsilon}}$, $p\in\mathbb{N},$ $\Delta \in \mathcal{S}_{>0},$ and $\alpha_i : \mathcal{S}\to \mathcal{S}$ are bounded linear operators with $\alpha_i:\mathcal{S}_{\geq 0}\to \mathcal{S}_{\geq 0}$ for all $i=1, 2, \dots, p,$ with $\alpha_p\neq 0.$ 
\end{definition}

\noindent As in the pw-fARCH framework, we call the parameter $\Delta$ the intercept term, and the parameters $\alpha_1, \dots, \alpha_p$ the (op-)ARCH operators. Here, $\Delta$ is a deterministic, self-adjoint, positive definite H-S operator, and each $\alpha_i$ is a bounded and linear operator mapping H-S to H-S operators, preserving self-adjointness and non-negative definiteness. As a consequence, $\Sigma_k$ is self-adjoint and positive definite. If \eqref{eq:definition op-ARCH} admits a strictly stationary solution with finite second moments, then
\begin{align}\label{eq:cond variance}
    \Exp\!\big[X^{\otimes 2}_{k}\,\big|\, \mathcal{F}_{k-1}\big] = \Sigma^{1/2}_{k}\mathscr{C}_{\boldsymbol{\varepsilon}}\Sigma^{1/2}_{k},\quad k \in \mathbb{Z}, 
\end{align}
so that the process exhibits conditional heteroscedasticity. Note that the model remains well-defined if $\Delta$ is merely bounded linear rather than H-S, but since consistent estimation from finite samples requires compactness, we keep the slightly stronger H-S assumption. The formula \eqref{eq:cond variance} highlights the standard identifiability issue in GARCH models: neither $\Sigma_k$ nor its defining parameters are uniquely determined by \eqref{eq:cond variance}, since unitary transformations of $\Sigma_k$ and $\mathscr{C}_{\boldsymbol{\varepsilon}}$ leave the conditional covariance unchanged. In multivariate GARCH models this is avoided by imposing $\mathscr{C}_{\boldsymbol{\varepsilon}} = \mathbb{I}$, identifying $\Sigma_k$ as the conditional covariance operator. In infinite-dimensional Hilbert spaces, however, the identity map $\mathbb{I}$ is not compact and this is not feasible. Instead, if $\mathscr{C}_{\boldsymbol{\varepsilon}}$ is any known injective covariance operator, then \eqref{eq:cond variance} uniquely identifies $\Sigma_k$. Hence, throughout we assume:

\begin{assumption}\label{as:C_eps known and injective} 
The covariance operator $\mathscr{C}_{\boldsymbol{\varepsilon}}$ is known and injective.    
\end{assumption}

\noindent When for example $\mathcal{H}=L^2[0,1]$, and 
$$
    \mathscr{C}_{\boldsymbol{\varepsilon}}(f)(t) = \int_0^1 C_\varepsilon(t,s)f(s)ds,
$$
for a covariance kernel $C_\varepsilon$, natural choices satisfying the above are Brownian motion errors, where $C_\varepsilon(t,s) = \sigma^2 \min\{t,s\}$, and Ornstein--Uhlenbeck errors, where $C_\varepsilon(t,s) = \sigma \exp(-\theta|t-s|).$ Although the conditional covariance operator does not coincide with $\Sigma_k$, it can be easily recovered from it using the known $\mathscr{C}_{\boldsymbol{\varepsilon}}$. 

In pw-f(G)ARCH models of higher order, Markovian state-space forms have been used to establish stationarity \citep[cf.][]{Ceroveckietal2019,Kuehnert2020}. Our op-ARCH$(p)$ process also admits a Markovian representation, though in a more intricate form:
\begin{gather}\label{quasi-state-space form}
    X^{\otimes2,[p]}_{\!k} = \Upsilon_{\!k}\big(\boldsymbol{\Delta} + \boldsymbol{\Psi}(X^{\otimes2,[p]}_{\!k-1})\big), \quad k\in\mathbb{Z}. 
\end{gather}
Here $X^{\otimes2,[p]}_{\!k}$ collects the past $p$ ``squared'' functions, so
\begin{align}\label{sk3}
    X^{\otimes2,[p]}_{\!k} \coloneqq \big(X^{\otimes2}_{k}, X^{\otimes2}_{k-1}, \dots, X^{\otimes2}_{k-p+1}\big)^\top\!,
\end{align}
with constant part $\boldsymbol{\Delta} \coloneqq (\Delta, 0, \dots, 0)^{\!\top}\!$, $\boldsymbol{\Psi}: \mathcal{S}^p_{\ge0}\to \mathcal{S}^p_{\ge0}$ refers to the operator-valued matrix
\begin{gather}\label{eq:matrix state-space form}
    \boldsymbol{\Psi} \coloneqq
    \begin{bmatrix}
        \alpha_1 & \cdots & \cdots & \cdots & \alpha_p\\
        \mathbb{I} & 0 & \cdots & \cdots & 0\\
        0 & \mathbb{I} & 0 & \cdots & 0\\
        \vdots & \ddots & \ddots & \ddots & \vdots\\
        0 & \cdots & 0 & \mathbb{I} & 0
    \end{bmatrix},
\end{gather}
and $\Upsilon_{\!k}:\mathcal{S}^p_{\ge0}\to\mathcal{S}^p_{\ge0}$ refers to the map defined by $\Upsilon_{\!k}(A) \coloneqq (A_1^{1/2}\varepsilon^{\otimes2}_{k}A_1^{1/2}, A_2, \dots, A_p)^\top$ for $A=(A_1,\dots,A_p)^\top\!.$ Although the representation \eqref{quasi-state-space form} is natural, establishing stationarity is challenging: the nonlinearity of $\Upsilon_k$ rules out the use of a top Lyapunov exponent \citep[cf.][]{Kingman1973,Liggett1985}, and the geometric moment contraction condition of \cite{WuShao2004}, as employed in \cite{Hoermannetal2013} for pw-fARCH models, does not appear applicable to the transformation $A\mapsto\Upsilon_k(\boldsymbol{\Delta}+\boldsymbol{\Psi}(A))$. Nevertheless, a strictly stationary solution can still be constructed algorithmically \citep[see][p.~289 for the $\mathrm{vech}$\,(G)ARCH case]{FrancqZakoian2019}. A more tractable model that is amenable to theoretical analysis will be introduced below.

Several of the results below, including crucially methods to estimate the operators in \eqref{eq:definition op-ARCH}, are simplified considerably by assuming the following: 

\begin{assumption}\label{as:Assumptions of (eps_k) and assoc. covop} 
$\mathscr{C}_{\boldsymbol{\varepsilon}}$ commutes with $\Sigma_k$ in \eqref{eq:definition op-ARCH} for all $k.$
\end{assumption}

\noindent Assumptions \ref{as:C_eps known and injective}--\ref{as:Assumptions of (eps_k) and assoc. covop} imply that $\Delta$ and the ranges of $\alpha_1,\dots,\alpha_p$ are diagonalizable with respect to the (known) eigenbasis of $\mathscr{C}_\varepsilon$. In finite-dimensional (G)ARCH models, this is typically ensured by taking $\mathscr{C}_\varepsilon = \mathbb{I}$. The choice of $\mathscr{C}_\varepsilon$ may thus be viewed as selecting the basis that diagonalizes the conditional covariance. In the results below, we explicitly state which arguments rely on this assumption. Moreover, to delve more deeply into the structure of the op-ARCH model, and provide comparisons to other multivariate GARCH models, we use the notation  
\begin{align}\label{sk1}
    \boldsymbol{\alpha}: \mathcal{S}^p_{\geq 0} \to \mathcal{S}_{\geq 0}, ~A=(A_1,\dots,A_p)^\top\mapsto \boldsymbol{\alpha}(A) = \sum^{p}_{i=1}\alpha_i(A_i),
\end{align}
and assume the following. 

\begin{assumption}\label{as: bold alpha_p H-S} 
The operator $\boldsymbol{\alpha}$ is H-S. 
\end{assumption}

\noindent Let $(e_j)_{j=1}^\infty$ be the eigenbasis associated to the covariance operator  $\mathscr{C}_{\boldsymbol{\varepsilon}}$. Assumption \ref{as: bold alpha_p H-S}, together with the fact that the space of H-S operators from $\mathcal{S}^p$ to $\mathcal{S}$ is a separable Hilbert space yields
\begin{align}\label{eq:expansion of H-S operators from S^p to S}
    \boldsymbol{\alpha} = \sum_{i=1}^p\sum_{j=1}^\infty\sum_{k=1}^\infty\sum_{\ell=1}^\infty\sum_{m=1}^\infty \,a_{ijk\ell m}\big[E_{ijk} \!\otimes_\mathcal{S}\! (e_\ell\otimes e_m)\big],
\end{align}
where $E_{ijk}$ is the vector placing $e_j \otimes e_k$ at position $i$ and zeros elsewhere. By the definitions of $\Sigma_k$ in the op-ARCH$(p)$ equation and $\boldsymbol{\alpha}$, and with $a_{ij\ell}\coloneqq a_{ijj\ell \ell},$ this leads to the simplified representation 
\begin{align}\label{eq:expansion of bold a_p for vech}
    \boldsymbol{\alpha} = \sum_{i=1}^p\sum_{j=1}^\infty\sum_{\ell=1}^\infty \,a_{ij\ell}\big[E_{ijj} \!\otimes_\mathcal{S}\! (e_\ell\otimes e_\ell)\big].
\end{align}
Each component $\alpha_i$ of $\boldsymbol{\alpha}$ maps $\mathcal{S}_{\geq 0}$ to $\mathcal{S}_{\geq 0}$, i.e.~non-negative definite, self-adjoint H-S operators into itself. Thus, for any $A = (A_1, \dots, A_p)^\top \in \mathcal{S}^p_{\geq 0}$, $\boldsymbol{\alpha}(A)$ is self-adjoint and non-negative definite under mild conditions. To establish the latter, note that we need
$$ 
    \big\langle \boldsymbol{\alpha}(A) (x), x \big\rangle \;=\; \sum^p_{i=1}\sum^\infty_{j=1}\sum^\infty_{\ell=1}\,a_{ij\ell}\langle A_i(e_j), e_j\rangle \langle x, e_\ell\rangle^2 \;\geq\; 0, \quad x \in \mathcal{H}.
$$
Since $\langle A_i(e_j), e_j\rangle \geq 0$ for all $i,j$, a sufficient condition for $\boldsymbol{\alpha}(A)$ to be non-negative definite is $a_{ij\ell} \geq 0$ for all $i,j,\ell.$

\section{CCC-op-ARCH model and structure}\label{Sec:CCC}

As indicated above, we now introduce a more parsimonious model. The expansion \eqref{eq:expansion of bold a_p for vech} is the most general form of the op-ARCH operators satisfying Assumption \ref{as:Assumptions of (eps_k) and assoc. covop}. In analogy to multivariate GARCH, this is akin to a ``VEC-ARCH'' specification \citep[see][Ch.~10.2.2]{FrancqZakoian2019}. Although such a model allows for a flexible serial dependence structure, it is often considered overparameterized, challenging to estimate, and theoretically difficult to analyze. A more tractable model is obtained by retaining only the diagonal terms in \eqref{eq:expansion of bold a_p for vech}. This section is devoted to the discussion of such a model, where Assumptions \ref{as:C_eps known and injective}--\ref{as: bold alpha_p H-S} hold. 

\begin{definition}\label{ccc-def} We say that a process $(X_k)_{k \in \mathbb{Z}}$ is a {\emph{Constant Conditional Correlation operator-level ARCH$(p)$}} (CCC-op-ARCH$(p)$) process if \eqref{eq:definition op-ARCH} holds with 
\begin{align}\label{eq:diagARCH}         
    \boldsymbol{\alpha} = \sum^p_{i=1}\sum^\infty_{\ell=1}\, a_{i\ell\ell}\big[E_{i\ell\ell}\!\otimes_\mathcal{S}\!(e_\ell\otimes e_\ell)\big].  
\end{align}
\end{definition}

\noindent This is called a ``CCC'' model since it assumes that  $\boldsymbol{\alpha}$ only depends on the diagonal-terms in its expansion, while the components of the process $X_i$ remain ``conditionally correlated'' through their common dependence on $\mathscr{C}_\varepsilon$. We note that one might also consider a CCC model of the form \eqref{eq:definition op-ARCH} where for a compact covariance operator $\mathscr{G}$, $\Sigma_k = H_k^{1/2} \mathscr{G} H_k^{1/2}$, and $H_k$ satisfies the recursion on the right-hand side of \eqref{eq:definition op-ARCH}. Under the commutativity Assumption \ref{as:Assumptions of (eps_k) and assoc. covop} and with $\boldsymbol{\alpha}$ following \eqref{eq:diagARCH}, this reduces to the given CCC model.  

Throughout this section, $(X_k)$ is assumed to be a CCC-op-ARCH$(p)$ model for some $p\in\mathbb{N}.$ One of the benefits of this model is that there exists a linear Markovian form associated with it. To be precise, we have (see also Example \ref{Ex:Markovianforms})
\begin{align}\label{sk2}
    X^{\otimes2,[p]}_{\!k,\mathrm{d}} = \boldsymbol{\Delta}_k + \boldsymbol{\Psi}_{\!k}\big(X^{\otimes2,[p]}_{\!k-1,\mathrm{d}}\big), \quad k\in\mathbb{Z},
\end{align}
where $X^{\otimes2,[p]}_{\!k,\mathrm{d}}$ is the diagonal part of $X^{\otimes2,[p]}_{\!k}$ in \eqref{sk3}, which is for each component defined by
\begin{align}\label{sk7}
    X^{\otimes2}_{\!k-i+1,\mathrm{d}} \coloneqq  \sum^\infty_{\ell=1}\,\langle X_{k-i+1},e_\ell\rangle^2(e_\ell\otimes e_\ell), \quad i \in \{1,...,p\},\;~k\in\mathbb{Z}.
\end{align}
Here $\boldsymbol{\Delta}_k \coloneqq (\psi_k(\Delta), 0, \dots, 0)^\top\!,$ and $\boldsymbol{\Psi}_{\!k}$ has the same form as $\boldsymbol{\Psi}$ in \eqref{eq:matrix state-space form} with each $\alpha_i$ replaced by $\psi_k\circ\alpha_i,$ where $\circ$ refers to composition, and 
\begin{align}\label{sk8}
    \psi_k(B)\coloneqq \sum^\infty_{j=1}\,\langle \varepsilon_k, e_j\rangle^2b_j(e_j\otimes e_j), \quad \mbox{ for } B = \sum^\infty_{j=1}b_j(e_j\otimes e_j)\in \mathcal{S}_{\geq 0}.
\end{align}

It should be noted that a host of other potential, simplified models starting from \eqref{eq:expansion of bold a_p for vech} might be considered. For ease of presentation and due to the empirical performance of the CCC-op-ARCH$(p)$ model in our  analyses below, we have chosen to focus on this case. We discuss this and avenues for future research in Section \ref{Section Discussion}. 

\subsection{Strict stationarity}
To deduce conditions under which this model admits a stationary solution, we introduce a \emph{top Lyapunov exponent} \citep[cf.][]{Kingman1973,Liggett1985}. To this end, we define
\begin{align}\label{eq:tau}
    \tau: \mathcal{B}_{\mathcal{S}^p_{\geq 0,\mathrm{d}}} \to [0,\infty), \quad \tau(B)\coloneqq \sup_{\|A\|_{\mathcal{S}^p}\le 1,\,A\in\mathcal{S}^p_{\ge0,\mathrm{d}}}\|B(A)\|_{\mathcal{S}},
\end{align}
where $\mathcal{S}_{\ge0,\mathrm{d}}\subset \mathcal{S}_{\ge0}$ denotes the subset of non-negative, self-adjoint H-S operators with the diagonal form $A=\sum_j a_j(e_j\otimes e_j)$, and $\mathcal{B}_{\mathcal{S}^p_{\ge0,\mathrm{d}}}$ the set of bounded linear operators on $\mathcal{S}_{\ge0,\mathrm{d}}.$ Although the functional $\tau$ does not define a norm, it satisfies several useful properties for our analysis. It is dominated by the operator norm, $\tau(B) \le \|B\|_{\mathcal{L}}$, compatible with the H-S norm so that  $\|B(A)\|_{\mathcal{S}} \le \tau(B)\|A\|_{\mathcal{S}},$ and sub-multiplicative, $\tau(B_1 \circ B_2) \le \tau(B_1)\tau(B_2)$.

\begin{proposition}\label{prop:existence gamma}
The \emph{top Lyapunov exponent} defined by
\begin{align}\label{Definition top Lyapunov exponent}
    \gamma \coloneqq \lim_{k\to\infty}\frac{1}{k}\ln\tau\big(\boldsymbol{\Psi}_{\!k}\circ\boldsymbol{\Psi}_{\!k-1}\circ\cdots\circ\boldsymbol{\Psi}_{\!1}\big) \;=\; \lim_{k\to\infty}\frac{1}{k}\Exp\ln\tau\big(\boldsymbol{\Psi}_{\!k}\circ\boldsymbol{\Psi}_{\!k-1}\circ\cdots\circ\boldsymbol{\Psi}_{\!1}\big), 
\end{align}
exists with $\gamma\in[-\infty,\infty)$, where the first limit holds almost surely. Moreover,
\begin{align}\label{sk6}
    \gamma = \inf_{k\in\mathbb{N}}\frac{1}{k}\Exp\ln\tau\big(\boldsymbol{\Psi}_{\!k}\circ\boldsymbol{\Psi}_{\!k-1}\circ\cdots\circ\boldsymbol{\Psi}_{\!1}\big).
\end{align}
\end{proposition}

\begin{theorem}\label{theo:strict_stat}
If 
\begin{align}\label{sk5}
    \gamma<0,
\end{align}
the CCC-op-ARCH$(p)$ process $(X_k)$ admits a strictly stationary, causal, and almost surely unique solution. The same holds for the associated process $(\Sigma_k)$.
\end{theorem}

The condition \eqref{sk5} is difficult to verify in practice, even with known, relatively simple, operators $\alpha_1,...,\alpha_p$. The following two sufficient conditions are more tractable.

\begin{proposition}\label{prop: explciti suff conditions for stat sol} Condition \eqref{sk5} is satisfied if, for some $n\in\mathbb{N}$ and $\nu > 0$,
\begin{align}\label{eq:exp opGARCH smaller 1}
	\Exp\tau^\nu\big(\boldsymbol{\Psi}_{\!n}\circ\boldsymbol{\Psi}_{\!n-1}\circ \,\cdots\, \circ\boldsymbol{\Psi}_{\!1}\big) <\, 1. \allowdisplaybreaks
\end{align}
More explicitly, \eqref{eq:exp opGARCH smaller 1} is satisfied with $n=p$ and $\nu=1$ if 
\begin{gather}     
    \|\boldsymbol{\alpha}\|_\mathcal{L}\Exp\!\|\varepsilon_0\|^2\sum^{p}_{\ell=1}\,\ell\big(\|\boldsymbol{\alpha}\|_\mathcal{L}\Exp\!\|\varepsilon_0\|^2\big)^{p-\ell} < 1\label{easier suff cond exp smaller 1}.\allowdisplaybreaks
\end{gather}
\end{proposition}

\begin{remark} {\rm 
\mbox{}\\[-2.5ex]
\begin{itemize}
    \item[\textnormal{(a)}] In contrast to real-valued (G)ARCH processes \citep[cf.][Theorem 2.4]{FrancqZakoian2019}, the condition \eqref{sk5} is not necessary, since norms on infinite-dimensional spaces are not equivalent \citep[cf.][Remark 1]{Ceroveckietal2019}. 

    \item[\textnormal{(b)}] The sufficient condition \eqref{easier suff cond exp smaller 1} is convenient, but relatively crude. A sharper analysis of this constraint may further relax the requirement, which we do not pursue here. Further, as each $\boldsymbol{\Psi}_{\!k}$ contains identity maps---and hence has norm at least $1$---\eqref{easier suff cond exp smaller 1} can, by the   definition of $\boldsymbol{\Psi}_{\!k}$, only hold for compositions $\boldsymbol{\Psi}_{\!n}\circ \boldsymbol{\Psi}_{\!n-1}\circ\cdots\circ \boldsymbol{\Psi}_{\!1}$ with $n\ge p$ (see the proof in Section \ref{Section proofs}). For simplicity, we focus on the case $n=p$.
\end{itemize}}
\end{remark}

\begin{example} Suppose $p=1$. Then considering \eqref{eq:exp opGARCH smaller 1} with $n=1$ leads to the sufficient stationarity condition 
$$ 
    \Exp\left[\,\sup_{j \ge 1} a_{1jj}^\nu\langle \varepsilon_1, e_j \rangle^{2\nu} \right]<1,
$$  
for some $\nu >0$. Moreover, one may see from this that 
$$
    \Exp\tau^\nu\big( \boldsymbol{\Psi}_{\!1}\big) \le \| {\alpha}_{1} \|_{\mathcal{L}}^\nu \Exp\left[\,\sup_{j \ge 1} \langle \varepsilon_1, e_j \rangle^{2\nu} \right].
$$
With $\nu =1$, this leads to the sufficient condition $\| {\alpha}_{1} \|_{\mathcal{L}} \|\mathscr{C}_\varepsilon\|_{\mathcal{L}}<1$.
\end{example}

\subsection{Existence of moments, weak dependence, and weak stationarity }\label{subsec:momentsWeakdependence}
 
Under suitable conditions, CCC-op-ARCH processes possess finite moments and display weak serial dependence. In particular, they are $L^p$-$m$-approximable, a weak dependence notion for functional data initially introduced in \cite{HoermannKokoszka2010}. A process $(Y_{k})_{k\in\mathbb{Z}}\subset\mathcal{H}$ is \emph{$L^p$-$m$-approximable} if (a) it admits a Bernoulli shift representation, i.e.~$Y_{k} = f(\varepsilon_{k}, \varepsilon_{k-1}, \ldots)$ for some i.i.d.~sequence $(\varepsilon_{k})_{k\in\mathbb{Z}}$ taking values in a measurable space $\mathbb{S}$, and a measurable mapping $f: \mathbb{S}^{\mathbb{N}} \to \mathcal{H},$ and (b) with $Y^{(m)}_{k}\!\coloneqq f(\varepsilon_{k}, \ldots, \varepsilon_{k-m+1}, \varepsilon'_{k-m}, \varepsilon'_{k-m-1}, \ldots)$, where $(\varepsilon'_k)_{k \in \mathbb{Z}}$ is an independent sequence of copies of $\varepsilon_{0},$ and with $\xi_{Y,p}(m) = (\Exp\!\|Y_{0}\!- Y^{(m)}_{0}\|_\mathcal{H}^p)^{1/p}$, 
\begin{align}\label{lp-m-cond}    
    \sum^\infty_{m=1}  \xi_{Y,p}(m) < \infty. 
\end{align}

\noindent According to (a), $L^p$-m-approximable processes are strictly stationary and ergodic. If \eqref{lp-m-cond} holds with $p \ge 2$, then the process $Y_k$ satisfies for example the central limit theorem. This condition is satisfied by many standard models, including functional AR and linear processes \citep[as illustrated in][]{HoermannKokoszka2010}, and also pw-(G)ARCH models. The following result provides sufficient conditions for this property and for the existence of finite moments. In this result, the notation $\|\cdot\|_4$ appears, referring to the norm of \emph{Schatten class operators of order~4} (see Section~\ref{Appendix: Section preliminaries}).
  
\begin{proposition}\label{prop:Explicit moments} 
Let \eqref{eq:exp opGARCH smaller 1} and $(\varepsilon_k) \subset L^{2\nu}_{\mathcal{H}}$ for the same $\nu$ hold. Then:
\begin{itemize}
    \item[\textnormal{(a)}] $\Exp\!\|X^{\otimes2}_{0}\|^{\nu}_{\mathcal{S}} = \Exp\!\|X_{0}\|^{2\nu} < \infty$ and $\Exp\!\|\Sigma^{1/2}_{0}\|^{2\nu}_4 = \Exp\!\|\Sigma_{0}\|^{\nu}_{\mathcal{S}} < \infty.$
    \item[\textnormal{(b)}] The processes $X_k$ and $\Sigma^{1/2}_k$ are  $L^{2\nu}$-$m$-, and $Y_k=X^{\otimes2}_{k}$ and $Y_k= \Sigma_{k}$ are $L^{\nu}$-$m$-approximable, each with geometrically decaying approximation errors, i.e.~$\xi_{X, 2\nu}(m) \leq c_\dagger \rho^m$ and $\xi_{Y, \nu}(m)\le c_\dagger \rho^m$ for some $c_\dagger>0$ and $\rho\in(0,1)$. 
\end{itemize}
\end{proposition}

\begin{proposition}\label{prop:weak stationarity}
Let $(X_k)$ be a weakly stationary CCC-op-ARCH$(p)$ process. Then, $(X_k)$ is a weak white noise. Further, 
$\mu_{\boldsymbol{\Sigma}} \coloneqq \Exp(\Sigma_0) = \Exp(\Sigma_k)$ for all $k,$ with
\begin{gather}\label{eq:Mean vol process ARCH}
    \mu_{\boldsymbol{\Sigma}} = \Delta + \sum^p_{i=1}\alpha_i(\mathscr{C}_{\!\boldsymbol{X}}).
\end{gather}
In addition, if Assumption \ref{as:Assumptions of (eps_k) and assoc. covop} and also  
\begin{align}\label{eq:product op and innovation cov op norm}
    \left\|\,\sum^p_{i=1}\alpha_i(\mathscr{C}_{\boldsymbol{\varepsilon}})\,\right\|_\mathcal{L} < 1
\end{align}
are satisfied, then $\mu_{\boldsymbol{\Sigma}}$ has the more explicit representation 
\begin{align}\label{Garch positive cond - inverse}
    \mu_{\boldsymbol{\Sigma}} = \bigg(\boldsymbol{\mathbb{I}} - \sum^p_{i=1}\,\alpha_i(\mathscr{C}_{\boldsymbol{\varepsilon}})\bigg)^{\!-1}(\Delta).
\end{align}
\end{proposition}

\begin{remark}
\mbox{}\\[-2.5ex]
\begin{itemize}
    \item[\textnormal{(a)}] Eq.~\eqref{Garch positive cond - inverse} parallels the classical formula for the unconditional variance of univariate (G)ARCH models \citep[cf.][]{FrancqZakoian2019}. Moreover, \eqref{eq:product op and innovation cov op norm} may still hold even when $\sum_{i=1}^p\|\alpha_i\|_{\mathcal{L}}\ge 1$, since $\|\mathscr{C}_{\boldsymbol{\varepsilon}}\|_{\mathcal{L}}<1$ for most relevant innovation processes (e.g., the standard Brownian motion).
        
    \item[\textnormal{(b)}] The moment condition in Proposition~\ref{prop:Explicit moments} for $\nu\le 1$ is stated only for completeness and is redundant, since finite second moments of $\varepsilon_0$ were already assumed.
\end{itemize}
\end{remark}

\subsection{Examples}

We illustrate the above results with two classes of op-ARCH processes on $\mathcal{H}=L^2[0,1]$, where $(\varepsilon_k)$ denotes an i.i.d.~sequence of standard Brownian motions on $[0,1]$.

\begin{example}\label{ex:1}
Let $p\in\mathbb{N}$. Consider the integral operator $\Delta$ with kernel $\min(t,s)$, $s,t\in[0,1]$, and define $\alpha_i(A)=a_i A$ for $A\in\mathcal{S}$ and scalars $a_i\ge 0$ for  $1\le i < p,$ and $a_p\neq 0.$ Then $\Delta\in\mathcal{S}_{>0}$ and each $\alpha_i\in\mathcal{L}_{\mathcal S}$ maps $\mathcal{S}_{\ge 0}$ to $\mathcal{S}_{\ge 0}$.
\begin{itemize}
    \item[\textnormal{(a)}] Let $p=1.$ Since $\mathbb{E}\|\varepsilon_0\|^2=\|\mathscr{C}_{\boldsymbol{\varepsilon}}\|_{\mathcal N}=1/2$ and $\|\boldsymbol{\alpha}\|_{\mathcal L}=a_1$, condition~\eqref{easier suff cond exp smaller 1} holds whenever $a_1<2$. In this case, Proposition~\ref{prop:Explicit moments} yields $\mathbb{E}\|X_0\|^2<\infty$ and $\mathbb{E}\|\Sigma_0\|_{\mathcal S}<\infty$, so $(X_k)$ is weakly stationary, a WWN by     Proposition~\ref{prop:weak stationarity}, and $L^2$-$m$-approximable.

    \item[\textnormal{(b)}] Let $p=3.$ Here, $\|\boldsymbol{\alpha}\|_{\mathcal L}\le a\coloneqq a_1+a_2+a_3$. Hence, \eqref{easier suff cond exp smaller 1} is satisfied if
    \[
        a(a^2+4a+12) < 8,
    \]
    which holds roughly for $a \leq 0.551$.
\end{itemize}
\end{example}

\begin{example}\label{ex:2}
In the following, let
\[
    \Delta=\sum_{\ell=1}^{\infty} d_\ell(e_\ell\!\otimes\!e_\ell), 
    \quad  
    \alpha_i=\sum_{\ell=1}^{\infty} a_{i\ell\ell}\,(e_\ell\!\otimes\!e_\ell)\!\otimes_{\mathcal S}\!(e_\ell\!\otimes\!e_\ell),
    \quad 1\le i\le p,
\]
where $(d_\ell)$ and $(a_{i\ell\ell})_\ell$ are positive, strictly decreasing, square-summable sequences. Then
\[
    \Sigma_k=\sum_{\ell=1}^{\infty} Z_{k,\ell}(e_\ell\!\otimes\!e_\ell), 
    \quad 
    X_k=\sum_{\ell=1}^{\infty} Z^{1/2}_{k,\ell}\,\langle\varepsilon_k,e_\ell\rangle e_\ell,
\]
with
\[
    Z_{k,\ell} = d_\ell +\sum_{i=1}^{p} a_{i\ell\ell}\,Z_{k-i,\ell}\,\langle\varepsilon_{k-i},e_\ell\rangle^2.
\]
To highlight structural aspects, let $p=1$ and assume $a_{1\ell\ell}=a\ell^{-2}$ for some $a>0$. Then
\[
    \|\boldsymbol{\alpha}\|_{\mathcal L}\,\mathbb{E}\|\varepsilon_0\|^2
    \;=\; \sup_{\ell \ge 1} |a_{1\ell\ell}|/2 \;=\; a/2,
\]
so, by Example~\ref{ex:1}, $(X_k)$ is strictly stationary, $L^2$-$m$-approximable, and a WWN when  $a<2$. Further in this case Assumption~\ref{as:Assumptions of (eps_k) and assoc. covop} and \eqref{eq:product op and innovation cov op norm} hold. Thus, by Proposition~\ref{prop:weak stationarity},
\[
    \mu_{\boldsymbol{\Sigma}}
    = \big(\boldsymbol{\mathbb{I}}-\alpha_1(\mathscr{C}_{\boldsymbol{\varepsilon}})\big)^{-1}(\Delta).
\]
\end{example}

\section{Estimation in the CCC-op-ARCH model}\label{Sec:estimation}
We now turn to the estimation of the parameters  $\alpha_1,...,\alpha_p,$ and $\Delta$  in the CCC-op-ARCH$(p)$. Henceforth, $\boldsymbol{X} = (X_k) \subset \mathcal{H}$ refers to a stationary CCC-op-ARCH$(p)$ process that possesses finite second moments from which we have observed a stretch of length $N$, $X_1,...,X_N$. We will assume throughout the remainder of this article that Assumption \ref{as:Assumptions of (eps_k) and assoc. covop} holds so that $\mathscr{C}_\varepsilon$ and $\Sigma_k$ commute. Under this assumption, 
$$
    \mu_{\boldsymbol{\Sigma}} = \Exp(\Sigma_k), \quad \mathscr{C}_{\!\boldsymbol{X}} = \Exp(X_k^{\otimes2}) = \mu_{\boldsymbol{\Sigma}} \mathscr{C}_{\boldsymbol{\varepsilon}}, \quad \Sigma_k^{1/2} \mathscr{C}_{\boldsymbol{\varepsilon}} \Sigma_k^{1/2}= \Sigma_k \mathscr{C}_{\boldsymbol{\varepsilon}}, \quad k \in \mathbb{Z},
$$
and the op-ARCH equations imply that
\begin{align}\label{yw-base}
    X^{\otimes2}_{k} - \mathscr{C}_{\!\boldsymbol{X}} 
    = \Sigma^{1/2}_{k}(\varepsilon^{\otimes 2}_{\!k} - \mathscr{C}_{\boldsymbol{\varepsilon}})\Sigma^{1/2}_{k} + \bigg(\sum^p_{i=1}\alpha_i(X^{\otimes2}_{k-i} - \mathscr{C}_{\!\boldsymbol{X}})\bigg)\mathscr{C}_{\boldsymbol{\varepsilon}}.
\end{align}

If we calculate the covariances of the above with $X^{\otimes2}_{k-i} $, $1\leq i \leq p$, due to causality the covariance with the first term on the right-hand side of \eqref{yw-base} will vanish, and the covariances with the second term may be expressed in terms of the $\alpha$ operators. In order to shift the nuisance operator $\mathscr{C}_\varepsilon$ off of the terms containing the $\alpha$'s,  we proceed first by applying on the right a \emph{Tikhonov}-regularized inverse 
\begin{align}\label{ceps-tik}
    \mathscr{C}^\dagger_{\!\boldsymbol{\varepsilon}} \coloneqq (\mathscr{C}_{\boldsymbol{\varepsilon}} + \vartheta_{\!N} \mathbb{I})^{-1}, 
\end{align}
where $\mathbb{I}$ denotes the identity map and $\vartheta_{\!N} > 0$ is an asymptotically vanishing regularization parameter, i.e., $\vartheta_{\!N} \to 0$. After this regularization, we also project onto a finite-dimensional space. Let $(a_j, e_j)$ be the eigenpairs of $\mathscr{C}_{\boldsymbol{\varepsilon}},$ i.e.~$a_1\geq a_2 \geq \cdots >0$ are the eigenvalues associated with the eigenfunctions $e_1, e_2, \dots$ of $\mathscr{C}_{\boldsymbol{\varepsilon}}.$ Then, the projection operator onto the linear space spanned by $e_1, \dots, e_K,$ $K \in \mathbb{N}$, is denoted by $\coprod^{e_K}_{e_1},$ and we define $\mathscr{C}^\ddagger_{\!\boldsymbol{\varepsilon}} \coloneqq \mathscr{C}_{\boldsymbol{\varepsilon}} \mathscr{C}^\dagger_{\!\boldsymbol{\varepsilon}}.$ According again to \eqref{yw-base}, 
\begin{align}
    \big(X^{\otimes2}_{k} - \mathscr{C}_{\!\boldsymbol{X}}\big)\mathscr{C}^\dagger_{\!\boldsymbol{\varepsilon}}\!\coprod^{e_K}_{e_1}
    &\,=\, \Sigma^{1/2}_{k}(\varepsilon^{\otimes 2}_{\!k} - \mathscr{C}_{\boldsymbol{\varepsilon}})\Sigma^{1/2}_{k}\mathscr{C}^\dagger_{\!\boldsymbol{\varepsilon}}\!\coprod^{e_K}_{e_1}\;  + \;\bigg(\sum^p_{i=1}\alpha_i(X^{\otimes2}_{k-i} - \mathscr{C}_{\!\boldsymbol{X}})\bigg)\mathscr{C}^\ddagger_{\!\boldsymbol{\varepsilon}}\!\coprod^{e_K}_{e_1}.\label{eq:quasi-YW}
\end{align}
Further, let $\boldsymbol{X}^{\otimes2}_{\!\boldsymbol{\varepsilon}} =(X^{\otimes2}_{\!k,\boldsymbol{\varepsilon}})_{k\in\mathbb{Z}}\subset\mathcal{S}$ and $\boldsymbol{X}^{\otimes2,[p]}\! = (X^{\otimes2,[p]}_{\!k})_{k\in\mathbb{Z}}\subset\mathcal{S}^{p}$ be the processes defined by
\begin{gather}
    X^{\otimes2}_{\!k,\boldsymbol{\varepsilon}}\coloneqq X^{\otimes2}_{\!k}\mathscr{C}^\dagger_{\!\boldsymbol{\varepsilon}}\!\coprod^{e_K}_{e_1}\,, \; \mbox{ and }  \;  X^{\otimes2,[p]}_{\!k}\coloneqq (X^{\otimes2}_{\!k}, \dots, X^{\otimes2}_{\!k-p+1})^\top\!, \quad k\in\mathbb{Z},
\end{gather}
and their centered versions by
\begin{gather}
    \tilde{X}^{\otimes2}_{\!k,\boldsymbol{\varepsilon}}\coloneqq X^{\otimes2}_{\!k,\boldsymbol{\varepsilon}} - \,\mathscr{C}_{\!\boldsymbol{X}}\mathscr{C}^\dagger_{\!\boldsymbol{\varepsilon}}\!\coprod^{e_K}_{e_1}\,, \; \mbox{ and }  \;  \tilde{X}^{\otimes2,[p]}_{\!k}\coloneqq X^{\otimes2,[p]}_{\!k} - (\mathscr{C}_{\!\boldsymbol{X}}, \dots, \mathscr{C}_{\!\boldsymbol{X}})^\top\!, \quad k\in\mathbb{Z}.
\end{gather}
We see according to \eqref{eq:quasi-YW} that the lag-1 cross-covariance operators $\mathscr{D} = \mathscr{D}_{K,N} = \mathscr{C}^1_{\!\boldsymbol{X}^{\otimes2,[p]}, \boldsymbol{X}^{\otimes2}_{\!\boldsymbol{\varepsilon}}} \in \mathcal{N}_{\mathcal{S}^p\!,\mathcal{S}}$ satisfy the following Yule–Walker (YW)-type equation:
\begin{align}\label{eq:YW initial}
    \mathscr{D} = \boldsymbol{\alpha}\mathscr{C} + \mathscr{R},
\end{align}
where
\begin{align}
    \mathscr{C} = \mathscr{C}_{\!\boldsymbol{X}^{\otimes2,[p]}} \in \mathcal{N}_{\mathcal{S}^p}, 
\end{align}
with $\boldsymbol{\alpha}$ from \eqref{sk1}, and where the remainder $\mathscr{R} = \mathscr{R}_{K,N}$ is defined by
\begin{align}\label{eq:Def remainders R_TNp}
    \mathscr{R}\coloneqq\Exp\!\big\langle \tilde{X}^{\otimes2,[p]}_0, \cdot\big\rangle_{\!\mathcal{S}}\Bigg[\,\Big[\boldsymbol{\alpha}\big(\tilde{X}^{\otimes2,[p]}_0\big)\Big]\bigg(\mathscr{C}^\ddagger_{\!\boldsymbol{\varepsilon}}\!\coprod^{e_K}_{e_1}-\,\mathbb{I}\bigg)\Bigg]\,.
\end{align}
Notice that all the operators in the YW-type equation \eqref{eq:YW initial} are well-defined due to causality of the involved processes and Proposition \ref{prop:Explicit moments}. 

If the remainder term is small, it seems natural in view of \eqref{eq:YW initial} to estimate $\boldsymbol{\alpha}$ by $\hat{\mathscr{D}}\hat{\mathscr{C}}^{-1}$. This, however, does not yield a useful estimate, as $\boldsymbol{\alpha}$ is not identifiable from $\boldsymbol{\alpha}\mathscr{C}$ when $\mathscr{C}$ is not injective. For example, in $\mathcal{H} = L^2[0,1]$ with $p=1$, the operator $J \in \mathcal{S}$ with kernel $j(s,t) = -1$ if $s \leq t$ and $j(s,t) = 1$ otherwise, satisfies $\langle \mathscr{C}(K), K \rangle_{\mathcal{S}} = 0$. This issue parallels the multivariate case where, for $X \in \mathbb{R}^d$, $\mbox{vec}(XX^\top) \in \mathbb{R}^{d^2}$ does not have a full-rank covariance matrix, while the ``half-vectorization'' $\mbox{vech}(XX^\top) \in \mathbb{R}^{d(d+1)/2}$ typically does. In what follows, we derive  estimators for the CCC-op-ARCH$(p)$ operators, that is, to reiterate, $\boldsymbol{\alpha}$ in \eqref{eq:expansion of bold a_p for vech} satisfies 
\begin{align}\label{eq:diagARCH-2}         
    \boldsymbol{\alpha} = \sum^p_{i=1}\sum^\infty_{\ell=1}\, a_{i\ell\ell}\big[E_{i\ell\ell}\!\otimes_\mathcal{S}\!(e_\ell\otimes e_\ell)\big]\,,  
\end{align}
based on the YW-type equation \eqref{eq:YW initial}, suitably modified to ensure identifiability, as well as estimators for the intercept term $\Delta$. 

\subsection{Finite-dimensional setting}\label{Sec: Estimation ARCH ops}

Although our ultimate goal is to derive consistent estimators for the infinite-dimensional operators in \eqref{eq:diagARCH-2}, we begin with estimating the CCC-op-ARCH operators $\alpha_1, \dots, \alpha_p,$ under the simplifying assumption that 
\begin{align}\label{eq:representation alpha_p, diag, fin}         
    \boldsymbol{\alpha} = \boldsymbol{\alpha}_K \coloneqq\sum^p_{i=1}\sum^K_{\ell=1}\, a_{i\ell\ell}[E_{i\ell\ell}\!\otimes_\mathcal{S}\!(e_\ell\otimes e_\ell)]\,   
\end{align}
for a finite integer $K$. Further, we assume that the observed curves $X_i$ are finite-dimensional, so that  with $X_{k,i} = \langle X_k, e_i\rangle:$
\begin{gather}
    X_k = \sum^K_{i=1}\,X_{k,i}e_i, \quad k\in\mathbb{Z}.\label{eq:finite-dim. data} 
\end{gather}  
The representation of $X_k$ yields with $X_{k,ij}\coloneqq X_{k,i}X_{k,j}:$
\[
    X^{\otimes 2}_k = \sum^K_{i=1} \sum^K_{j=1}\,X_{k,ij}(e_i\otimes e_j),
\]
implying that each $X^{\otimes 2, [p]}_k,$ with $p\in\mathbb{N},$ is characterized by the block matrix
\begin{align*}
    \Big(\big(X_{k,ij}\big)_{i,j=1}^K, \big(X_{k-1,ij}\big)_{i,j=1}^K, \dots, \big(X_{k-p+1,ij}\big)_{i,j=1}^K \Big)^\top \in \mathbb{R}^{pK \times K}.    
\end{align*}
 
Throughout, the map $\diag : \mathbb{R}^K \to \mathbb{R}^{K \times K}$ and its adjoint $\diag^\ast : \mathbb{R}^{K \times K} \to \mathbb{R}^K$ construct a diagonal matrix from a vector and create a vector by extracting the diagonal of the input matrix, respectively. Note that $\diag : \mathbb{R}^{pK} \to \mathbb{R}^{pK \times K}$ and $\diag^\ast : \mathbb{R}^{pK \times K} \to \mathbb{R}^{pK}$ are also component-wise defined, i.e. for any vectors $x_1, \dots, x_p\in\mathbb{R}^{K}$ and matrices $A_1, \dots, A_p\in\mathbb{R}^{K\times K},$
\begin{align*}
    \diag(x_1, \dots, x_p) &\coloneqq \big(\diag(x_1), \dots, \diag(x_p)\big)^{\!\top}\!, \mbox{ and }\\[1ex]
    \diag^\ast(A_1, \dots, A_p) &\coloneqq \big(\diag^\ast(A_1), \dots, \diag^\ast(A_p)\big)^{\!\top}\!.
\end{align*}
From \eqref{eq:representation alpha_p, diag, fin} and \eqref{eq:finite-dim. data}, it follows with $\tilde{X}_{m,ij} \coloneqq X_{m,ij} - \Exp(X_{m,ij})$ for any $m$:
\begin{align*}
    &\diag^\ast\!\boldsymbol{\alpha}\mathscr{C}\!\diag(x_1, \dots, x_p)\\
    &\quad= \sum^p_{i=1}\sum^p_{j=1}\sum^K_{k=1}\,x^{(k)}_i\diag^\ast\Exp\!\Big(\tilde{X}_{1-i,kk}\alpha_j\big(\tilde{X}^{\otimes 2}_{1-j}\big)\Big)\allowdisplaybreaks\\
    &\quad = \sum^p_{i=1}\sum^p_{j=1}\sum^K_{k=1}\,x^{(k)}_i\Exp\!\Big(\tilde{X}_{1-i,kk}\big(a_{j11}\tilde{X}_{1-j,11}, a_{j22}\tilde{X}_{1-j,22}, \dots, a_{jKK}\tilde{X}_{1-j,KK} \big)^{\!\top}\Big)\\
    &\quad = \boldsymbol{\alpha}_{\mathrm{d}}\mathscr{C}_{\mathrm{d}}(x_1, \dots, x_p),
\end{align*}
with $\mathscr{C}_{\mathrm{d}} \coloneqq \diag^\ast\mathscr{C}\diag \in \mathbb{R}^{pK\times pK},$ and where $\boldsymbol{\alpha}_{\mathrm{d}} \in \mathbb{R}^{K\times pK}$ is the block matrix 
$$
    \boldsymbol{\alpha}_{\mathrm{d}} \coloneqq \big(\diag(a_{111}, a_{122}, \dots, a_{1KK})\; \cdots \;\diag(a_{p11}, a_{p22}, \dots, a_{pKK})\big).
$$
Therefore, due to the identity \eqref{eq:YW initial}, it holds that
\begin{align}\label{eq:YW of diag alpha_p} 
\mathscr{D}_{\mathrm{d}} = \boldsymbol{\alpha}_{\mathrm{d}}\mathscr{C}_{\mathrm{d}} + \mathscr{R}_{\mathrm{d}},
\end{align}
where $\mathscr{D}_{\mathrm{d}}\coloneqq \diag^\ast\!\mathscr{D}\diag \in \mathbb{R}^{pK\times K}$ and $\mathscr{R}_{\mathrm{d}}\coloneqq \diag^\ast\mathscr{R}\diag \in \mathbb{R}^{pK \times K},$ with $\mathscr{D}, \mathscr{R}$ from \eqref{eq:YW initial}. Moreover, as all the coefficients of interest $a_{ijj},$ with $1\leq i \leq p,$ $1\leq j \leq K,$ are contained in 
$$
    \tilde{\boldsymbol{\alpha}}_{\mathrm{d}} \coloneqq \diag^\ast(\boldsymbol{\alpha}_{\mathrm{d}}),
$$ 
we propose the estimator
\begin{align*} 
    \hat{\boldsymbol{\alpha}}_{\mathrm{d}} \coloneqq \diag^\ast\!\big(\hat{\mathscr{D}}_{\mathrm{d}}\hat{\mathscr{C}}^{-1}_{\mathrm{d}}\big),
\end{align*}
provided $\hat{\mathscr{C}}_{\mathrm{d}}\coloneqq\diag^\ast\hat{\mathscr{C}}\diag$ is non-singular, and $\hat{\mathscr{D}}_{\mathrm{d}}\coloneqq \diag^\ast\!\hat{\mathscr{D}}\diag$, with 
\begin{gather}
    \hat{\mathscr{C}}\coloneqq \frac{1}{N}\sum_{k=p}^{N}\,\big(X^{\otimes2, [p]}_{\!k}\! - \hat{m}_p\big)\otimes_\mathcal{S}\big(X^{\otimes2, [p]}_{\!k} - \hat{m}_p\big), \mbox{ and }\label{eq:est lag-0-cov-op}\\[1ex]
  \hat{\mathscr{D}}\coloneqq \frac{1}{N}\sum_{k=p}^{N-1}\,\big(X^{\otimes2, [p]}_{\!k} - \hat{m}_p\big)\otimes_\mathcal{S}\bigg[\big(X^{\otimes2}_{\!k+1} - \hat{m}'_1\big)\mathscr{C}^\dagger_{\!\boldsymbol{\varepsilon}}\!\coprod^{e_K}_{e_1}\bigg].\label{eq:est lag-1-cross-cov-op}
\end{gather}

\begin{assumption}\label{as: C_d and hat C_d finite injective} 
\mbox{}\\[-2.5ex]
\begin{itemize}
    \item[\textnormal{(a)}] For all $N$ sufficiently large, with probability one  $\hat{\mathscr{C}}_{\mathrm{d}} \in \mathbb{R}^{K\times K}\!$ is non-singular.
    
    \item[\textnormal{(b)}] The matrix $\mathscr{C}_{\mathrm{d}} \in \mathbb{R}^{K\times K}$ is non-singular. 
\end{itemize}
\end{assumption}

\begin{theorem}\label{Th: Cons est alphap fin} Let Assumptions \ref{as:Assumptions of (eps_k) and assoc. covop}, \ref{as: C_d and hat C_d finite injective},  \eqref{eq:representation alpha_p, diag, fin}, \eqref{eq:finite-dim. data}, and the conditions of Proposition \ref{prop:Explicit moments} with $\nu >4$ hold, and $\vartheta_{\!N}=\mathrm O(N^{-1/2})$ with $\vartheta_{\!N}$ in \eqref{ceps-tik} hold. Then, for any norm $\|\cdot\|$ on $\mathbb{R}^{K\times pK},$ 
$$
    \|\hat{\boldsymbol{\alpha}}_{\mathrm{d}} - \tilde{\boldsymbol{\alpha}}_{\mathrm{d}}\| = \mathrm{O}_{\Prob}(N^{-1/2}).
$$
\end{theorem}

The following example shows that the covariance operator of an CCC-op-ARCH process can be injective in a finite-dimensional setting.

\begin{example}\label{ex:injectivity C fin., diag}
Let $p=1,$ and assume \eqref{eq:representation alpha_p, diag, fin}--\eqref{eq:finite-dim. data} hold. Suppose further that $X_0$ has the structure of Example \ref{ex:2}, let $\boldsymbol{\varepsilon} = (\varepsilon_k)$ be an i.i.d.~process of standard Brownian motions on $[0,1]$, and assume $a_{1\ell\ell}<\pi^2/12$ for all $\ell\in\mathbb{N}$. By the Karhunen--Loève expansion \citep[Theorem~7.3.5]{HsingEubank2015}, the scores $\langle \varepsilon_k, e_\ell \rangle$ are centered Gaussian variables, independent across $\ell$, with variances $\lambda_\ell = (\ell-1/2)^{-2}\pi^{-2}$, the eigenvalues of $\mathscr{C}_{\boldsymbol{\varepsilon}}$. The structure of $X_0$ gives $X_{0,\ell\ell} = \langle X_0, e_\ell \rangle^2 = Z_{0,\ell}\langle \varepsilon_0, e_\ell \rangle^2$ for each $\ell\in\mathbb{N}$, where
$$
    Z_{0,\ell} = d_\ell + a_{1\ell\ell}Z_{-1,\ell}\langle \varepsilon_{-1}, e_\ell \rangle^2, \quad \ell\in\mathbb{N},
$$
with $d_\ell$ being the coefficients in the diagonalization of $\Delta$. Since $Z_{0,\ell}$ is independent of $\varepsilon_0$ for each $\ell$, it follows that
\begin{align*}
    \mathscr{C}_{\mathrm{d}} 
    &= \Exp\diag^\ast(\tilde{X}_0^{\otimes 2})\big[\diag^\ast(\tilde{X}_0^{\otimes 2})\big]^\top \\
    &= \Big(\Exp\!\big(X^2_{0,ij}\big) - \Exp\!\big(X_{0,ii}\big)\!\Exp\!\big(X_{0,jj}\big)\Big)_{i,j=1}^K \\
    &= \diag\!\bigg(a_1^2\Big[3\Exp(Z_{0,1}^2) - \big(\Exp(Z_{0,1})\big)^2\Big], \dots, a_K^2\Big[3\Exp(Z_{0,K}^2) - \big(\Exp(Z_{0,K})\big)^2\Big]\bigg).
\end{align*}
Strict and weak stationarity of the volatility processes $(Z_{k,\ell})_k$ are ensured since $a_{1\ell\ell}\in(0,\pi^2/12)$ and $\lambda_\ell\in(0,4/\pi^2]$ imply $\lambda_\ell a_{1\ell\ell}\in(0,1)$ for all $\ell$, and because \citep[cf.][Theorem 2.5 and Remark 2.6]{FrancqZakoian2019}
$$
    \Exp(Z_{0,\ell}) = \frac{d_\ell}{1-\lambda_\ell a_{1\ell\ell}}, \quad \ell\in\mathbb{N}.
$$
Moreover, $\Exp(Z_{0,\ell}^2)$ exists for all $\ell$, as $3\lambda_\ell^2 a_{1\ell\ell}^2\in(0,1)$, and thus $\lambda_\ell a_{1\ell\ell} \in(0,1)$, with
$$
    \Exp(Z_{0,\ell}^2) = \frac{d_\ell^2(1+\lambda_\ell a_{1\ell\ell})}{(1-\lambda_\ell a_{1\ell\ell})(1-3\lambda_\ell^2 a_{1\ell\ell}^2)}, \quad \ell\in\mathbb{N}.
$$
Therefore, as
$$
    3\Exp(Z_{0,\ell}^2) - \big(\Exp(Z_{0,\ell})\big)^2 \;=\; \frac{2d_\ell^2}{(1-\lambda_\ell a_{1\ell\ell})^2(1-3\lambda_\ell^2 a_{1\ell\ell}^2)} \;>\; 0, \quad \ell\in\mathbb{N},
$$
it follows that $\mathscr{C}_{\mathrm{d}}$ is a diagonal matrix with strictly positive diagonal entries, and therefore is invertible.
\end{example}

\subsection{Infinite-dimensional setting}\label{Sec: Estimation inf-dim}

To extend the results of the previous section to the infinite-dimensional setting, the operation ``$\diag$'' must be generalized. In order to do so, we consider for H-S operators $A \in \mathcal{S}$ their orthogonal series expansion in the basis $(e_i\otimes e_j).$ Further we let $\ell^2 = \ell^2(\mathbb{N})$ denote the space of square-summable real valued sequences $(a_i)^\infty_{i=1} \subset \mathbb{R}.$ The bounded operator $\diag : \ell^2 \to \mathcal{S}$ and its adjoint $\diag^\ast : \mathcal{S} \to \ell^2$ are defined by
\begin{alignat*}{2}
    \diag\big((a_i)^\infty_{i=1}\big) \coloneqq \sum_{i=1}^\infty a_i (e_i \otimes e_i)\,,& \qquad 
    &&\diag^\ast\!\bigg(\sum_{i=1}^\infty \sum_{j=1}^\infty a_{ij}(e_i \otimes e_j)\bigg) \coloneqq (a_{ii})^\infty_{i=1}\,.
\end{alignat*}

Here, we assume that $\boldsymbol{\alpha}$ has the infinite-dimensional form \eqref{eq:diagARCH-2}. The YW-type equation \eqref{eq:YW of diag alpha_p} also holds here, with $\mathscr{D}_{\mathrm{d}} \coloneqq \diag^\ast\mathscr{D}\diag : (\ell^2)^p \to \ell^2$, $\mathscr{C}_{\mathrm{d}} \coloneqq \diag^\ast\mathscr{C}\diag : (\ell^2)^p \to (\ell^2)^p$, and $\mathscr{R}_{\mathrm{d}} \coloneqq \diag^\ast\mathscr{R}\diag : (\ell^2)^p \to \ell^2$. The operator $\boldsymbol{\alpha}_{\mathrm{d}} : (\ell^2)^p \to \ell^2$ is
\begin{align}\label{eq:Def bold alpha_p diag}
    \boldsymbol{\alpha}_{\mathrm{d}} \coloneqq \big(\diag\!\big((a_{1jj})^\infty_{j=1}\big)\; \cdots\; \diag\!\big((a_{pjj})^\infty_{j=1}\big)\big),
\end{align}
where each $\diag((a_{ijj})^\infty_{j=1})$ is the infinite-dimensional diagonal matrix of its coefficients. Then, for any $x = (x_1, \dots, x_p)^\top \in (\ell^2)^p$, with $x_i = (x_{ij})^\infty_{j=1}$,
$$
    \boldsymbol{\alpha}_{\mathrm{d}}(x) = \bigg(\sum^p_{i=1}a_{ijj}x_{ij}\bigg)^\infty_{j=1}\,,
$$
which lies in $\ell^2$ whenever $\sup_{i,j}|a_{ijj}| < \infty$.

Unlike in the finite-dimensional case, we cannot directly estimate all coefficients of $\boldsymbol{\alpha}_{\mathrm{d}}$, i.e.
$$
    \tilde{\boldsymbol{\alpha}}_{\mathrm{d}} \coloneqq \diag^\ast(\boldsymbol{\alpha}_{\mathrm{d}}).
$$
We therefore adopt Tikhonov regularization and project onto a finite-dimensional subspace of dimension $K \in \mathbb{N}$, with $K = K_{\!N} \to \infty$ as $N \to \infty$. To estimate $\tilde{\boldsymbol{\alpha}}_{\mathrm{d}}$, we propose
\begin{align}\label{eq:estimator for all coefficients in diag, infinite} 
    \hat{\boldsymbol{\alpha}}_{\mathrm{d}} \coloneqq \diag^\ast\!\bigg(\hat{\mathscr{D}}_{\mathrm{d}}\hat{\mathscr{C}}^\dagger_{\mathrm{d}}\!\coprod^{\hat{c}_{K,\mathrm{d}}}_{\hat{c}_{1,\mathrm{d}}}\bigg)\,,
\end{align}
where $\hat{\mathscr{D}}_{\mathrm{d}} \coloneqq \diag^\ast\hat{\mathscr{D}}\diag$ with $\hat{\mathscr{D}}$ in \eqref{eq:est lag-0-cov-op},
$\hat{\mathscr{C}}_{\mathrm{d}} \coloneqq \diag^\ast\hat{\mathscr{C}}\diag$ with $\hat{\mathscr{C}}$ in \eqref{eq:est lag-1-cross-cov-op}, 
$$\hat{\mathscr{C}}^\dagger_{\mathrm{d}} \coloneqq (\hat{\mathscr{C}}_{\mathrm{d}} + \vartheta_{\!N}\mathbb{I})^{-1}  \mbox{ with } \vartheta_{\!N} \to 0,$$ and where $(\hat{\lambda}_{j,\mathrm{d}}, \hat{c}_{j,\mathrm{d}})$ and $(\lambda_{j,\mathrm{d}}, c_{j,\mathrm{d}})$ are eigenpairs of $\hat{\mathscr{C}}_{\mathrm{d}}$ and $\mathscr{C}_{\mathrm{d}}$, respectively. To establish the consistency of this estimator, we impose:

\begin{assumption}\label{as: Dim eigenspaces op-ARCH, diag, inf} 
There exists $\xi \in \mathbb{N}$ such that:
\begin{itemize}
    \item[\textnormal{(a)}] The dimensions of all eigenspaces of $\mathscr{C}_{\mathrm{d}}$ are bounded above by $\xi;$
    \item[\textnormal{(b)}] For all large $N$, with probability one, the eigenvalues of $\hat{\mathscr{C}}_{\mathrm{d}}$ satisfy $\hat{\lambda}_{j,\mathrm{d}} \neq \hat{\lambda}_{j+1,\mathrm{d}}$ for $j=1,\dots, K+\xi.$ 
\end{itemize}
\end{assumption}

\noindent We define $(\Lambda_{\ell, \mathrm{d}})_{\ell \in \mathbb{N}}$, the reciprocal eigengaps of $\mathscr{C}_{\mathrm{d}}$ as
\begin{align}\label{eq:Def Lambda_K}
    \Lambda_{\ell,\mathrm{d}} \coloneqq \frac{1}{\lambda_{\ell,\mathrm{d}} - \lambda_{\ell_b,\mathrm{d}}}, \quad \ell \in \mathbb{N},
\end{align}
where
\begin{align}\label{eq:Def ell b}
    \ell_b \coloneqq \inf \big\{ j > \ell : \lambda_{j,\mathrm{d}} < \lambda_{\ell,\mathrm{d}} \big\}.
\end{align}
Thus $\Lambda_{\ell, \mathrm{d}}$ is the reciprocal eigengap between the eigenspace of $c_{\ell,\mathrm{d}}$ and the next distinct one. By Assumption \ref{as: Dim eigenspaces op-ARCH, diag, inf}~(a), $\ell_b - \ell \leq \xi$. By part (b), the empirical analogues are
\begin{align}\label{eq:Def hat Lambda_K}
    \hat{\Lambda}_{\ell,\mathrm{d}} \coloneqq \frac{1}{\hat{\lambda}_{\ell,\mathrm{d}} - \hat{\lambda}_{\ell+1,\mathrm{d}}}\,, \quad 1 \leq \ell \leq K + \xi.
\end{align}
Lemma \ref{Lemma: Cons lagged cross-cov non-centered} and the definitions of $\hat{\mathscr{C}}_{\mathrm{d}},\mathscr{C}_{\mathrm{d}}$ yield $\hat{\Lambda}_{\ell,\mathrm{d}} = \mathrm{O}_{\Prob}(\Lambda_{\ell,\mathrm{d}})$ for $1\leq \ell \leq K+\xi$ \citep[cf.][Lemma A.3]{KuehnertRiceAue2026}. 

To establish consistency in the H-S norm, we impose a regularity condition that governs the approximation of $\boldsymbol{\alpha}$ by its finite-dimensional projections. This condition is analogous to the Sobolev-type smoothness assumptions introduced in \cite{hall:meister:2007} for deconvolution problems.

\begin{assumption}\label{eq: Sob cond op-ARCH, infinite diag} 
For some $\gamma>0$,
\begin{align}
    \sum^p_{i=1}\sum^\infty_{\ell=1} \,a^2_{i\ell\ell}(1 + \ell^{2\gamma})\, < \,\infty.    
\end{align}
\end{assumption}

We note that this assumption is stricter than $\boldsymbol{\alpha}$ is H-S, which is equivalent to  square-summability of the coefficients. We may now state our main consistency result.
 
\begin{theorem}\label{theo:Cons est alphap infinite diag} 
Suppose $(X_k)$ is a CCC-op-ARCH$(p)$ process satisfying the conditions of Proposition \ref{prop:Explicit moments} for $\nu=4$. Further, let  Assumptions \ref{as:Assumptions of (eps_k) and assoc. covop}, \ref{as: bold alpha_p H-S}, \ref{as: Dim eigenspaces op-ARCH, diag, inf}, and \ref{eq: Sob cond op-ARCH, infinite diag} hold. Let $K=K_{\!N}\to\infty$, $\vartheta_{\!N}\to 0$ with $K^{\gamma+1/2}a^{-2}_K\Lambda^2_{K,\mathrm{d}} = \mathrm{O}(N^{1/2})$ and  $\vartheta_{\!N} = \mathrm{O}(\min(a_K,\lambda_{K, \mathrm{d}})K^{-\gamma})$, with $\gamma$ defined in Assumption \ref{eq: Sob cond op-ARCH, infinite diag}. Then
$$              
    \|\hat{\boldsymbol{\alpha}}_{\mathrm{d}} - \tilde{\boldsymbol{\alpha}}_{\mathrm{d}}\|_{\mathcal{S}} = \mathrm{O}_{\Prob}(K^{-\gamma}).
$$
\end{theorem}

\begin{example}\label{ex:Example infinite-dim rates}
Let $p=1,$ and $X_0$ be as in Example \ref{ex:2}, where $(\varepsilon_k)$ is a sequence of i.i.d.~standard Brownian motions on $[0,1]$, with   $a_{1\ell\ell} \leq \pi^2/12$ for all $\ell\in\mathbb{N}.$ Further suppose that the square-summable sequences $(a_{1\ell\ell}), (d_\ell)\subset(0,\infty)$ of coefficients of $\alpha_1$ and $\Delta$ are strictly decreasing.

First, we consider the case when $a_{1\ell\ell} \asymp \ell^{-a}$ and $d_\ell \asymp \ell^{-d}$ as $\ell \to \infty$ for some $a,d>1/2$. Assumptions \ref{as:Assumptions of (eps_k) and assoc. covop}--\ref{as: bold alpha_p H-S} then hold. Since the entries of $\mathscr{C}_{\mathrm{d}}$ (cf. Example \ref{ex:injectivity C fin., diag}) are distinct (as $(\lambda_\ell)$, $(a_{1\ell\ell})$, $(d_\ell)$ are decreasing), Assumption \ref{as: Dim eigenspaces op-ARCH, diag, inf}~(a) holds, and Assumption \ref{eq: Sob cond op-ARCH, infinite diag} is satisfied for any $\gamma \in (0, a-1/2)$. Let $\gamma = a-1/2-c$ for small $c\in (0,a-1/2).$ The eigenvalues of $\mathscr{C}_{\boldsymbol{\varepsilon}}$ satisfy $a_K \asymp K^{-2}$, and those of $\mathscr{C}_{\mathrm{d}}$ fulfill
$$
    \lambda_{K,\mathrm{d}} \;=\; a_K^2\Big[3\Exp(Z^2_{0,K}) - \big(\Exp(Z_{0,K})\big)^2\Big] \;\asymp\; K^{-4}d_K^2 \;\asymp\; K^{-(2d+4)}.
$$
Therefore, the reciprocal eigengaps satisfy $\Lambda_{K,\mathrm{d}} \asymp K^{2d+5}$. With $K=K_{\!N}\asymp N^{1/2(a-c+4d+14)}$, it follows
$$
    K^{\gamma+1/2}a_K^{-2}\Lambda^2_{K,\mathrm{d}} \;\asymp\; K^{a-c+4d+14} \;=\; \mathrm{O}(N^{1/2}).
$$

Choosing $\vartheta_{\!N} \to 0$ with $\vartheta_{\!N} = \mathrm{O}(\min(a_K,\lambda_{K, \mathrm{d}})K^{-\gamma})$, Theorem \ref{theo:Cons est alphap infinite diag} gives
$$
    \|\hat{\boldsymbol{\alpha}}_{\mathrm{d}} - \tilde{\boldsymbol{\alpha}}_{\mathrm{d}}\|_{\mathcal{S}} = \mathrm{O}_{\Prob}\big(N^{-\frac{2a-1-2c}{4a-4c+16d+56}}\big).
$$

Faster rates occur for larger $a$, slower decay of $(\lambda_\ell)$, and smaller $d$. For instance, with $d=1$, one achieves a rate near $N^{-1/4}$ if $a=18$. 

Now, suppose $a_{1\ell\ell} \asymp q^\ell$ for $q\in(0,1).$ In this case Assumption \ref{eq: Sob cond op-ARCH, infinite diag} holds for all $\gamma>0$, and hence with an appropriate choice of $K_N,$ we obtain the near parametric rate $N^{-1/2}\colon$
$$
    \|\hat{\boldsymbol{\alpha}}_{\mathrm{d}} - \tilde{\boldsymbol{\alpha}}_{\mathrm{d}}\|_{\mathcal{S}} = \mathrm{O}_{\Prob}\big(N^{-1/2+\epsilon}\big), \mbox{ for any }\epsilon>0.
$$
\end{example}
 
\begin{remark} {\rm \noindent A weak convergence result for $\hat{\boldsymbol{\alpha}}_{\mathrm{d}}$ to a non-trivial limit is not available for the full operators in the fAR model underlying our ARCH framework \citep[][Theorem 3.2]{Mas2007}. Under technical conditions, Theorem 3.1 of the same work gives asymptotic normality for prediction errors at fixed points. In functional linear regression, which is in the context of our parameter estimation closely related,  \cite{KuttaDierickxDette2022} obtain a pivotal test statistic for the slope operator under smoothness assumptions. While similar ideas might extend to our setting, we focus on weak consistency.}
\end{remark}

\subsection{Estimation of the Intercept term}\label{Sec:estimation Delta}

From $\mathscr{C}_{\!\boldsymbol{X}} = \mu_{\boldsymbol{\Sigma}}\mathscr{C}_{\boldsymbol{\varepsilon}}$ and Eq.~\eqref{eq:definition op-ARCH}, it follows 
\begin{align}\label{eq:Identity for Delta in ARCH}
    \Delta \mathscr{C}_{\boldsymbol{\varepsilon}} = \mathscr{C}_{\!\boldsymbol{X}} - \big(\boldsymbol{\alpha}(m_p)\big)\mathscr{C}_{\boldsymbol{\varepsilon}},
\end{align}
where $m_p \coloneqq \Exp(X^{\otimes2,[p]}_0)$. Accordingly, we estimate $\Delta$ by
\begin{align}\label{eq:Esimator Delta ARCH}
    \hat{\Delta} \coloneqq \Big[\hat{\mathscr{C}}_{\!\boldsymbol{X}} - \big(\boldsymbol{\hat{\alpha}}(\hat{m}_p)\big)\mathscr{C}_{\boldsymbol{\varepsilon}}\Big]\mathscr{C}^\dagger_{\!\boldsymbol{\varepsilon}}\!\coprod^{e_{K}}_{e_1},
\end{align}
with $\hat{\mathscr{C}}_{\!\boldsymbol{X}} = \hat{m}_1$ and $\hat{m}_p$ defined in \eqref{eq:Def moment estimates}.  

We next state a consistency result for $\Delta.$ Since $\mathscr{C}_{\boldsymbol{\varepsilon}}$ and $\Sigma_k$ commute for each $k,$ it holds 
\begin{align}\label{delta-d-eq}
    \Delta = \sum_{i=1}^\infty d_i (e_i \otimes e_i),    
\end{align}
for some non-negative, square-summable sequence $(d_i)$. Consistency of the estimation errors for $\Delta$ is also derived based on a Sobolev condition.

\begin{proposition}\label{prop:asymptotic result Delta ARCH, inf} 
Let the conditions of Theorem \ref{theo:Cons est alphap infinite diag} hold. Further, for some $\delta>0,$ assume that the coefficients in \eqref{delta-d-eq} satisfy 
\begin{align}\label{eq:Sobolev cond Delta} 
    \sum^\infty_{i=1}\,d^2_i(1 + i^{2\delta}) < \infty.    
\end{align}
Then, it holds that
$$
    \|\hat{\Delta} - \Delta\|_\mathcal{S} = \mathrm{O}_{\Prob} \big(a^{-1}_KN^{-1/2} \big) + \mathrm{O}_{\Prob}(K^{-\delta}) + \mathrm{O}_{\Prob}\big(\|\boldsymbol{\hat{\alpha}} - \boldsymbol{\alpha}\|_\mathcal{S}\big)\,.
$$ 
\end{proposition}

\section{Simulation Study}\label{Section Simulation}

In this section, we present the results of simulation experiments that aim to illustrate the CCC-op-ARCH$(p)$ process, and evaluate the estimation procedures detailed in Section \ref{Sec:estimation}. In each of the examples below, we view $(X_k)=\{X_i(t), \; k \in \mathbb{Z}, \; t\in [0,1]\}$ as real-valued stochastic processes taking values in the Hilbert space $\mathcal{H} =L^2[0,1]$. All analysis was done on a personal laptop in the {\tt R} programming language; \cite{r}. Code that may be used to reproduce the the numerical work below is available at \url{github.com/jrvanderdoes/fungarch/}. 

\subsection{Implementation Details}
 
The proposed estimators require the user to specify $K$, the dimension reduction parameter, as well as the Tikhonov parameter $\vartheta_N$. Throughout, we chose $K$ according to a modified total-variation-explained (TVE) criterion. Namely, with $(e_j)$ again denoting the eigenfunctions of $\mathscr{C}_\varepsilon$, and  $X_{i,k} = \sum_{j=1}^k \langle X_i, e_j \rangle e_j$, we then choose 
$$
    K = \min\left\{ k \; : \;   \frac{\sum_{i=1}^N  \| X_i - X_{i,k} \|^2}{\sum_{i=1}^N \|X_i\|^2 } \le 1-\mbox{TVE} \right\}.
$$
We set TVE to $0.9$ below, unless otherwise specified. In the subsequent application to high-frequency asset price data, a 90\% TVE typically resulted in a $K$ in the range of 10--15. 

In order to choose $\vartheta_N$, we employ 1-step ahead cross-validation. The data of length $N$ are split into training and testing sets of size $N_{train}$ and $N_{test}$. Below we use a respective $80\%$/$20\%$ split. Since the subsequent data analysis aims to forecast pointwise conditional quantiles, we chose $\vartheta_N$ in order to minimize an integrated check-loss function measuring how well $\hat{\Sigma}_i$ may be used to predict the quantiles of $X_i$. In particular, let the $\alpha$ level check-loss function $\rho_\alpha: \mathbb{R} \to [0,\infty)$ be denoted as
$$
    \rho_\alpha(u)=u \times\big(\alpha-\mathds{1}_{\{u<0\}}\big).
$$
Note that if a CCC-op-ARCH model has Gaussian innovations, then the pointwise conditional $\alpha$ quantile of $X_i(t)$ is 
$$
    \sqrt{{\Sigma}^{1/2}_{j}\mathscr{C}_{\boldsymbol{\varepsilon}}{\Sigma}^{1/2}_{j}(t,t)} \times \Phi^{-1}\left(\frac{\alpha}{C_\varepsilon^{1/2}(t,t)} \right).
$$
For a given fitted CCC-op-ARCH model producing forecasts of the conditional covariance operator $\hat{\Sigma}_j$, viewed as a function of the parameter $\vartheta_N$ and computed with an expanding window,  we then chose  $\vartheta_N$ to minimize 
$$
    \mbox{CV}_{Err}(\vartheta_N) = \frac{1}{N_{test}} \sum_{j \in \mbox{ test set }} \int_0^1 \rho_\alpha\left(\sqrt{   \hat{\Sigma}^{1/2}_{j}\mathscr{C}_{\boldsymbol{\varepsilon}}\hat{\Sigma}^{1/2}_{j}(t,t)} \times \Phi^{-1}\left(\frac{\alpha}{C_\varepsilon^{1/2}(t,t)} \right) - X_i(t) \right)\mathrm{d}t.
$$
Optimizing for the Tikhonov parameter $\vartheta_N$ is somewhat computationally intensive for large $p$, $K$, and $N$. An alternative approach that is computationally efficient, although does not necessarily lead to a consistent estimator, is to use Moore--Penrose pseudoinverses \citep{moore1920, penrose1955}. In this case we modify \eqref{eq:estimator for all coefficients in diag, infinite} by replacing  $\mathscr{C}^\dagger_{\!\boldsymbol{\varepsilon}}\!\coprod^{e_{K}}_{e_1}$ and $\hat{\mathscr{C}}^\dagger_{\mathrm{d}}$ in the definition of $\hat{\boldsymbol{\alpha}}_{\mathrm{d}}$ with 
\begin{align}
    \mathscr{C}^\ddagger_{\!\boldsymbol{\varepsilon}}= \sum_{i=1}^{K}\,a^{-1}_i(e_i \otimes  e_i)\,, \quad \mbox{ and } \quad \mathscr{C}^\ddagger_{\mathrm{d}}   = \sum_{i=1}^K\,\hat\xi^{-1}_i(\hat\varphi_i \otimes \hat\varphi_i),
\end{align}
where $(\hat\xi_i,\hat\varphi_i)$ are eigenvalues and eigenfunctions of $\hat{\mathscr{C}}$. We also compared to this estimator in our simulation experiments. 

\subsection{Data Generation}
 
\begin{figure}
    \centering
    \subfloat[Example I]{
        \includegraphics[width=0.4\textwidth,trim={3cm 2cm 4cm 3.5cm},clip]{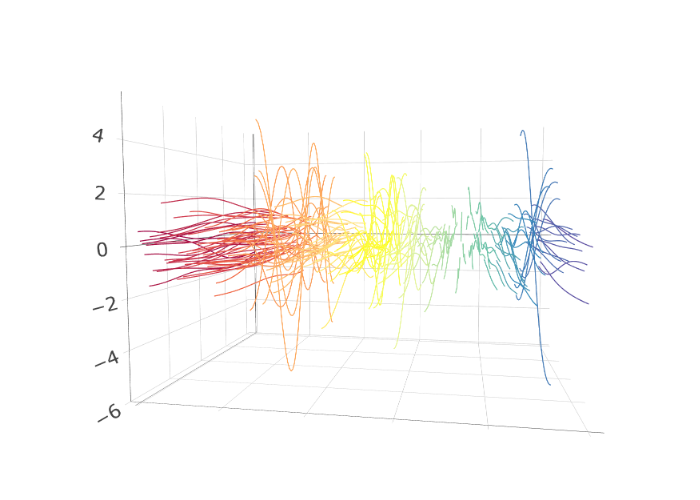}
    } \hfill
    \subfloat[Example II]{
        \includegraphics[width=0.4\textwidth,trim={3cm 2cm 4cm 3cm},clip]{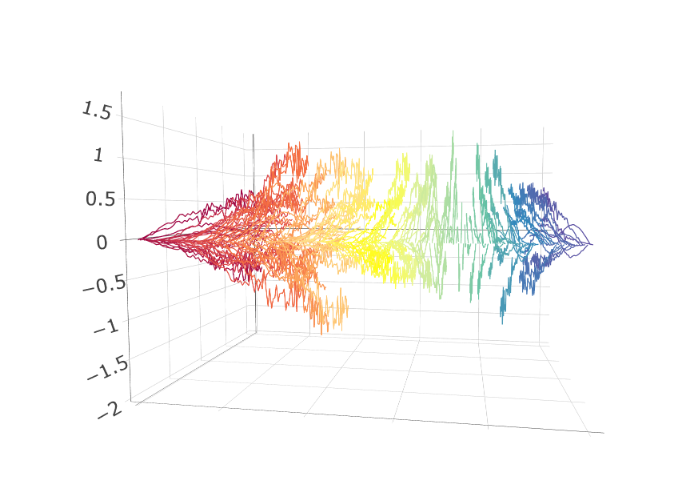}
    }
    \caption{\textbf{CCC-op-ARCH Examples}. (a) Uses $C_\epsilon$ based on Ornstein–Uhlenbeck errors. (b) Uses $C_\epsilon$ based on Brownian motion errors.}
    \label{fig:examples}
\end{figure}

We simulated CCC-op-ARCH$(p)$ data $(X_k)=\{X_k(t), \; k \in \mathbb{Z}, \;t\in [0,1]\}$ as stochastic processes taking values in the space $\mathcal{H} =L^2[0,1]$. We took the error covariance operator $\mathscr{C}_{\boldsymbol{\varepsilon}}$ to be a kernel-integral operator with kernel $C_\varepsilon$ corresponding to either an Ornstein-Uhlenbeck (OU) process
\begin{align}\label{ou-err}
    C_{\varepsilon}(t,s) = e^{-|t-s|/2},
\end{align}
or a standard Brownian motion (BM)
\begin{align}\label{bm-err}
    C_{\varepsilon}(t,s) = \min(t,s)\,.
\end{align}
After generating errors with the specified covariance structure, each functional data object was simulated so that
\begin{align*}
    X_i=\Sigma_i^{1/2}(\varepsilon_i)\,,  \quad   \Sigma_i=\Delta+\sum_{j=1}^p\alpha_j(X_{i-j}\otimes X_{i-j}),
\end{align*}
for $i=-b, \dots, N$, with $b=100$ denoting a burn-in period that is discarded. The op-ARCH operators were constructed as
\begin{align} \label{eq:alpha_sim}
    \alpha_j = \sum_{k=1}^d a_{k,j} (e_k\otimes e_k) \otimes ( e_k\otimes e_k),
\end{align}
with ${\bf a}_j=(a_{1,j}, \dots, a_{d,j})$ denoting a vector of scale parameters. We further set $\Delta=\mathscr{C}_\varepsilon\,$. In the simulations below, each functional data object is simulated on a grid of $r=50$ equally spaced points on the unit interval $[0,1]$. We verified in unreported simulations that increasing the value of $r$ had a negligible impact on the reported results, although taking a small value of $r$  $(r<10)$ did negatively impact the results on model estimation error.  

Spaghetti-rainbow plots illustrating the sample paths of CCC-op-ARCH$(1)$ processes of length $N=200$ are shown in Figure \ref{fig:examples}. With $\lambda_i$, $i\in\{1,2,..\}$ denoting the ordered eigenvalues of $\mathscr{C}_\varepsilon$, Figure \ref{fig:examples}(a) uses ${\bf a}_1=(0, 1.6/\lambda_2, 1.6/\lambda_3, 1,6/\lambda_4)$ and OU errors, and \ref{fig:examples}(b) uses ${\bf a}_1=(0, 1.1/\lambda_2, 1.1/\lambda_3,1.1/\lambda_4)$ with BM errors. These sample paths share some similarity with the real data studied in the following Section \ref{Section Applications}, including periods of volatility/heteroscedasticity.  

\subsection{Consistency Results}

\begin{figure}
    \centering
    \includegraphics[width=0.45\textwidth,clip]{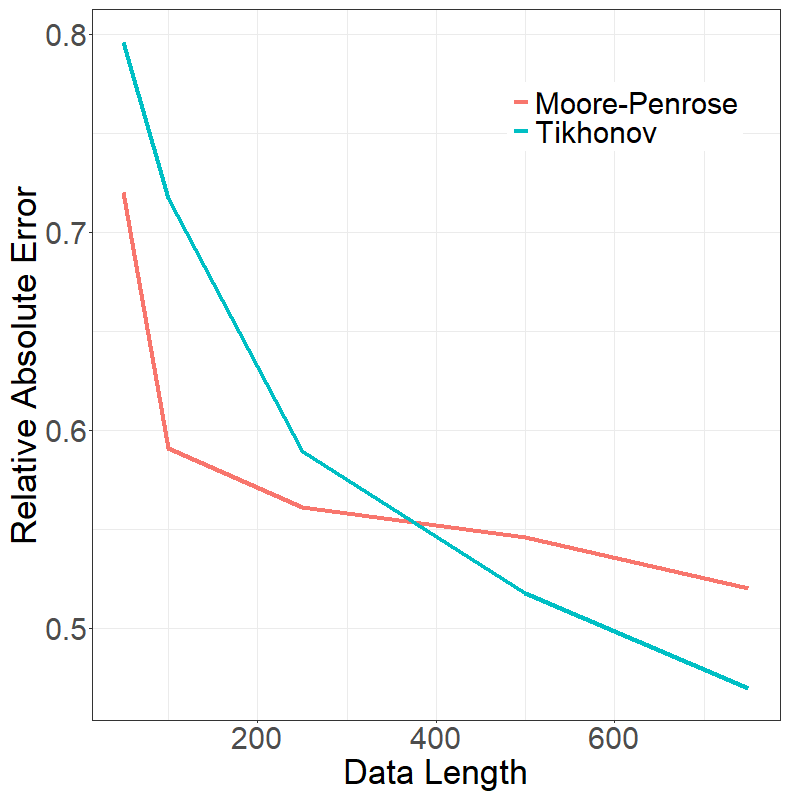} \hfill
    \includegraphics[width=0.45\textwidth,clip]{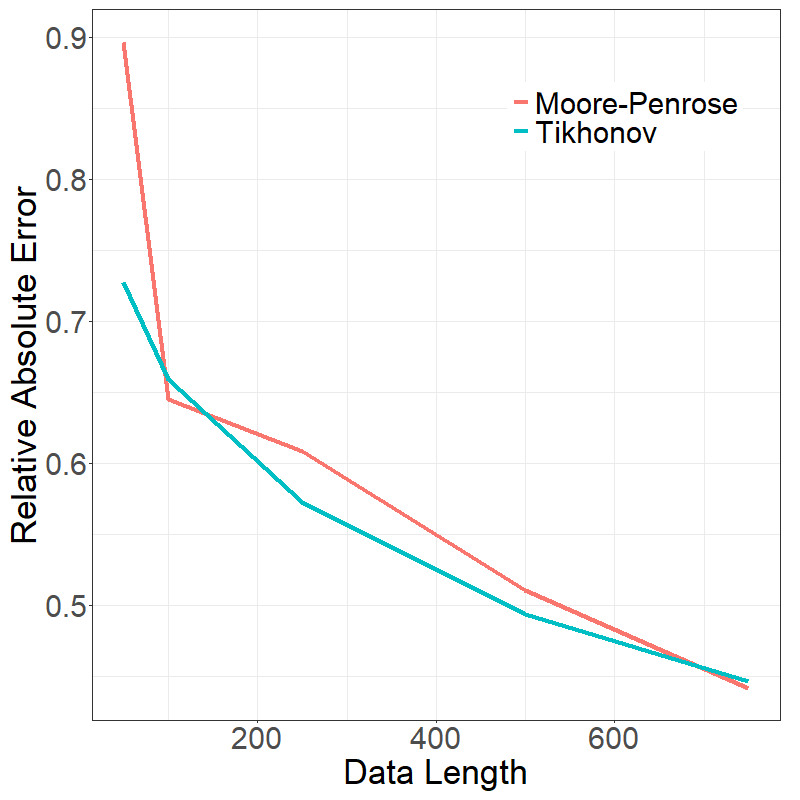}
    \caption{ Relative absolute error $e_{N,\alpha}$  for Tikhonov versus Moore Penrose-based estimators for simulated CCC-op-ARCH$(1)$ data. The left-hand panel considers a low-dimensional setting, and the right-hand panel shows a high-dimensional setting.  }  
    \label{fig:MPT}
\end{figure}

We first present results on the consistency properties of the estimators proposed in Section \ref{Sec:estimation}. For each setting, we generated samples with sample sizes $N$ independently $500$ times, with $N\in \{50, 100, 250, 500, 750\}$, using OU errors \eqref{ou-err}. We considered the $\alpha_j$ as in \eqref{eq:alpha_sim}, with either ${\bf a}_j=(0.7, 0.7)$, which we term ``low-dimensional'', and ${\bf a}_j =(1,1/2^2,...,1/20^2)$, which we term ``high-dimensional''.  

\begin{figure}
    \centering
    \subfloat[Relative estimation error $e_{N,\alpha}$ as a function of $N$ with simulated CCC-op-ARCH$(1)$ data.]{
        \includegraphics[width=0.4\textwidth,clip]{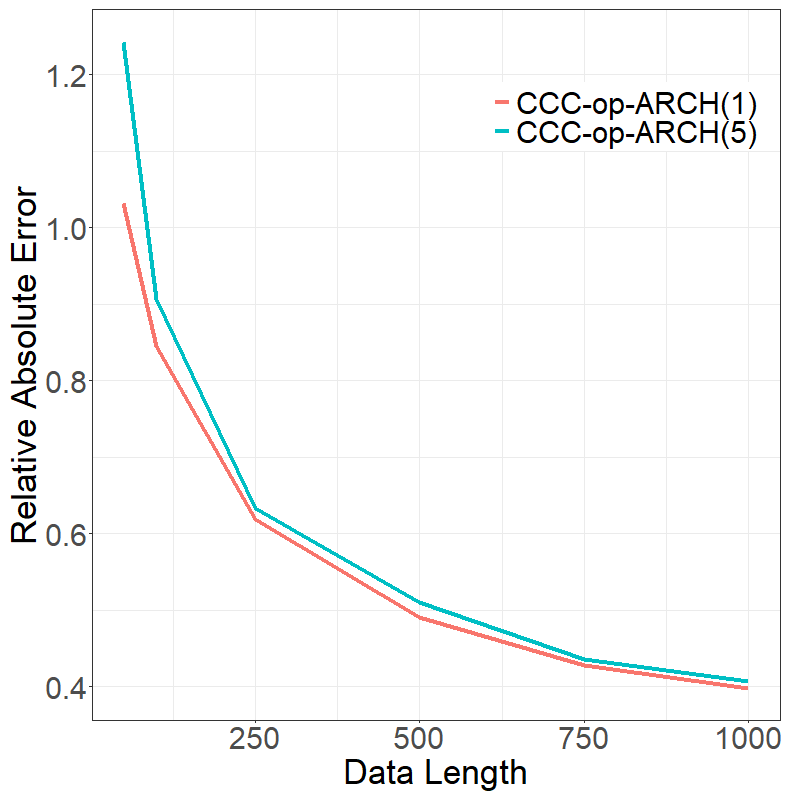}
    } \hfill
    \subfloat[Relative estimation error $e_{N,\alpha}$ as a function of $N$ with simulated CCC-op-ARCH$(5)$ data.]{
        \includegraphics[width=0.4\textwidth,trim={0cm 0cm 0cm 0cm},clip]{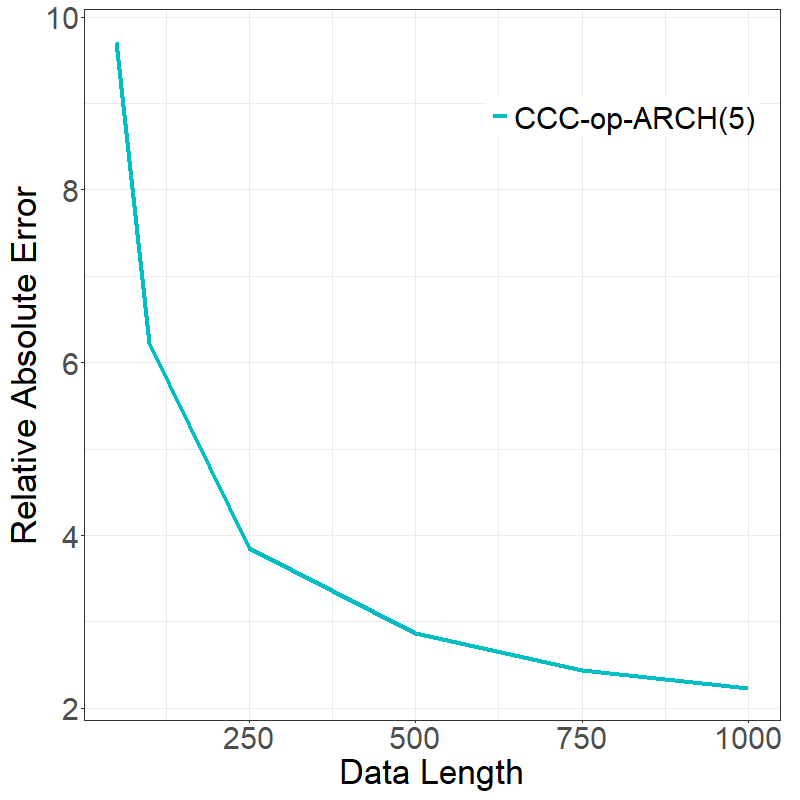}
    } \\
    \subfloat[Relative estimation error $e_{N,\Delta}$ as a function of $N$ with simulated CCC-op-ARCH$(1)$ data.]{
        \includegraphics[width=0.4\textwidth,clip]{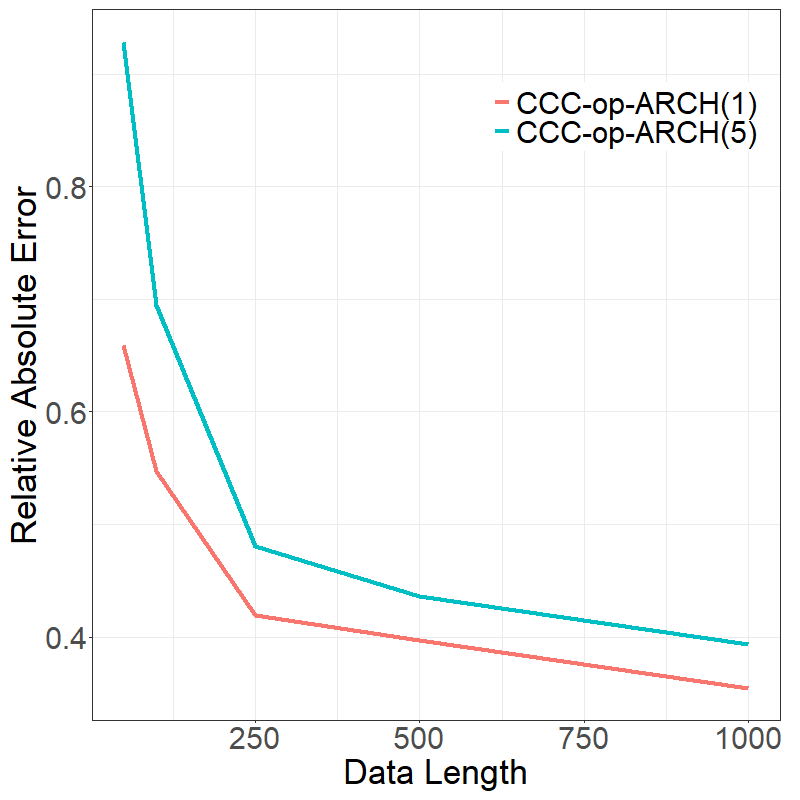}
    }    
    \caption{\textbf{Estimation Consistency}. Relative absolute error for estimation of $\Delta$ and the $\alpha_j$'s of CCC-op-ARCH$(1)$ and CCC-op-ARCH$(5)$ models in the low-dimensional setting.}
    \label{fig:consistency}
\end{figure}
Figure \ref{fig:MPT} and \ref{fig:consistency} show plots of the relative absolute error between $\Delta$ and $\hat{\Delta}$
\begin{align}
    e_{N,\Delta} = \frac{1}{500}\sum_{s=1}^{500} \frac{ \|\Delta-\hat{\Delta}_{s,N}\|}{\|\Delta \|}\,.
\end{align}
as well as for $\alpha$ and $\hat{\alpha}$,
\begin{align} 
    e_{N,\alpha} = \frac{1}{500p}\sum_{s=1}^{500} \sum_{i=1}^p \frac{\|\alpha_i-\hat{\alpha}_{i,s,N}\|}{ \|\alpha_i \|}\,,
\end{align}
for increasing values of $N$.  

Figure \ref{fig:MPT} illustrates the difference in estimation error from using either the Tikhonov or Moore--Penrose inverses in defining the estimator. The plots appear to confirm the asymptotic consistency of the proposed estimators in these settings. We observed that in the low-dimensional setting, the Moore--Penrose estimator tended to perform somewhat worse than the Tikhonov based estimator, whereas in the high-dimensional case the performance between the two methods was more comparable.  We note that the cross-validation criteria for determining $\vartheta_{\!N}$ does not intend to optimize the normed estimation error for the $\alpha_j$, but rather attempts to improve quantile forecasting.

Figure \ref{fig:consistency} shows plots of $e_{N,\Delta}$ and $e_{N,\alpha}$ in the low-dimensional setting for CCC-op-ARCH$(1)$ and CCC-op-ARCH$(5)$ data.  We observed decreasing estimation error as a function of $N$ in each setting, including when fitting a CCC-op-ARCH$(5)$ models to data generated according to a CCC-op-ARCH$(1)$ model (see Figure \ref{fig:consistency}(a)). 

\section{Application to Intra-Day Return Data}\label{Section Applications}

ARCH models are most commonly applied to model financial return data. We considered intra-day price data of the Exchange-Traded-fund SPY that tracks the S\&P500 index. The returns were created by converting the stock prices into overnight cumulative intraday returns.

\begin{definition}
    Let $P_j(t)$, $j=0,\dots,N$, be the price of a financial asset at time $t$ on day $j$. The overnight cumulative intraday returns (OCIDRs) are defined as
    \begin{equation*}
        R_j(t) = 100\left(\log P_j(t) - \log P_{j-1}(1) \right), \quad j=1,\dots, N, \quad t\in [0,1]\,.
    \end{equation*}
\end{definition}

The specific data we considered were obtained over two 3 year periods, 2018--2020 and 2022--2024,  at a resolution of one observation every $10$ minutes ($r=39$). The full sample had a size of $N=735$ (2018--2020) and $N=717$ (2022--2024). These data and the resulting OCIDR curves in 2018--2020 are illustrated in Figure \ref{fig:sp500}. 

We took as the goal of this analysis to compare the proposed model along with several other methods to forecast conditional quantiles (Value-at-Risk) of the curves $R_j$, and to evaluate the goodness-of-fit of the CCC-op-ARCH model to the data. After fitting a CCC-op-ARCH$(p)$ model, the lower $\alpha$ quantile curve forecast of $R_j(t)$ is   
\begin{align*} 
    \hat{V}_{j,\alpha}(t)= \sqrt{   \hat{\Sigma}^{1/2}_{j}\mathscr{C}_{\boldsymbol{\varepsilon}}\hat{\Sigma}^{1/2}_{j}(t,t)} \times \Phi^{-1}\left(\frac{\alpha}{C_\varepsilon^{1/2}(t,t)} \right), \quad t\in[0,1],
\end{align*}
where $\Phi^{-1}$ is the quantile function of the standard normal distribution. In addition to comparing to the pointwise ``historical'' quantile computed from all the previous observations, we also fit a pw-fARCH$(1)$ model as in \eqref{eq:pw-fARCH def}, using the estimation method of \cite{Ceroveckietal2019},  and forecast the $\alpha$ quantile curve as $\hat{V}_{j,\alpha}(t)=\hat{\sigma}_i(t)\Phi^{-1}(\alpha).$ 

Plots of the curves $\hat{V}_{j,0.05}(\cdot)$ for several models along with the observed OCIDR curves for a particularly volatile period in the S\&P 500 index including the COVID-19 Lockdown period in March, 2020, are shown in the top panel of Figure \ref{fig:app_envelope}. We observed that each model exhibited a certain degree of heteroscedasticity, although this was the most pronounced for the CCC-op-ARCH$(5)$ model. The bottom panels of Figure \ref{fig:app_envelope} show forecasted 95\% pointwise confidence sets for the OCIDR curves based on the CCC-op-ARCH$(5)$ model.

\begin{figure}
    \centering
    \subfloat[Data And Forecast]{
        \includegraphics[width=\textwidth,trim={0cm, 0cm, 0cm 0cm},clip]{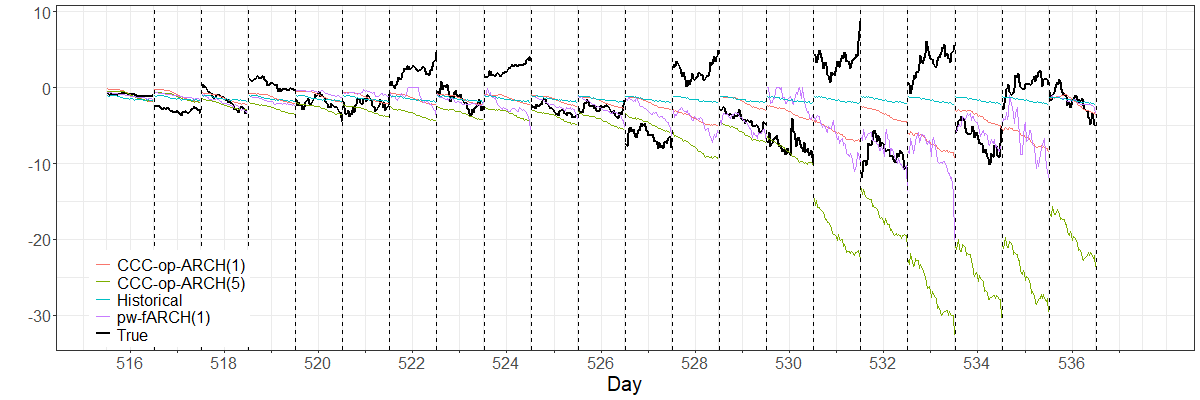}
    } \\
    \subfloat[OCIDR curves and Forecasts]{
        \includegraphics[width=0.45\textwidth,trim={6cm 4.5cm 6.5cm 5.5cm},clip]{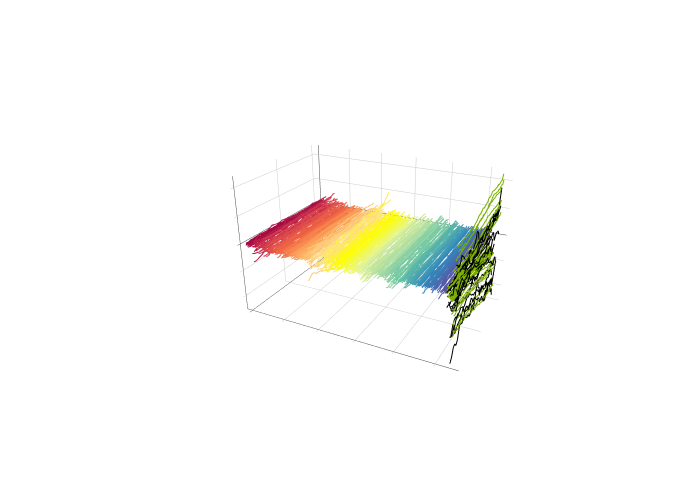}
    } \hfill
    \subfloat[Zoomed Forecast]{
        \includegraphics[width=0.45\textwidth,trim={6cm 4.5cm 6.5cm 5.5cm},clip]{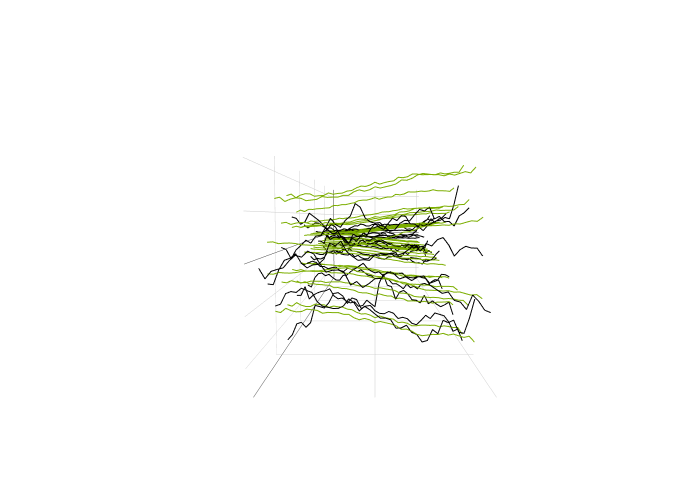}
    }
    \caption{Plots of the curves $\hat{V}_{j,0.05}(\cdot)$ for the CCC-op-ARCH$(p)$ models, $p \in \{1,5\}$, as well as the historical and pw-fARCH$(1)$ model along side the observed OCIDR curves for a particularly volatile period in the S\&P 500 index including the COVID-19 Lockdown period in March, 2020.}
    \label{fig:app_envelope}
\end{figure}

In order to assess the accuracy of these quantile curve forecasts, we split each data set into a training set based on the first two years ($N_{train}=482$, 2018--2020 and $N_{train}=465$, 2022--2024) and then forecast a third year of data ($N_{test}=253$, 2018--2020 and $N_{test}=252$, 2022--2024), which we called the test set. Each quantile function was forecasted one step ahead using an expanding window, and the average (integrated) violation rate of each model was computed as
\begin{align}\label{vr-def}
    VR_\alpha = \frac{1}{N_{test}} \sum_{j \in \text{ test set }} \int_0^1 \mathds{1}_{\{R_j(t) < \hat{V}_{j,\alpha}(t)\}}\,\mathrm{d}t. 
\end{align}
Table \ref{tab:applications} provides the observed average violation rates for each model for the nominal levels $\alpha=0.01$ and $\alpha=0.05$. We noticed that the CCC-op-ARCH$(1)$, CCC-op-ARCH$(5)$, and pw-fARCH$(1)$ models exhibited reasonably accurate coverage probabilities for the $\alpha=0.05$ level, withstanding the highly volatile S\&P500 (2018--2020) sample, especially relative to the historical quantile forecast. The CCC-op-ARCH$(5)$ and pw-fARCH$(1)$ model performed well at the $\alpha=0.01$ level, with CCC-op-ARCH$(5)$ performing the best overall.  
\begin{table}
    \centering
    \begin{tabular}{l cccc}
        \toprule
            & CCC-op-ARCH$(1)$ & CCC-op-ARCH$(5)$ & pw-fARCH$(1)$ & Historical \\\midrule
          \midrule
        \multicolumn{4}{@{}l}{\textsc{Nominal $ \alpha =0.05$}}\\
            S\&P500 (2018--2020) & 0.120 & 0.026 & 0.072 & 0.103 \\
            S\&P500 (2022--2024) & 0.049 & 0.031 & 0.033 & 0.017 \\
          \midrule
        \multicolumn{4}{@{}l}{\textsc{Nominal $ \alpha =0.01$}}\\
            S\&P500 (2018--2020) & 0.088 & 0.012 & 0.039 & 0.038\\
            S\&P500 (2022--2024) & 0.027 & 0.010 & 0.007 & 0.003 \\
          \bottomrule                          
    \end{tabular}
    \caption{One-step ahead average violation rates $VR$ computed as in equation \eqref{vr-def} for OCIDRs based on fitting two years (expanding) of data, and forecasting one-day ahead for a third year.}
    \label{tab:applications}
\end{table}

\begin{figure}
    \centering
    \subfloat[2022-2024]{
        \includegraphics[width=0.45\textwidth, trim={0cm 0cm 0cm 0}, clip]{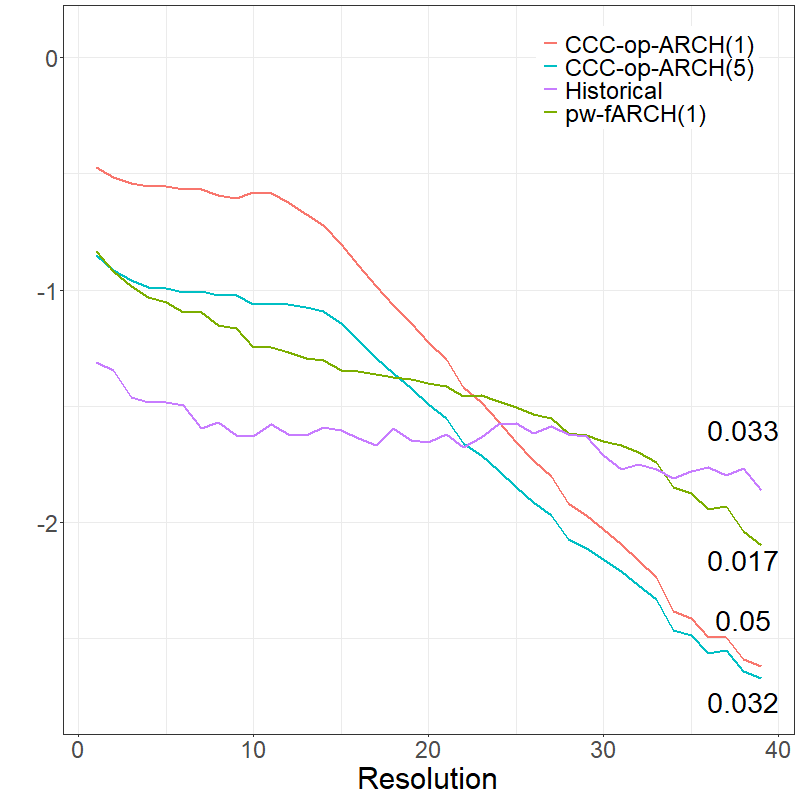}
    } \hfill
    \subfloat[2018-2020]{
        \includegraphics[width=0.45\textwidth,trim={0cm 0cm 0cm 0}, clip]{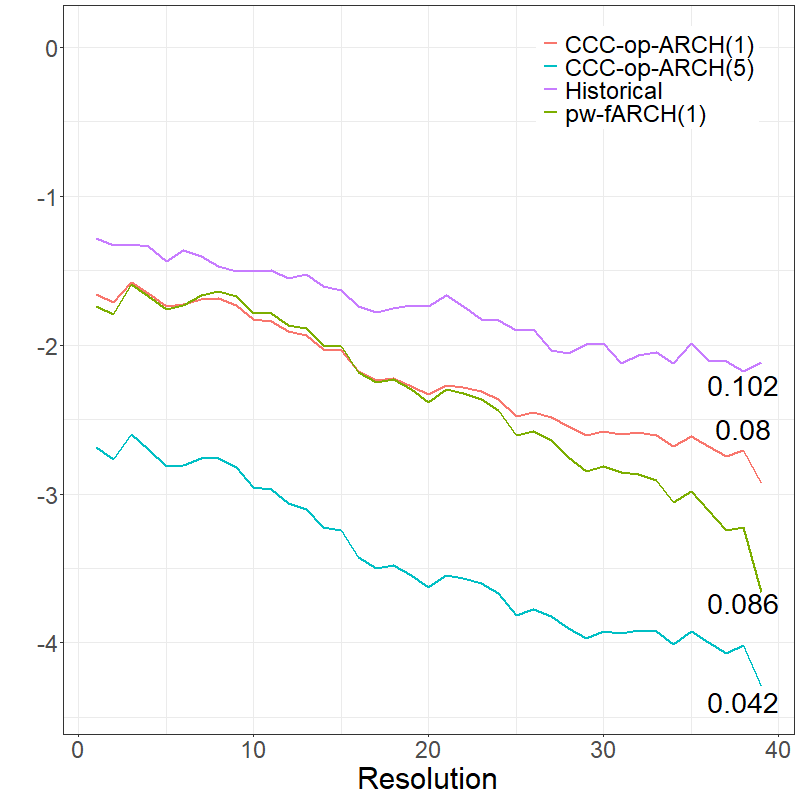}
    }
    \caption{Comparisons of the average conditional quantile forecast curves $\bar{V}_{0.05}$ computed as in \eqref{barvdef} for each model. The left-hand panel is estimated from the 2022--2024 sample and the right-hand is estimated from the 2018--2020 sample. }
    \label{fig:tightness}
\end{figure}
To further measure the fidelity of these forecasts to the data and contrast the forecasts of the models, we also computed the average quantile curve
\begin{align}\label{barvdef} 
    \bar{V}_\alpha(t) = \frac{1}{N_{test}} \sum_{j \in \text{ test set}}\hat{V}_{j,\alpha}(t)
\end{align}
for each model. These curves are shown with $\alpha=0.05$ in Figure \ref{fig:tightness} for each of the S\&P 500 samples. We observed that for 2022--2024 sample, the CCC-op-ARCH models tended to estimate a larger contrast between the variance of the curves at the beginning and end of the day, especially when compared to the pw-fARCH$(1)$ model. In the 2018--2020 sample, the average curves $\bar{V}_\alpha(t)$  were somewhat flatter for each model, although in this case only the CCC-op-ARCH$(5)$ produced forecasts with approximately nominal coverage. 

\begin{figure}
    \centering
    \subfloat[CCC-op-ARCH$(1)$ Residuals]{
        \includegraphics[width=0.48\textwidth,trim={4cm, 3cm, 2.5cm 6cm},clip]{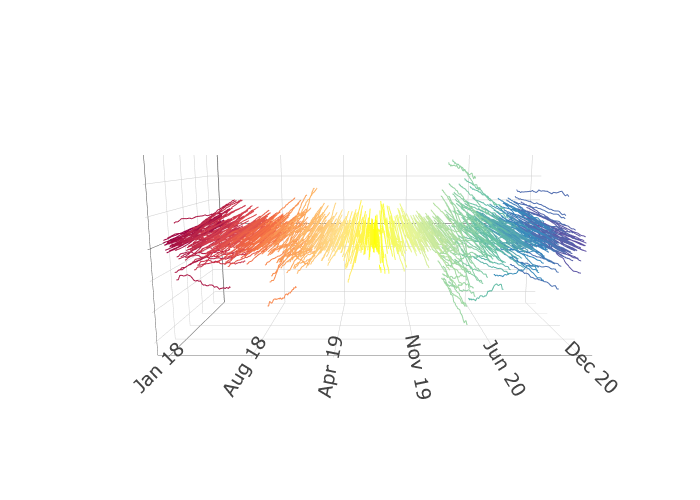}
    } \hfill
    \subfloat[CCC-op-ARCH$(5)$ Residuals]{
        \includegraphics[width=0.48\textwidth,trim={4cm, 3cm, 2.5cm 6cm},clip]{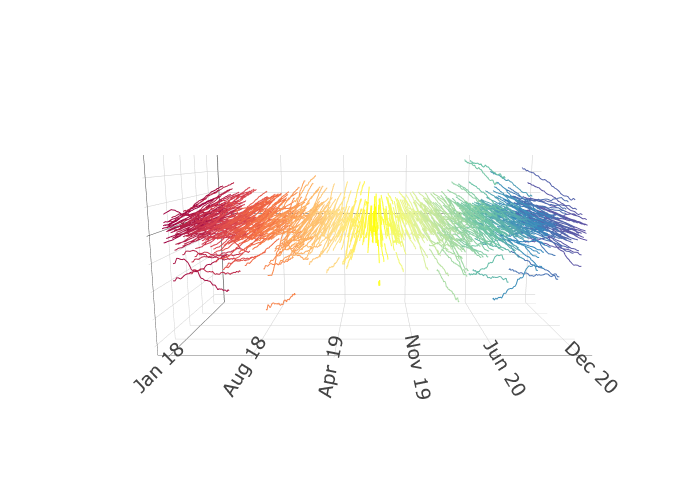}%trim={6cm 4.5cm 6.5cm 5.5cm},
    }
    \caption{Plots of residual curves computed from the CCC-op-ARCH$(p)$ models as defined in \eqref{resid-def}. }
    \label{fig:residuals}
\end{figure}

In order to assess the goodness-of-fit of the estimated CCC-op-ARCH models, we computed model residuals by applying a Moore--Penrose style pseudoinverse of $\hat{\Sigma}_i$ to $X_i$. Specifically, letting 
$$
    \hat{\Sigma}_i^\dagger = \sum_{j=1}^K\,\langle \hat{\Sigma}_i( e_j), e_j \rangle^{-1}(e_j \otimes  e_j),
$$
we defined residual curves 
\begin{align}\label{resid-def}
    \hat{\varepsilon}_i(t) = \hat{\Sigma}_i^\dagger(X_i)(t), \quad  t\in [0,1], \quad i \in \{1+p,...,N\}.
\end{align}
Plots of these residuals computed from CCC-op-ARCH$(p)$ models with $p\in\{1,5\}$ are shown in Figure \ref{fig:residuals}. Visually CCC-op-ARCH$(1)$ residuals appeared to retain some volatility. To evaluate for remaining conditional heteroscedasticity in the residuals, we investigated for serial correlation in the sequence of squared residual curves  $Y_i(\cdot)=\hat{\varepsilon}_i^2(\cdot)$. In particular, we computed the \emph{Spherical AutoCorrelation Function} (SACF), 
$$
    \tilde{\rho}_h=\frac{1}{N-p} \sum_{i=1+p}^{N-h}\left\langle \frac{ Y_i-\mu }{\|Y_i -\mu\|},  \frac{ Y_{i+h}-\mu }{\|Y_{i+h} -\mu\|}\right\rangle, 
$$
as introduced in \cite{yeh:rice:2023}, which is a robust estimator of autocorrelation in sequences of curves. We additionally applied the white noise test of \cite{Kokoszka:Rice:Shang:2017} to the squared residual curves \citep[see][for a review of these methods]{Kim:Kokoszka:Rice:2023}.

Plots of $\tilde{\rho}_h$ as a function of $h$ are shown in Figure \ref{fig:white_noise} for the squared residual curves derived from the S\&P 500 data (2018--2020). While the original squared OCIDR curves $R_i^2$ as well as the squared residuals from the CCC-op-ARCH$(1)$ model exhibit strong serial correlation, the squared residuals from the CCC-op-ARCH$(5)$ model appear  reasonably uncorrelated.  Table \ref{tab:whiteness} shows the $p$-values of white noise tests applied to the squared residual series, which also support the conclusion that the CCC-op-ARCH$(1)$ model does not entirely explain the observed conditional heteroscedasticity in the data, while the CCC-op-ARCH$(5)$ model appears to fit the data  well. The results were similar for the other sample considered. 

\begin{table}
    \centering
    \begin{tabular}{c| cccc}
            \toprule
        \textbf{Test} &   \makecell{\textbf{Original}\\\textbf{Data$^2$}} & \makecell{\textbf{CCC-op-ARCH$(1)$}\\\textbf{Residuals$^2$}} & \makecell{\textbf{CCC-op-ARCH$(5)$}\\\textbf{Residuals$^2$}} \\ 
            \midrule
        \\
        Maximal Lag $=$ 3 & <0.001 & <0.001 & 0.227\\
        Maximal Lag $=$ 10 & <0.001 & <0.001 & 0.591\\
          \bottomrule         
    \end{tabular}
    \caption{$p$-values of the white noise of \cite{Kokoszka:Rice:Shang:2017} applied to the squared residuals of models fit to the S\&P 500 data 2018--2020  data. The results were similar  2022--2024 sample.}
    \label{tab:whiteness}
\end{table}

\begin{figure}
    \centering
    \subfloat[SACF of original Data]{
        \includegraphics[width=0.4\textwidth,clip]{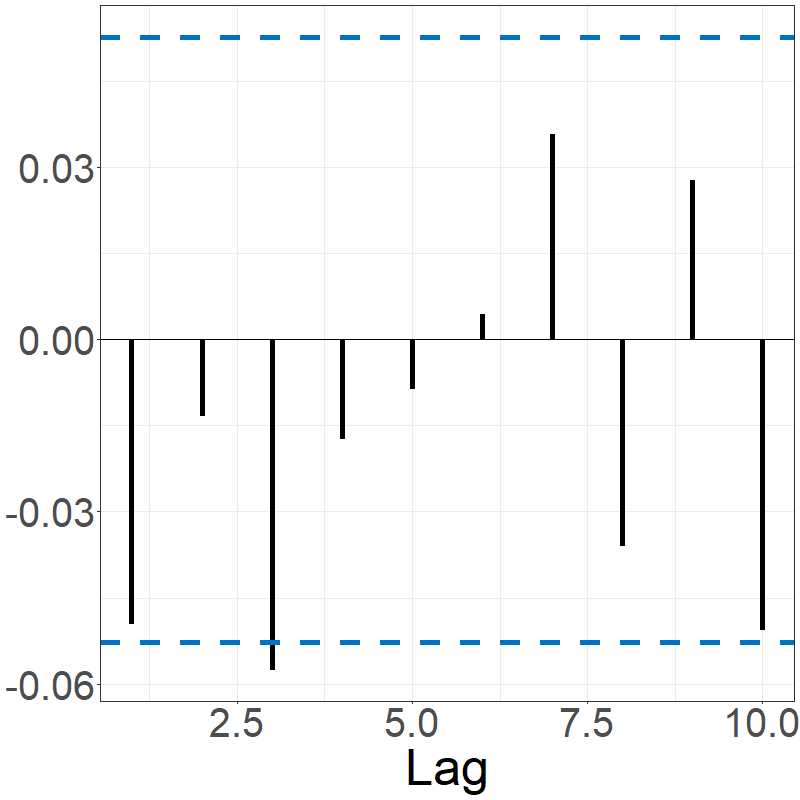}
    } \hfill
    \subfloat[SACF of original squared Data]{
        \includegraphics[width=0.4\textwidth,clip]{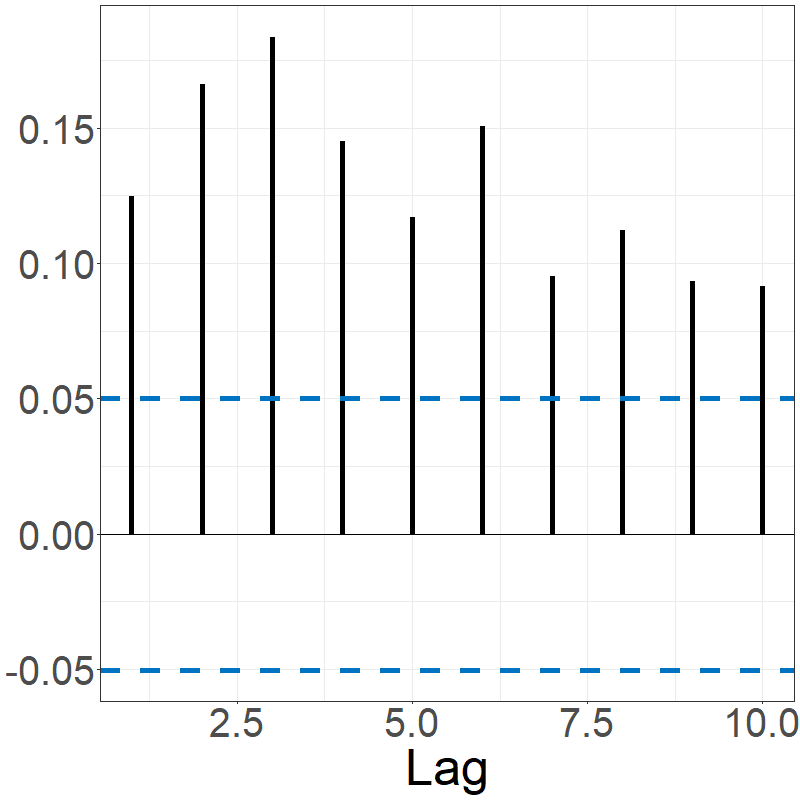}
    } \\
    \subfloat[CCC-op-ARCH$(1)$ model squared residuals]{
        \includegraphics[width=0.4\textwidth,trim={0cm 0cm 0cm 0cm},clip]{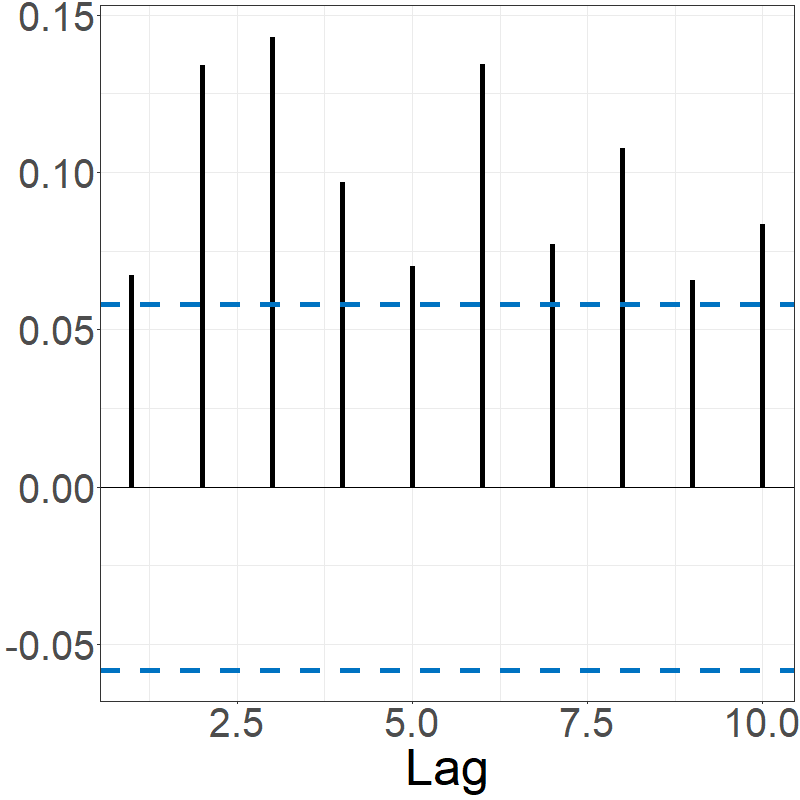}
    } \hfill
    \subfloat[CCC-op-ARCH$(5)$ model squared residuals]{
        \includegraphics[width=0.4\textwidth,clip]{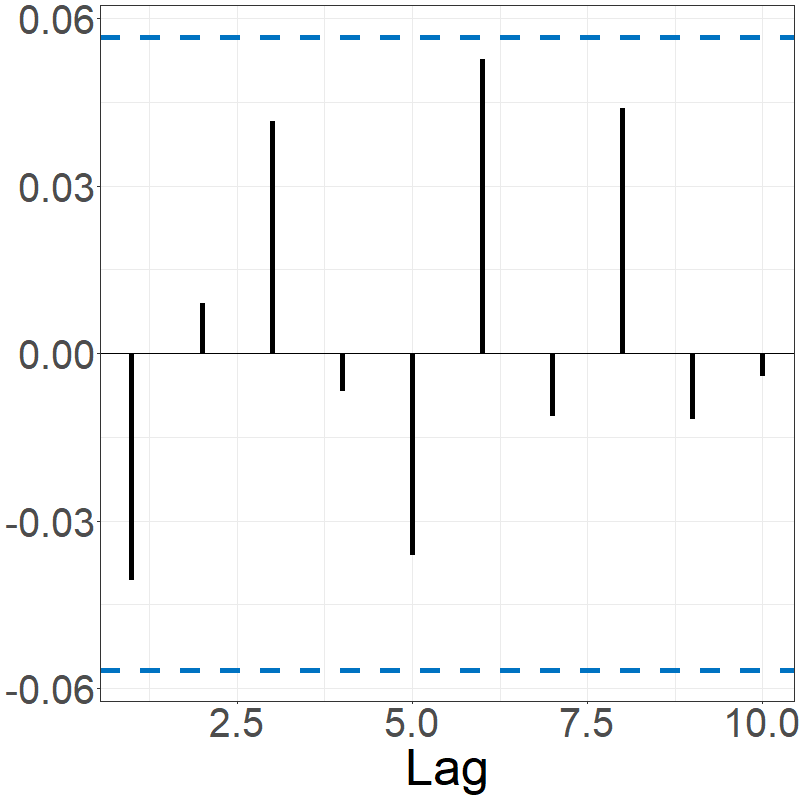}
    }    
    \caption{\textbf{SACF Plots}. Estimation of SACF for the original and CCC-op-ARCH models of the S\&P 500 data, 2018--2020.}
    \label{fig:white_noise}
\end{figure}

\section{Discussion}\label{Section Discussion}
This article introduces an ARCH model for processes taking values in general separable Hilbert spaces which we call operator-level ARCH (op-ARCH) models. A key advantage of this over previous functional conditional heteroscedasticity models is that it models the complete conditional covariance function, rather than only the pointwise variance. We establish sufficient conditions for strict stationarity. Weak stationarity and the existence of finite moments and weak dependence are also discussed. Consistent operator estimates are derived both in the finite- and infinite-dimensional setting via a Yule-Walker (YW) approach. An identifiability issue complicates the direct application of YW-type equations, even under Tikhonov regularization for ill-posedness. To address this, we propose a CCC-operator-level ARCH model, which permits consistent estimation via modified YW-type equations. In finite dimensions, parametric rates are achieved, while in infinite dimensions, explicit rates depending on eigenvalue decay and operator approximation are established. An example illustrating near-parametric rates in the infinite-dimensional case is given. After detailing several aspects of implementing the proposed methods, we present  results of Monte-Carlo simulation experiments, which suggest that the proposed estimators indeed appear to be consistent. In an application to cumulative intra-day return curves, the CCC-op-ARCH$(5)$ model appeared to perform well in explaining/modeling the observed heteroscedasticity in the curves, and provides alternative forecasts of the daily volatility of the curves when compared to existing pointwise models.  

The model may be extended to arbitrary separable Banach spaces, drawing on the estimation frameworks of \cite{RuizAlvarez2019} and \cite{DetteKokotAue2020}. Further, it would be valuable to generalize our estimation procedure to more general H-S operators. Finally, extending the theory to operator-valued GARCH processes appears promising, and work in this direction is ongoing.

Our attention was focused with regards to estimation in the ``diagonal situation'' \eqref{eq:diagARCH}. Although our preliminary investigations suggest that, as in the multivariate setting, the full ``VEC'' model in \eqref{eq:expansion of bold a_p for vech} is challenging to work with, other potential models are possible. These might include analogs of the BEKK, CCC, and DCC multivariate GARCH models, see  \citet[Ch.~10]{FrancqZakoian2019}. We leave this as a broad avenue for future research.   

\paragraph{Acknowledgements} Parts of the article were written while Sebastian K\"uhnert was employed at University of California, Davis. 

\paragraph{Funding} Alexander Aue was partially supported by NSF DMS 2515821. Sebastian K\"uhnert was partially supported by TRR 391 \textit{Spatio-temporal Statistics for the Transition of Energy and Transport} (Project number 520388526) funded by the Deutsche Forschungsgemeinschaft (DFG, German Research Foundation). Gregory Rice was partially supported by the Discovery Grant RGPIN 50503-11525 3100 105 from the Natural Science and Engineering Research Council of Canada. 

\bibliography{opARCHV2}

\begin{thebibliography}{}

\bibitem[\protect\citeauthoryear{Andersen, Davis, Kreiss, and Mikosch}{Andersen et~al.}{2009}]{AndersenEtAl2009}
Andersen, T., R.~Davis, J.-P. Kreiss, and T.~Mikosch (2009).
\newblock {\em Handbook of Financial Time Series.}
\newblock Berlin: Springer.

\bibitem[\protect\citeauthoryear{Andersen, Su, Todorov, and Zhang}{Andersen et~al.}{2024}]{Andersen02042024}
Andersen, T.~G., T.~Su, V.~Todorov, and Z.~Zhang (2024).
\newblock Intraday periodic volatility curves.
\newblock {\em Journal of the American Statistical Association\/}~{\em 119\/}(546), 1181--1191.

\bibitem[\protect\citeauthoryear{Asai, McAleer, and Yu}{Asai et~al.}{2006}]{mvgarch2}
Asai, M., M.~McAleer, and J.~Yu (2006).
\newblock Multivariate stochastic volatility: A review.
\newblock {\em Econometric Reviews\/}~{\em 25\/}(2-3), 145--175.

\bibitem[\protect\citeauthoryear{Aue, Horv\'ath, and Pellatt}{Aue et~al.}{2017}]{Aueetal2017}
Aue, A., L.~Horv\'ath, and D.~Pellatt (2017).
\newblock Functional generalized autoregressive conditional heteroskedasticity.
\newblock {\em Journal of Time Series Analysis\/}~{\em 38}, 3--21.

\bibitem[\protect\citeauthoryear{Bauwens, Laurent, and Rombouts}{Bauwens et~al.}{2006}]{mvgarch1}
Bauwens, L., S.~Laurent, and J.~V.~K. Rombouts (2006).
\newblock Multivariate {GARCH} models: A survey.
\newblock {\em Journal of Applied Econometrics\/}~{\em 21\/}(1), 79--109.

\bibitem[\protect\citeauthoryear{Billingsley}{Billingsley}{1995}]{Billingsley}
Billingsley, P. (1995).
\newblock {\em Probability and Measure\/} (3 ed.).
\newblock New York: Wiley.

\bibitem[\protect\citeauthoryear{Bollerslev}{Bollerslev}{1986}]{Bollerslev1986}
Bollerslev, T. (1986).
\newblock Generalized autoregressive conditional heteroskedasticity.
\newblock {\em Journal of Econometrics\/}~{\em 31}, 307--327.

\bibitem[\protect\citeauthoryear{Bosq}{Bosq}{2000}]{Bosq2000}
Bosq, D. (2000).
\newblock {\em Linear Processes in Function Spaces}.
\newblock Lecture Notes in Statistics. New York: Springer.

\bibitem[\protect\citeauthoryear{Bosq and Blanke}{Bosq and Blanke}{2007}]{BosqBlanke2007}
Bosq, D. and D.~Blanke (2007).
\newblock {\em Inference and Prediction in Large Dimensions}.
\newblock Wiley Series in Probability and Statistics. Chichester, UK: John Wiley \& Sons.

\bibitem[\protect\citeauthoryear{Cerovecki, Francq, H\"ormann, and Zako\"{i}an}{Cerovecki et~al.}{2019}]{Ceroveckietal2019}
Cerovecki, C., C.~Francq, S.~H\"ormann, and J.-M. Zako\"{i}an (2019).
\newblock Functional {GARCH} models: {T}he quasi-likelihood approach and its applications.
\newblock {\em Journal of Econometrics\/}~{\em 209\/}(2), 353--375.

\bibitem[\protect\citeauthoryear{Dette, Kokot, and Aue}{Dette et~al.}{2020}]{DetteKokotAue2020}
Dette, H., K.~Kokot, and A.~Aue (2020).
\newblock Functional data analysis in the banach space of continuous functions.
\newblock {\em Annals of Statistics\/}~{\em 48\/}(2), 1168--1192.

\bibitem[\protect\citeauthoryear{Engle}{Engle}{1982}]{Engle1982}
Engle, R. (1982).
\newblock Autoregressive conditional heteroskedasticity with estimates of the variance of {U}.{K}. inflation.
\newblock {\em Econometrica\/}~{\em 50\/}(4), 987--1008.

\bibitem[\protect\citeauthoryear{Francq and Zako\"{i}an}{Francq and Zako\"{i}an}{2019}]{FrancqZakoian2019}
Francq, C. and J.-M. Zako\"{i}an (2019).
\newblock {\em {GARCH} Models: {S}tructure, Statistical Inference and Financial Applications.\/} (2 ed.).
\newblock Chichester: John Wiley \& Sons Ltd.

\bibitem[\protect\citeauthoryear{Gouri\'eroux}{Gouri\'eroux}{1997}]{Gourieroux1997}
Gouri\'eroux, C. (1997).
\newblock {\em {ARCH} Models and Financial Applications.}
\newblock New York: Springer.

\bibitem[\protect\citeauthoryear{Hall and Meister}{Hall and Meister}{2007}]{hall:meister:2007}
Hall, P. and A.~Meister (2007).
\newblock A ridge-parameter approach to deconvolution.
\newblock {\em Annals of Statistics\/}~{\em 35\/}(4), 1535--1558.

\bibitem[\protect\citeauthoryear{H\"ormann, Horv\'ath, and Reeder}{H\"ormann et~al.}{2013}]{Hoermannetal2013}
H\"ormann, S., L.~Horv\'ath, and R.~Reeder (2013).
\newblock A functional version of the {ARCH} model.
\newblock {\em Econometric Theory\/}~{\em 29\/}(2), 267--288.

\bibitem[\protect\citeauthoryear{H\"ormann and Kokoszka}{H\"ormann and Kokoszka}{2010}]{HoermannKokoszka2010}
H\"ormann, S. and P.~Kokoszka (2010).
\newblock Weakly dependent functional data.
\newblock {\em Annals of Statistics\/}~{\em 38}, 1845--1884.

\bibitem[\protect\citeauthoryear{Horv\'ath and Kokoszka}{Horv\'ath and Kokoszka}{2012}]{HorvathKokoszka2012}
Horv\'ath, L. and P.~Kokoszka (2012).
\newblock {\em Inference for Functional Data with Applications}.
\newblock New York: Springer.

\bibitem[\protect\citeauthoryear{Hsing and Eubank}{Hsing and Eubank}{2015}]{HsingEubank2015}
Hsing, T. and R.~Eubank (2015).
\newblock {\em Theoretical Foundations of Functional Data Analysis, with an Introduction to Linear Operators}.
\newblock West Sussex: Wiley.

\bibitem[\protect\citeauthoryear{Kearney, Shang, and Zhao}{Kearney et~al.}{2023}]{KearneyShangZhao2023}
Kearney, F., H.~L. Shang, and Y.~Zhao (2023).
\newblock Intraday fx volatility-curve forecasting with functional {GARCH} approaches.
\newblock \url{https://arxiv.org/abs2311.18477}.

\bibitem[\protect\citeauthoryear{Kim, Kokoszka, and Rice}{Kim et~al.}{2023}]{Kim:Kokoszka:Rice:2023}
Kim, M., P.~Kokoszka, and G.~Rice (2023).
\newblock White noise testing for functional time series.
\newblock {\em Statistics {S}urveys\/}~{\em 17}, 119--168.

\bibitem[\protect\citeauthoryear{Kingman}{Kingman}{1973}]{Kingman1973}
Kingman, J. F.~C. (1973).
\newblock Subadditive ergodic theory.
\newblock {\em The Annals of Probability\/}~{\em 1\/}(6), 883--899.

\bibitem[\protect\citeauthoryear{Kokoszka, Rice, and Shang}{Kokoszka et~al.}{2017}]{Kokoszka:Rice:Shang:2017}
Kokoszka, P., G.~Rice, and H.~L. Shang (2017).
\newblock Inference for the autocovariance of a functional time series under conditional heteroscedasticity.
\newblock {\em Journal of {M}ultivariate {A}nalysis\/}~{\em 162}, 32--50.

\bibitem[\protect\citeauthoryear{K\"uhnert}{K\"uhnert}{2019}]{Kuehnert2019}
K\"uhnert, S. (2019).
\newblock {\em \"Uber funktionale {ARCH}- und {GARCH}-Zeitreihen}.
\newblock Ph.\ D. thesis, University of Rostock.

\bibitem[\protect\citeauthoryear{K\"uhnert}{K\"uhnert}{2020}]{Kuehnert2020}
K\"uhnert, S. (2020).
\newblock Functional {ARCH} and {GARCH} models: A {Y}ule-{W}alker approach.
\newblock {\em Electronic Journal of Statistics\/}~{\em 14\/}(2), 4321--4360.

\bibitem[\protect\citeauthoryear{K\"uhnert}{K\"uhnert}{2024}]{Kuehnert2024}
K\"uhnert, S. (2024).
\newblock Estimating lagged (cross-)covariance operators of \textnormal{$L^p$}-$m$-approximable processes in {C}artesian product {H}ilbert spaces.
\newblock {\em Journal of Time Series Analysis\/}~{\em 46\/}(3), 582--595.

\bibitem[\protect\citeauthoryear{K\"uhnert, Rice, and Aue}{K\"uhnert et~al.}{2026}]{KuehnertRiceAue2026}
K\"uhnert, S., G.~Rice, and A.~Aue (2026).
\newblock Estimating invertible processes in {H}ilbert spaces, with applications to functional {ARMA} processes.
\newblock {\em Bernoulli\/}~{\em 32\/}(2), 1523--1546.

\bibitem[\protect\citeauthoryear{Kutta, Dierickx, and Dette}{Kutta et~al.}{2022}]{KuttaDierickxDette2022}
Kutta, T., G.~Dierickx, and H.~Dette (2022).
\newblock {Statistical inference for the slope parameter in functional linear regression}.
\newblock {\em Electronic Journal of Statistics\/}~{\em 16\/}(2), 5980--6042.

\bibitem[\protect\citeauthoryear{Laksaci, Alshahrani, Almanjahie, and Kaid}{Laksaci et~al.}{2025}]{Laksacietal2025}
Laksaci, A., F.~Alshahrani, I.~M. Almanjahie, and Z.~Kaid (2025).
\newblock Nonparametric multifunctional {GARCH} time series data analysis: Application to dynamic forecasting in financial data.
\newblock {\em AIMS Mathematics\/}~{\em 10\/}(11), 26459--26483.

\bibitem[\protect\citeauthoryear{Li, Sun, and Liu}{Li et~al.}{2025}]{LiSunLiu2025}
Li, Z., H.~Sun, and J.~Liu (2025).
\newblock A functional {GARCH} model with multiple constant parameters. \textit{Computational Economics}.

\bibitem[\protect\citeauthoryear{Liggett}{Liggett}{1985}]{Liggett1985}
Liggett, T. (1985).
\newblock An improved subadditive ergodic theorem.
\newblock {\em Annals of Probability\/}~{\em 13\/}(4), 1279--1285.

\bibitem[\protect\citeauthoryear{Mas}{Mas}{2007}]{Mas2007}
Mas, A. (2007).
\newblock Weak convergence in the functional autoregressive model.
\newblock {\em Journal of Multivariate Analysis\/}~{\em 98\/}(6), 1231--1261.

\bibitem[\protect\citeauthoryear{Moore}{Moore}{1920}]{moore1920}
Moore, E.~H. (1920).
\newblock On the reciprocal of the general algebraic matrix.
\newblock {\em Bulletin of the American Mathematical Society\/}~{\em 26}, 294--295.

\bibitem[\protect\citeauthoryear{Penrose}{Penrose}{1955}]{penrose1955}
Penrose, R. (1955).
\newblock A generalized inverse for matrices.
\newblock {\em Mathematical Proceedings of the Cambridge Philosophical Society\/}~{\em 51\/}(3), 406–413.

\bibitem[\protect\citeauthoryear{{R Core Team}}{{R Core Team}}{2020}]{r}
{R Core Team} (2020).
\newblock {\em R: A Language and Environment for Statistical Computing}.
\newblock Vienna, Austria: R Foundation for Statistical Computing.

\bibitem[\protect\citeauthoryear{Reimherr}{Reimherr}{2015}]{Reimherr2015}
Reimherr, M. (2015).
\newblock Functional regression with repeated eigenvalues.
\newblock {\em Statistics \& Probability Letters\/}~{\em 107}, 62--70.

\bibitem[\protect\citeauthoryear{Rice, Wirjanto, and Zhao}{Rice et~al.}{2023}]{RiceEtAL2023}
Rice, G., T.~Wirjanto, and Y.~Zhao (2023).
\newblock Exploring volatility of crude oil intraday return curves: {A} functional {GARCH}-{X} model.
\newblock {\em Journal of Commodity Markets\/}~{\em 32}, 100361.

\bibitem[\protect\citeauthoryear{Ruiz-Medina and \'Alvarez-Li\'ebana}{Ruiz-Medina and \'Alvarez-Li\'ebana}{2019}]{RuizAlvarez2019}
Ruiz-Medina, M. and J.~\'Alvarez-Li\'ebana (2019).
\newblock Strongly consistent autoregressive predictors in abstract {B}anach spaces.
\newblock {\em Journal of Multivariate Analysis\/}~{\em 170}, 186--201.

\bibitem[\protect\citeauthoryear{Sun and Yu}{Sun and Yu}{2020}]{SunYu2020}
Sun, H. and B.~Yu (2020).
\newblock Volatility asymmetry in functional threshold {GARCH} model.
\newblock {\em Journal of Time Series Analysis\/}~{\em 41\/}(1), 95--109.

\bibitem[\protect\citeauthoryear{Weidmann}{Weidmann}{1980}]{Weidmann1980}
Weidmann, J. (1980).
\newblock {\em Linear Operators in Hilbert spaces}.
\newblock Graduate Texts in Mathematics, 68. New York-Berlin: Springer.

\bibitem[\protect\citeauthoryear{Wu and Shao}{Wu and Shao}{2004}]{WuShao2004}
Wu, W. and X.~Shao (2004).
\newblock Limit theorems for iterated random functions.
\newblock {\em Journal of Applied Probability\/}~{\em 41\/}(2), 425--436.

\bibitem[\protect\citeauthoryear{Yeh, Rice, and Dubin}{Yeh et~al.}{2023}]{yeh:rice:2023}
Yeh, C.-K., G.~Rice, and J.~A. Dubin (2023).
\newblock {Functional spherical autocorrelation: A robust estimate of the autocorrelation of a functional time series}.
\newblock {\em Electronic Journal of Statistics\/}~{\em 17\/}(1), 650 -- 687.

\end{thebibliography}

\appendix

\section{Preliminaries}\label{Appendix: Section preliminaries}

This section presents a few  fundamental features required for our proofs. A compact operator $A:\mathcal{H}\to\mathcal{H}_\star,$ where $\mathcal{H}$ and $\mathcal{H}_\star$ are separable Hilbert spaces,  belongs to the \emph{Schatten class of order} $1\leq p < \infty,$ denoted by $\mathcal{S}_{p,\mathcal{H},\mathcal{H}_\star}$ if its singular values $s_1(A) \geq s_2(A) \geq \dots \geq 0$ are $p$-summable. The class $\mathcal{S}_{p,\mathcal{H},\mathcal{H}_\star}$ is equipped with the norm
$$
    \|A\|_p \coloneqq \bigg(\sum_{i=1}^\infty s_i^p(A)\bigg)^{\!1/p}\!,\quad A \in \mathcal{S}_{p,\mathcal{H},\mathcal{H}_\star}.
$$
The spaces of nuclear and H-S operators are $\mathcal{N}_{\mathcal{H},\mathcal{H}_\star} = \mathcal{S}_{1,\mathcal{H},\mathcal{H}_\star}$ and $\mathcal{S}_{\mathcal{H},\mathcal{H}_\star}=\mathcal{S}_{2,\mathcal{H},\mathcal{H}_\star}$, respectively, and the space of bounded linear operators $\mathcal{L}_{\mathcal{H},\mathcal{H}_\star}$ is denoted by $\mathcal{S}_{\infty,\mathcal{H},\mathcal{H}_\star}$, with norm $\|\cdot\|_\infty \coloneqq \|\cdot\|_{\mathcal{L}}.$ It is well known that $\mathcal{S}_{q,\mathcal{H},\mathcal{H}_\star} \subsetneq \mathcal{S}_{p,\mathcal{H},\mathcal{H}_\star}$ and $\|\cdot\|_p \leq \|\cdot\|_q$ for all $1\leq p < q \leq \infty.$ These operators fulfill the following H\"older-type inequality \citep[Section~7]{Weidmann1980}. 

\begin{lemma}\label{Lemma: operator-valued hoelder inquality} 
Let $\mathcal{H}_1, \mathcal{H}_2, \mathcal{H}_3$ be Hilbert spaces, and let $p,q,r \in [1,\infty]$ with $\tfrac{1}{p} + \tfrac{1}{q} = \tfrac{1}{r}$, and $\tfrac{1}{\infty}\coloneqq 0.$ Then, for any $A\in\mathcal{S}_{p,\mathcal{H}_2,\mathcal{H}_3}$ and $B\in\mathcal{S}_{q,\mathcal{H}_1,\mathcal{H}_2},$ it holds that $AB\in\mathcal{S}_{r,\mathcal{H}_1,\mathcal{H}_3}$, with
\[
    \|AB\|_r \leq \|A\|_p \, \|B\|_q.
\]
\end{lemma}

The next auxiliary result is a direct consequence of the properties of the given operators and norms, and we therefore omit the proof.

\begin{lemma}\label{lem:operator-valued matrix inequalities} 
Let $m, n \in \mathbb{N},$ suppose $B_j$ are bounded linear operators between Banach spaces, and let 
$C_{ij}$ be H-S operators for $1 \le i \le m$, $1 \le j \le n$. Then, the following holds:
\begin{itemize}
    \item[\textnormal{(a)}] The vector $(B_1\,\cdots\,B_n)$ is a bounded linear operator between Banach spaces, with  
    \begin{align}\label{inequality for the matrix norms}
      \|(B_1\,\cdots\,B_n)\|_{\mathcal{L}} \leq \sum_{i=1}^n \|B_i\|_{\mathcal{L}}\,.
    \end{align}  
    
    \item[\textnormal{(b)}] The matrix $\boldsymbol{C} \coloneqq (C_{ij})^{m,n}_{i,j=1}$ is H-S, with 
    \begin{align}
      \|\boldsymbol{C}\|^2_{\mathcal{S}} = \sum_{i=1}^m \sum_{j=1}^n \|C_{ij}\|^2_{\mathcal{S}}\,.
    \end{align} 
\end{itemize}
\end{lemma}

\section{Notes}\label{Appendix: Section Notes}

In the following, we illustrate the validity of the Markovian forms in the op-ARCH and the CCC-op-ARCH model.

\begin{example}\label{Ex:Markovianforms} Let $p = 2.$
\begin{itemize}
    \item[\textnormal{(a)}] From the definition of $\Sigma_{0}$ in the op-ARCH equation \eqref{eq:definition op-ARCH} and the elements introduced in Section \ref{Sec:General model}, it follows that
    \begin{gather*}
        \boldsymbol{\Delta} + \boldsymbol{\Psi}\big(X^{\otimes2,[2]}_{\!-1}\big) \;=\; 
        \begin{bmatrix}
            \Delta\\
            0
        \end{bmatrix} 
        + 
        \begin{bmatrix}
    		\alpha_1\!   & \!\alpha_2\\
    		\mathbb{I}\! & \!0\\
    	\end{bmatrix}
        \!\!
        \begin{bmatrix}
            X^{\otimes 2}_{-1}\\
            X^{\otimes 2}_{-2}
        \end{bmatrix}
        \;=\;
        \begin{bmatrix}
            \Delta + \sum^3_{i=1}\alpha_i(X^{\otimes 2}_{-i})\\
            X^{\otimes 2}_{-1}
        \end{bmatrix}
        \;=\;
        \begin{bmatrix}
            \Sigma_{0} \\
            X^{\otimes 2}_{-1}
        \end{bmatrix}. 
    \end{gather*}
    Consequently, it holds indeed  
    \begin{gather*}
        \Upsilon_{\!0}\big(\boldsymbol{\Delta} + \boldsymbol{\Psi}\big(X^{\otimes2,[2]}_{\!-1}\big)\big) 
        \;=\; 
        \begin{bmatrix}
            \Sigma^{1/2}_{0}\varepsilon^{\otimes 2}_{0}\Sigma^{1/2}_{0}\\
            X^{\otimes 2}_{-1}
        \end{bmatrix}
        \;=\; 
        \begin{bmatrix}
            X^{\otimes 2}_{0}\\
            X^{\otimes 2}_{-1}
        \end{bmatrix} \;=\; X^{\otimes2,[2]}_0,
    \end{gather*}
    which is the Markovian form \eqref{quasi-state-space form} of the general op-ARCH$(2)$ model.
    
        \item[\textnormal{(b)}] Let $(X_k)$ be the CCC-op-ARCH$(2)$ process with operators
        \[
            \Delta = \sum^\infty_{j=1} d_j (e_j \otimes e_j), \quad \alpha_i = \sum^\infty_{j=1} a_{ijj}\,(e_j \otimes e_j)\otimes_{\mathcal{S}}(e_j \otimes e_j), \quad i=1,2,
        \]
        where $(d_j)$, $(a_{1jj})_j$, and $(a_{2jj})_j$ are square-summable sequences. By \eqref{eq:definition op-ARCH}, and using
    \[
        \varepsilon_0 = \sum^\infty_{\ell=1}\,\langle \varepsilon_0, e_\ell\rangle e_\ell, \quad X^{\otimes 2}_{n,\mathrm{d}} = \sum^\infty_{\ell=1}\,\langle X_n, e_\ell\rangle^2 (e_\ell\otimes e_\ell),
    \]
    we obtain for any $m$,
    \begin{align*}
    \langle X_0, e_m\rangle^2
    &=
    \langle \varepsilon_0, e_m\rangle^2\,
    \big\langle \Sigma_0^{1/2}(e_m), e_m \big\rangle^2 \\
    &=
    \langle \varepsilon_0, e_m\rangle^2
    \Big[ d_m + a_{1mm}\langle X_{-1}, e_m \rangle^2 
          + a_{2mm}\langle X_{-2}, e_m \rangle^2 \Big].
    \end{align*}
    Consequently, by the definition of $\psi_0$ in \eqref{sk8} and $\alpha_i^0 = \psi_0(\alpha_i(\cdot))$, $i=1,2$, we obtain
    \[
    X^{\otimes 2,[2]}_{0,\mathrm{d}}
    \;=\;
    \begin{bmatrix}
    X^{\otimes 2}_{0,\mathrm{d}}\\[0.3em]
    X^{\otimes 2}_{-1,\mathrm{d}}
    \end{bmatrix}
    \;=\;
    \begin{bmatrix}
    \psi_0(\Delta)\\[0.3em]
    0
    \end{bmatrix}
    +
    \begin{bmatrix}
    \alpha_1^0 & \alpha_2^0\\
    \mathbb{I} & 0
    \end{bmatrix}
    \begin{bmatrix}
    X^{\otimes 2}_{-1,\mathrm{d}}\\[0.3em]
    X^{\otimes 2}_{-2,\mathrm{d}}
    \end{bmatrix}
    \;=\;
    \boldsymbol{\Delta}_0
    +
    \boldsymbol{\Psi}_{\!0}\big(X^{\otimes 2,[2]}_{-1,\mathrm{d}}\big),
    \]
    which confirms the Markovian representation of the CCC-op-ARCH$(2)$ model for $p=2$.
\end{itemize}
\end{example}

The estimates for $\mathscr{C}$ and $\mathscr{D}$ are based on a sample $X_1, \dots ,X_N$ of $\boldsymbol{X}=(X_k)\subset\mathcal{H},$ with sample size $N,$ defined by
\begin{gather}
    \hat{\mathscr{C}}\coloneqq \frac{1}{N}\sum_{k=p}^{N}\,\big(X^{\otimes2, [p]}_{\!k}\! - \hat{m}_p\big)\otimes_\mathcal{S}\big(X^{\otimes2, [p]}_{\!k} - \hat{m}_p\big), 
    \label{eq:est lag-0-cov-op}\\      
    \hat{\mathscr{D}}\coloneqq \frac{1}{N}\sum_{k=p}^{N-1}\,\big(X^{\otimes2, [p]}_{\!k} - \hat{m}_p\big)\otimes_\mathcal{S}\bigg[\big(X^{\otimes2}_{\!k+1} - \hat{m}'_1\big)\mathscr{C}^\dagger_{\!\boldsymbol{\varepsilon}}\!\coprod^{e_K}_{e_1}\,\bigg],\label{eq:est lag-1-cross-cov-op}
\end{gather}
respectively, with first moment estimates 
\begin{gather}\label{eq:Def moment estimates}
  \hat{m}_p\coloneqq\frac{1}{N}\sum_{\ell=p}^{N}X^{\otimes2,[p]}_{\ell}, \qquad \hat{m}'_1\coloneqq\frac{1}{N}\sum_{\ell=p}^{N-1}X^{\otimes2}_{\ell+1}.   
\end{gather}

\section{Proofs}\label{Section proofs}

Proving our results requires the consistent estimation of means and (lagged) (cross-)covariance operators. To this end, we formulate the following auxiliary results.

\begin{lemma}\label{lemma:Estimation 2nd moments} Under the conditions of Proposition  \ref{prop:Explicit moments} for $\nu=2$, for $\hat{m}_p$ in \eqref{eq:Def moment estimates} holds
\begin{align*}
    \Exp\!\|\hat{m}_p  - m_p\|^2_{\mathcal{S}} = \mathrm O(N^{-1}).
\end{align*}
\end{lemma}

\begin{proof}[\textbf{Proof}] The claim follows from $m_p=(m_1, \dots, m_1)^\top\!,$ $\Exp\!\|\hat{m}_p  - m_p\|^2_{\mathcal{S}} = p\Exp\!\|\hat{m}_1  - m_1\|^2_{\mathcal{S}},$ and $m_1=\mathscr{C}_{\!\boldsymbol{X}},$ together with $L^2$-$m$-approximability of $\boldsymbol{X}^{\otimes 2}$ by Proposition \ref{prop:Explicit moments}, and \citet[Theorem 3.1]{Kuehnert2024}.
\end{proof}

\begin{lemma}\label{Lemma: Cons lagged cross-cov non-centered} Let the conditions of Proposition \ref{prop:Explicit moments} for $\nu=4$ hold. Then, for the estimators $\hat{\mathscr{C}}$ and $\hat{\mathscr{D}}$ in \eqref{eq:est lag-0-cov-op}--\eqref{eq:est lag-1-cross-cov-op} holds 
\begin{gather*}
    \|\hat{\mathscr{C}} - \mathscr{C}\|_{\mathcal{S}} = \mathrm O_{\Prob}(N^{-1/2}), \qquad \|\hat{\mathscr{D}} - \mathscr{D}\|_{\mathcal{S}} = \mathrm O_{\Prob}\big(a^{-2}_KN^{-1/2}\big).
\end{gather*}
\end{lemma}

\begin{proof}[\textbf{Proof}] Due to $L^4$-$m$-approximability of $(X^{\otimes 2}_{\!k})$ by Proposition \ref{prop:Explicit moments}, which transfers to the process $(X^{\otimes 2,[p]}_k)_k$, the claim holds for the estimator $\hat{\mathscr{C}}'$ based on the true first moment $m_p$ instead of $\hat{m}_p$ \citep[Theorem 3.1]{Kuehnert2024}. Moreover, according to the definitions of $\hat{\mathscr{C}}',$ $\hat{\mathscr{C}}$ and $X^{\otimes 2, [p]}_k$, the identity $(a-b)\otimes(c-d) = a\otimes c - a\otimes d - b\otimes c + b\otimes d$, and $\|a\otimes b\|_{\mathcal{S}} = \|a\|\|b\|$ along with the Cauchy--Schwarz inequality, we obtain
\begin{align*}
    \|\hat{\mathscr{C}} - \hat{\mathscr{C}}'\|_\mathcal{S} 
    &\leq \frac 1 N \left(\sum^N_{k=p}\, \|\hat{m}_p - m_p\|_{\mathcal{S}}\left[2\big\|X^{\otimes 2, [p]}_k\big\|_{\mathcal{S}} + \|\hat{m}_p\|_{\mathcal{S}} + \|m_p\|_{\mathcal{S}}\right]\right)\\
    &= \|\hat{m}_p - m_p\|_{\mathcal{S}}\left[\frac 2 N \left(\sum^N_{k=p}\left(\sum^p_{\ell=1}\|X_{k+1-\ell}\|^2\right)^{\!1/2}\right) + \|\hat{m}_p\|_{\mathcal{S}} + \|m_p\|_{\mathcal{S}}\right].
\end{align*}
Further, since $\|X_{k+1-\ell}\| = \mathrm{O}_{\Prob}(1)$ for all $k,\ell,$ $\|m_p\|^2_{\mathcal{S}}=p\|m_1\|^2_{\mathcal{S}},$ and $\|\hat{m}_p - m_p\|^2_{\mathcal{S}}=\mathrm{O}_{\Prob}(N^{-1})$ by Lemma \ref{lemma:Estimation 2nd moments}, we conclude $\|\hat{\mathscr{C}} - \hat{\mathscr{C}}'\|_\mathcal{S} = \mathrm{O}_{\Prob}(N^{-1/2}).$ Hence, in fact 
$$
    \|\hat{\mathscr{C}} - \mathscr{C}\|_{\mathcal{S}} \;\leq\; \|\hat{\mathscr{C}} - \hat{\mathscr{C}}'\|_{\mathcal{S}} + \|\hat{\mathscr{C}}' - \mathscr{C}\|_{\mathcal{S}} \;=\; \mathrm{O}_{\Prob}(N^{-1/2}).
$$ 
We now turn to $\mathscr{D}$. The $L^4$-$m$-approximability of $(X^{\otimes 2}_k)$ carries over to the process $(X^{\otimes 2}_{k,\boldsymbol{\varepsilon}})_k$, where $X^{\otimes2}_{k,\boldsymbol{\varepsilon}} = X^{\otimes2}_k\mathscr{C}^\dagger_{\!\boldsymbol{\varepsilon}}\!\coprod_{e_1}^{e_K}$, since $\mathscr{C}^\dagger_{\!\boldsymbol{\varepsilon}}\!\coprod_{e_1}^{e_K}$ is deterministic for any $K,N$, and
$$
    \bigg\|\mathscr{C}^\dagger_{\!\boldsymbol{\varepsilon}}\!\coprod_{e_1}^{e_K}\bigg\|_{\mathcal{L}} \;=\; \sup_{1\leq j\leq K} (a_j + \vartheta_{\!N})^{-1} = (a_K + \vartheta_{\!N})^{-1} \;\leq\; a_K^{-1}.
$$
Following the proof of \citet[Theorem 3.1]{Kuehnert2024}, and using sub-multiplicativity of the operator norm, we find that for the estimator $\hat{\mathscr{D}}'$ based on $m_1$ instead of $\hat{m}_1,$ it holds  
$$
    \|\hat{\mathscr{D}}' - \mathscr{D}\|_{\mathcal{S}} = \mathrm{O}_{\Prob}\big(a_K^{-2} N^{-1/2}\big),
$$
and by arguments analogous to above, the same rate holds for $\|\hat{\mathscr{D}} - \mathscr{D}\|_{\mathcal{S}}$ as claimed.
\end{proof}

\begin{proof}[\textbf{Proof of Proposition \ref{prop:existence gamma}}] The existence of the top Lyapunov exponent $\gamma$ follows from Theorem~1.10 of \cite{Liggett1985}, whose assumptions must be verified. To this end, define
\begin{align*}
    X_{m,n}\coloneqq\ln\tau\big(\boldsymbol{\Psi}_{\!n-1}\circ\boldsymbol{\Psi}_{\!n-2}\circ \,\cdots\, \circ\boldsymbol{\Psi}_{\!m}\big), \quad 0\le m<n,
\end{align*}
where $\tau$ denotes the functional in \eqref{eq:tau} and $\boldsymbol{\Psi}_{\!\ell}$ the maps in the Markovian representation \eqref{sk2} of the CCC-op-ARCH$(p)$ model. 

By the definition of $\boldsymbol{\Psi}_{\!0},$ we obtain for any $A=(A_1,\dots,A_p)^\top\in\mathcal{S}^p_{\ge0,\mathrm{d}}\colon$
\begin{align*}
    \tau^2(\boldsymbol{\Psi}_{\!0}) &= \sup_{\|A\|_{\mathcal{S}^p}\le1,\,A\in\mathcal{S}^p_{\ge0,\mathrm{d}}} 
    \bigg\|\sum_{i=1}^p\psi_0(\alpha_i(A_i))\bigg\|^2_\mathcal{S} \,+\, \sum_{i=1}^{p-1}\|A_i\|^2\\
    &\le 1 \,+\, \sup_{\|A\|_{\mathcal{S}^p}\le1,\,A\in\mathcal{S}^p_{\ge0,\mathrm{d}}}
    \bigg\|\sum_{i=1}^p\psi_0(\alpha_i(A_i))\bigg\|^2_\mathcal{S}\,.
\end{align*}
Further, by the definition of $\psi_0$ and $\boldsymbol{\alpha} =(\alpha_1~\cdots~\alpha_p) : \mathcal{S}^p\to\mathcal{S}$ in \eqref{eq:diagARCH}, elementary conversions yield
\begin{align*}
    \bigg\|\sum_{i=1}^p\psi_0(\alpha_i(A_i))\bigg\|_\mathcal{S}
    \le\, \|\varepsilon_0\|^2\|\boldsymbol{\alpha}(A)\|_\mathcal{S}
    \,\leq\, \|\varepsilon_0\|^2\|\boldsymbol{\alpha}\|_{\mathcal{L}}\|A\|_{\mathcal{S}}\,.
\end{align*}
Hence, using $\ln(1+x)\le x$ for $x\ge0,$ and since $\varepsilon_0$ has finite second moments, we get
\begin{align}\label{eq:Lig1}
  \Exp\!\big(\max(0,X_{0,1})\big)
  \;\le\; \|\boldsymbol{\alpha}\|_{\mathcal{L}}
  \Exp\!\|\varepsilon_0\|^2 ~<~ \infty.
\end{align}

\noindent Moreover, sub-multiplicativity of $\tau$ implies sub-additivity in the sense that 
\begin{align}
    X_{0,n} \le X_{m,n} + X_{0,m}, \quad 0<m<n,\label{eq:Lig2}
\end{align}
and since all factors are i.i.d., we have
\begin{gather}
    (X_{m,m+k})_{k\in\mathbb{N}} \,\stackrel{d}{=}\, (X_{m+1,m+k+1})_{k\in\mathbb{N}}, \qquad m\ge1,\label{eq:Lig3}\\
    (X_{nk,(n+1)k})_{n\in\mathbb{N}} \text{ is strictly stationary for each } k\ge1.\label{eq:Lig4}
\end{gather}
By substituting \eqref{eq:Lig1} for (1.3) in \cite{Liggett1985}, and noting that \eqref{eq:Lig2}--\eqref{eq:Lig4} correspond to their conditions (1.7)--(1.9), respectively, all assumptions of their Theorem~1.10 are satisfied  (see also their remarks on p.~1280). Thus, $\gamma$ in \eqref{Definition top Lyapunov exponent} exists, and the identity stated there holds. Finally, \eqref{sk6} follows directly from Fekete’s sub-additive lemma.
\end{proof}

\begin{proof}[\textbf{Proof of Theorem \ref{theo:strict_stat}}] The proof is inspired by \cite{Ceroveckietal2019}. Multiple applications of the Markovian form \eqref{sk2} yield, by linearity of $\boldsymbol{\Psi}_{\!k}$,
\begin{align}\label{iterative application of the multivariate representation}
	X^{\otimes2,[p]}_{\!k,\mathrm{d}}
	= \boldsymbol{\Delta}_k 
	+ \sum_{\ell=1}^\infty \boldsymbol{\Psi}_{\!k,\ell}(\boldsymbol{\Delta}_{k-\ell}),
\end{align}
provided the series convergences almost surely (a.s.), and where $\boldsymbol{\Psi}_{\!k,\ell}$ is defined as a composition of $\ell$ factors through
\begin{align}\label{sk10}
    \boldsymbol{\Psi}_{\!k,\ell} \coloneqq \boldsymbol{\Psi}_{\!k}\circ\boldsymbol{\Psi}_{\!k-1}\circ\cdots\circ\boldsymbol{\Psi}_{\!k-\ell+1}, \quad \ell\in\mathbb{N}.
\end{align}
By the definition of $\boldsymbol{\Delta}_{k-\ell}$ in \eqref{sk2}–\eqref{sk8}, using $\Delta=\sum_j d_j(e_j\otimes e_j)$ with positive square-summable sequence $(d_j)$ and non-degeneracy of the innovations by injectivity of the covariance operator (Assumption~\ref{as:C_eps known and injective}), we have $\|\boldsymbol{\Delta}_{k-\ell}\|^2_{\mathcal{S}} = \sum_j \langle\varepsilon_{k-\ell},e_j\rangle^4 d_j^2>0~\text{a.s.},$ thus $\Exp\!\|\boldsymbol{\Delta}_{k-\ell}\|_{\mathcal{S}} \le \|\Delta\|_{\mathcal{L}} \Exp\!\|\varepsilon_{k-\ell}\|^2 <\infty$, and in turn $\ln\|\boldsymbol{\Delta}_{k-\ell}\|_{\mathcal{S}}\in (-\infty,\infty)$~a.s. Subsequently, $\frac{1}{\ell}\ln\|\boldsymbol{\Delta}_{k-\ell}\|_{\mathcal{S}}\to 0$ a.s.~as $\ell\to\infty.$ By the definition of $\gamma$ in \eqref{Definition top Lyapunov exponent}, due to properties of $\tau$, and as $(\boldsymbol{\Psi}_{\!k})$ is i.i.d., it holds
\begin{align*} 
    \limsup_{\ell \rightarrow \infty}\frac{1}{\ell}\ln\big\|\boldsymbol{\Psi}_{\!k,\ell}(\boldsymbol{\Delta}_{k-\ell})\big\|_{\mathcal{S}} &\leq \,\limsup_{\ell \rightarrow \infty}\frac{1}{\ell}\ln\tau(\boldsymbol{\Psi}_{\!k,\ell}) \,+\, \limsup_{\ell \rightarrow \infty}\frac{1}{\ell}\ln\!\|\boldsymbol{\Delta}_{k-\ell}\|_{\mathcal{S}}\\
    &= \gamma \quad\mbox{a.s.}
\end{align*}
Hence, by \eqref{sk5},
\[
    \limsup_{\ell\to\infty}\big\|\boldsymbol{\Psi}_{\!k,\ell}(\boldsymbol{\Delta}_{k-\ell})\big\|_{\mathcal{S}}^{1/\ell} \;\le\; e^{\gamma} \;<\; 1 \quad \text{a.s.},
\]
so the series in \eqref{iterative application of the multivariate representation} converges almost surely by Cauchy’s rule, which guarantees the existence of a solution $(X^{\otimes2,[p]}_{\!k,\mathrm{d}})_k$ to \eqref{sk2} for the CCC-op-ARCH$(p)$ model. By causality and \citet[Theorem~36.4]{Billingsley}, this solution is strictly stationary,  transferring to $(\Sigma_k)$ as  $\Sigma_k\in\sigma(X^{\otimes2,[p]}_{\!k-1,\mathrm{d}})$, and to $(X_k)$ via $X_k=\Sigma_k^{1/2}(\varepsilon_k)$.

To show almost sure uniqueness of the solution, let $(Y_k)$ be another solution of \eqref{sk2}. For each $k$ and $N$,
$$
    Y_k = X^{\otimes2,[p]}_{\!k,\mathrm{d},N} + \boldsymbol{\Psi}_{\!k,N}\big(X^{\otimes2,[p]}_{\!k,\mathrm{d},N-1}\big), \quad\mbox{ where~} X^{\otimes2,[p]}_{\!k,\mathrm{d},N} \coloneqq \boldsymbol{\Delta}_k + \sum_{\ell=1}^N\boldsymbol{\Psi}_{\!k,N}(\boldsymbol{\Delta}_{k-\ell}).
$$
Subsequently, due to the fact that $\gamma<0$ implies (as $N\to\infty$)
\[
    \big\|X^{\otimes2,[p]}_{\!k,\mathrm{d},N} - X^{\otimes2,[p]}_{\!k,\mathrm{d}}\big\|_{\mathcal{S}}\to 0\quad \text{and} \quad \tau(\boldsymbol{\Psi}_{\!k,N})\to 0 \quad\text{a.s.},
\]
we obtain
\begin{align*}
    \big\|Y_k - X^{\otimes2,[p]}_{\!k,\mathrm{d}}\big\|_{\mathcal{S}} \;\le\; \big\|X^{\otimes2,[p]}_{\!k,\mathrm{d},N} - X^{\otimes2,[p]}_{\!k,\mathrm{d}}\big\|_{\mathcal{S}} \,+\, \tau(\boldsymbol{\Psi}_{\!k,N})\|Y_{k-N-1}\|_{\mathcal{S}} ~\stackrel{N \to\infty}{\longrightarrow}~ 0 \quad\text{a.s.}
\end{align*}
Thus, as the distribution of $\|Y_{k-N-1}\|_{\mathcal{L}}$ is independent of $N$, the solution is almost surely unique.
\end{proof}

\begin{proof}[\textbf{Proof of Proposition \ref{prop: explciti suff conditions for stat sol}}] For $n\in \mathbb{N}$ and $\nu >0,$ with $\boldsymbol{\Psi}_{\!n,n} =\boldsymbol{\Psi}_{\!n}\circ\boldsymbol{\Psi}_{\!n-1}\circ\cdots\circ\boldsymbol{\Psi}_{\!1},$ and using the representation \eqref{sk6}, and Jensen’s inequality, it follows that
\begin{align*} 
    n\nu\gamma \;\leq\; \nu\Exp\ln\tau(\boldsymbol{\Psi}_{\!n,n}) \;=\; \Exp\ln\tau^\nu(\boldsymbol{\Psi}_{\!n,n}) \;\leq\; \ln\Exp\tau^\nu(\boldsymbol{\Psi}_{\!n,n})\,.
\end{align*}
Hence, \eqref{eq:exp opGARCH smaller 1} implies \eqref{sk5}.

Next, we show that \eqref{easier suff cond exp smaller 1} implies \eqref{eq:exp opGARCH smaller 1}. For illustrative purposes, first consider the case $p=3.$ Then, for the composition of the operator-valued matrices $\boldsymbol{\Psi}_{\!k}$ in the Markovian form via the functions $\psi_k,$ with $\alpha_{k,i} \coloneqq \psi_k\circ\alpha_i,$ for $\boldsymbol{\Psi}_{\!3,3} = \boldsymbol{\Psi}_{\!3}\circ\boldsymbol{\Psi}_{\!2}\circ\boldsymbol{\Psi}_{\!1}$ holds
\begin{align*}
    \boldsymbol{\Psi}_{\!3,3} &\!=\!\!
    \begin{bmatrix}
        \alpha_{3,1} \!\!&\! \alpha_{3,2} \!\!&\! \alpha_{3,3}\\
        \mathbb{I} \!\!&\! 0 \!\!&\! 0\\
        0 \!\!&\! \mathbb{I} \!\!&\! 0
    \end{bmatrix}
    \!\!
    \begin{bmatrix}
        \alpha_{2,1} \!\!&\! \alpha_{2,2} \!\!&\! \alpha_{2,3}\\
        \mathbb{I} \!\!&\! 0 \!\!&\! 0\\
        0 \!\!&\! \mathbb{I} \!\!&\! 0
    \end{bmatrix}
    \!\!
    \begin{bmatrix}
        \alpha_{1,1} \!\!&\! \alpha_{1,2} \!\!&\! \alpha_{1,3}\\
        \mathbb{I} \!\!&\! 0 \!\!&\! 0\\
        0 \!\!&\! \mathbb{I} \!\!&\! 0
    \end{bmatrix}\\
    &\!=\!\! 
    \begin{bmatrix}
        \alpha_{3,1}(\alpha_{2,1}\alpha_{1,1} \!+\! \alpha_{2,2}) \!+\! \alpha_{3,2}\alpha_{1,1} \!+\! \alpha_{3,3} \!\!&\! \alpha_{3,1}(\alpha_{2,1}\alpha_{1,2} \!+\! \alpha_{2,3}) \!+\! \alpha_{3,2}\alpha_{1,2} \!\!&\! \alpha_{3,1}\alpha_{2,1}\alpha_{1,3} \!+\! \alpha_{3,2}\alpha_{1,3} \\
        \alpha_{2,1}\alpha_{1,1} \!+\! \alpha_{2,2} \!\!&\! \alpha_{2,1}\alpha_{1,2} \!+\!  \alpha_{2,3} \!\!&\! \alpha_{2,1}\alpha_{1,3}\\
        \alpha_{1,1} \!\!&\! \alpha_{1,2} \!\!&\! \alpha_{1,3}
    \end{bmatrix}
\end{align*}
This structure is also given for general $p.$ Namely, for the entries $\boldsymbol{\Psi}_{\!p,p;i,j},$ $1\leq i,j\leq p,$ of  $\boldsymbol{\Psi}_{\!p,p} = \boldsymbol{\Psi}_{\!p}\circ\boldsymbol{\Psi}_{\!p-1}\circ\cdots\circ\boldsymbol{\Psi}_{\!1}$ for general $p,$ holds by putting $\alpha_{k,\ell} = 0$ for $\ell>p\colon$
\begin{alignat*}{2}
    \boldsymbol{\Psi}_{\!p,p;p,j} &= \alpha_{1,j}, &&\quad 1\leq j \leq p,\\
    \boldsymbol{\Psi}_{\!p,p;p-\ell,j} &= \bigg(\sum^\ell_{k=1}\alpha_{\ell+1,k}\boldsymbol{\Psi}_{\!p,p;p-\ell+k,j}\bigg) + \alpha_{\ell+1,j+1},&&\quad 1\leq \ell < p, ~1\leq j \leq p.
\end{alignat*}

Further, by the definition of $\alpha_{k,i},$ for any $i,k,$ it holds $\|\alpha_{k,i}\|_\mathcal{L} \le \|\boldsymbol{\alpha}\|_{\mathcal{L}}\|\varepsilon_k\|^2,$ and for any $A=(A_1, \dots, A_p)\in\mathcal{S}^p$ with $\|A\|_\mathcal{S} \leq 1,$ also
$$
    \bigg\|\sum^p_{i=1}\alpha_{k,i}(A_i)\bigg\|_\mathcal{S} \le \|\boldsymbol{\alpha}\|_{\mathcal{L}}\|\varepsilon_k\|^2.
$$
This along with the structure of $\boldsymbol{\Psi}_{\!p,p}$ (see the structure above for $p=3),$ sub-multiplicativity of the given operators, and the fact that $(\varepsilon_k) \subset L^2$ is i.i.d., gives
\begin{align*}
    \Exp \tau(\boldsymbol{\Psi}_{\!p,p}) 
    &\leq  \Exp\!\bigg(\sup_{\|A\|_{\mathcal{S}}\leq 1}\!\big\|\boldsymbol{\Psi}_{\!p,p}(A)\big\|_{\mathcal{S}}\bigg)\\
    &\leq \|\boldsymbol{\alpha}\|_\mathcal{L}\Exp\!\|\varepsilon_0\|^2\sum^{p-1}_{\ell=0}\,(p-\ell)\big(\|\boldsymbol{\alpha}\|_\mathcal{L}\Exp\!\|\varepsilon_0\|^2\big)^\ell\,.
\end{align*}
Since this identity matches the claim, the proof is complete.
\end{proof}

\begin{proof}[\textbf{Proof of Proposition  \ref{prop:Explicit moments}}] Throughout the entire proof, let $\nu >0.$ Further, 
\begin{align*}
    c_{n,\nu} \coloneqq 
    \begin{cases}
        1, & \nu \in (0,1],\\
        n^{\nu-1}, & \nu > 1,
    \end{cases} \qquad n\in\mathbb{N}.
\end{align*}

\noindent (a)\, We first show that the diagonal part $X^{\otimes 2}_{0,\mathrm{d}}$ of $X^{\otimes 2}_0$ has a finite $2\nu$th moment. Since $\tau$ is compatible with the H-S norm, and sub-multiplicative, \eqref{iterative application of the multivariate representation}--\eqref{sk10} and elementary inequalities yield, using $c_{m,\nu}c_{n,\nu}=c_{mn,\nu}$ for any $m,n$ and $c_{m,\nu}\le c_{n,\nu}$ for $m<n$, the upper bound
\begin{align}
	\big\|X^{\otimes2,[p]}_{0,\mathrm{d}}\big\|^{\nu}_{\mathcal{S}}
    &\leq c_{3,\nu}\bigg[\,\|\boldsymbol{\Delta}_{0}\|^{\nu}_{\mathcal{S}} + \bigg(\sum_{\ell=1}^{n-1}\tau(\boldsymbol{\Psi}_{\!0,\ell})\|\boldsymbol{\Delta}_{-\ell}\|_{\mathcal{S}}\bigg)^{\!\nu} + \bigg(\sum_{\ell=n}^{\infty}\tau(\boldsymbol{\Psi}_{\!0,\ell})\|\boldsymbol{\Delta}_{-\ell}\|_{\mathcal{S}}\bigg)^{\!\nu}\,\bigg]\notag\\
	&\leq  c_{3n,\nu}\bigg[\,\|\boldsymbol{\Delta}_{0}\|^{\nu}_{\mathcal{S}} + \sum_{\ell=1}^{n-1}\|\boldsymbol{\Delta}_{-\ell}\|^\nu_{\mathcal{S}}\prod_{m=1}^\ell\,\tau^\nu(\boldsymbol{\Psi}_{\!1-m}) + \bigg(\sum_{\ell=n}^{\infty}\tau(\boldsymbol{\Psi}_{\!0,\ell})\|\boldsymbol{\Delta}_{-\ell}\|_{\mathcal{S}}\bigg)^{\!\nu}\,\bigg].\label{sk9}
\end{align}

For the first term (cf.~the proof of Theorem \ref{theo:strict_stat}), we obtain
\begin{align}\label{1st term Lpapprox}
   \Exp\!\|\boldsymbol{\Delta}_0\|^\nu_{\mathcal{S}}
   \;\le\; \|\Delta\|^\nu_{\mathcal{L}}\Exp\!\|\varepsilon_{0}\|^{2\nu} \;<\; \infty.
\end{align}

Moreover, by arguments from Proposition \ref{prop:existence gamma}, we have
\begin{align}\label{sk11}
    \Exp\tau^\nu(\boldsymbol{\Psi}_{\!0})
    \;\le\; c_{2,\nu/2}\big(1+\|\boldsymbol{\alpha}\|^\nu_\mathcal{L}\Exp\!\|\varepsilon_0\|^{2\nu}\big) \;<\; \infty.
\end{align}
Since $\boldsymbol{\Delta}_k,\boldsymbol{\Psi}_{\!k}\in\sigma(\varepsilon_k)$, the sequences $(\boldsymbol{\Delta}_k)$ and $(\boldsymbol{\Psi}_{\!k})$ are i.i.d., and independent across different time indices. Hence, for the second term in \eqref{sk9},
\begin{align*}
    \Exp\bigg(\sum_{\ell=1}^{n-1}\|\boldsymbol{\Delta}_{-\ell}\|^\nu_{\mathcal{S}}\prod_{m=1}^\ell\,\tau^\nu(\boldsymbol{\Psi}_{\!1-m})\bigg)
    \;=\; \Exp\!\|\boldsymbol{\Delta}_{0}\|^{\nu}_{\mathcal{S}}\sum_{\ell=1}^{n-1}\big[\Exp\tau^\nu(\boldsymbol{\Psi}_{\!0})\big]^\ell
    \;<\; \infty.
\end{align*}

We now treat the last term of \eqref{sk9}. The condition \eqref{eq:exp opGARCH smaller 1} reads 
\begin{align}\label{sk12}
    \xi_{n,\nu} \coloneqq\Exp\tau^\nu(\boldsymbol{\Psi}_{\!n,n})<1, 
\end{align}
where $\boldsymbol{\Psi}_{\!k,\ell}=\boldsymbol{\Psi}_{\!k}\circ\boldsymbol{\Psi}_{\!k-1}\circ\cdots\circ\boldsymbol{\Psi}_{\!k-\ell+1}$. To use this structure, we decompose each $\boldsymbol{\Psi}_{\!0,\ell}$ into a part whose length is a multiple of $n$ and a remainder. Therefore, by writing $\ell = \lfloor \ell/n\rfloor n + \ell \bmod n,$ we obtain the factorization
\begin{align*}
    \boldsymbol{\Psi}_{\!0,\ell}
    &= \big(\boldsymbol{\Psi}_{\!0}\circ\boldsymbol{\Psi}_{\!-1} \circ \cdots \circ \boldsymbol{\Psi}_{\!-\lfloor \ell/n\rfloor n+1}\big)\circ\big(\boldsymbol{\Psi}_{\!-\lfloor \ell/n\rfloor n}\circ\cdots\circ \boldsymbol{\Psi}_{\!-\ell+1}\big)\notag\\
    &= \boldsymbol{\Psi}_{\!0,\lfloor \ell/n\rfloor n} \circ \boldsymbol{\Psi}_{\!-\lfloor \ell/n\rfloor n,\; \ell\bmod n}. 
\end{align*}
The first composition contains exactly $\lfloor \ell/n\rfloor n$ operators and thus consists of $\lfloor \ell/n\rfloor$ blocks of length $n$, while the second contains the remainder $\ell\bmod n$ terms. By using sub-multiplicativity of $\tau$, independence, and \eqref{sk11}--\eqref{sk12}, we therefore obtain for any $\nu \in (0,1],$ since $|\xi_{n,\nu}|<1,$
\begin{align*}
    \Exp\bigg(\sum_{\ell=n}^{\infty}\tau(\boldsymbol{\Psi}_{\!0,\ell})\|\boldsymbol{\Delta}_{-\ell}\|_{\mathcal{S}}\bigg)^{\!\nu} 
    &\le \Exp\!\|\boldsymbol{\Delta}_{0}\|^{\nu}_{\mathcal{S}}\sum_{\ell=n}^{\infty}\Exp\tau^\nu(\boldsymbol{\Psi}_{\!0,\ell})\\
    &\le \Exp\!\|\boldsymbol{\Delta}_{0}\|^\nu_{\mathcal{S}}\bigg(\sum_{k=1}^{n-1}\big[\Exp\tau^\nu(\boldsymbol{\Psi}_{\!0})\big]^k\bigg)\sum_{\ell=1}^{\infty}\,\xi^\ell_{n,\nu}\\
    &< \infty,
\end{align*}
and for $\nu>1$, the Minkowski’s inequality yields
\begin{align*}
    \Exp\bigg(\sum_{\ell=n}^{\infty}\tau(\boldsymbol{\Psi}_{\!0,\ell})\|\boldsymbol{\Delta}_{-\ell}\|_{\mathcal{S}}\bigg)^{\!\nu} 
    &\le \Exp\!\|\boldsymbol{\Delta}_{0}\|^\nu_{\mathcal{S}}\bigg(\sum_{\ell=n}^{\infty}\big[\Exp\tau^\nu(\boldsymbol{\Psi}_{\!0,\ell})\big]^{1/\nu}\bigg)^{\!\nu}\\
    &\le \Exp\!\|\boldsymbol{\Delta}_{0}\|^\nu_{\mathcal{S}}\bigg(\sum_{k=1}^{n-1}\big[\Exp\tau^\nu(\boldsymbol{\Psi}_{\!0})\big]^k\bigg)\bigg(\sum_{\ell=1}^\infty\,\xi_{n,\nu}^{\ell/\nu}\bigg)^{\!\nu}\\
    &< \infty.
\end{align*}
Altogether, all expectations in \eqref{sk9} are finite, and consequently, $\Exp\!\|X^{\otimes2,[p]}_{0,\mathrm{d}}\|^{\nu}_{\mathcal{S}}<\infty.$ Next, due to $\Sigma_k=\Delta+\boldsymbol{\alpha}(X^{\otimes 2,[p]}_{k-1}),$ and since $\boldsymbol{\alpha}$ annihilates the off-diagonal part of $X^{\otimes 2,[p]}_{k-1}$ (see \eqref{eq:diagARCH}), under our assumptions, it holds
\begin{align*}
    \Exp\!\|\Sigma_0\|^\nu_\mathcal{S}
    \;\le\; c_{2,\nu}\Big(\|\Delta\|^\nu_\mathcal{S} + \|\boldsymbol{\alpha}\|^\nu_\mathcal{L}\Exp\!\|X^{\otimes 2,[p]}_{0,\mathrm{d}}\|^\nu_\mathcal{S}\Big) ~<~ \infty.
\end{align*}
Further, since $\|\Sigma^{1/2}_0\|_4=\|\Sigma_0\|^{1/2}_\mathcal{S}$, we obtain $\Exp\!\|\Sigma^{1/2}_0\|^{2\nu}_4<\infty$. Finally, as $\|X^{\otimes 2}_0\|_\mathcal{S}=\|X_0\|^2$ and $X_0=\Sigma^{1/2}_0(\varepsilon_0),$ with $\Sigma_0$ being independent of $\varepsilon_0$, we have
\begin{align*}
    \Exp\!\|X^{\otimes 2}_0\|^\nu_\mathcal{S} \;=\; \Exp\!\|X_0\|^{2\nu} \;\le\; \Exp\!\|\Sigma^{1/2}_0\|^{2\nu}_4\Exp\!\|\varepsilon_0\|^{2\nu} \;=\; \Exp\!\|\Sigma_0\|^{\nu}_\mathcal{S}\Exp\!\|\varepsilon_0\|^{2\nu} ~<~ \infty.
\end{align*}

\smallskip

\noindent (b)\, Here too, as a preliminary step, we focus on the diagonal part $X^{\otimes 2}_{k,\mathrm{d}}$ of $X^{\otimes 2}_{k}.$ By definition of each $\boldsymbol{\Psi}_{\!m}$ it holds $\boldsymbol{\Psi}_{\!k,\ell} = \boldsymbol{\Psi}_{\!k}\circ\boldsymbol{\Psi}_{\!k-1}\circ\cdots\circ\boldsymbol{\Psi}_{\!k-\ell+1} = f(\varepsilon_k,\varepsilon_{k-1},\dots,\varepsilon_{k-\ell+1})$ for some measurable function $f.$ Therefore, $\boldsymbol{\Psi}^{(m)}_{\!k,\ell} = f(\varepsilon_k,\varepsilon_{k-1},\dots,\varepsilon_{k-m+1},\varepsilon'_{m},\varepsilon'_{m-1},\dots,\varepsilon'_{m-\ell+1})$ in the spirit of $L^p$-$m$-approximability (cf.~Section \ref{subsec:momentsWeakdependence}). Since $\boldsymbol{\Delta}_{-\ell}$ depends only on $\varepsilon_{-\ell}$, with $\boldsymbol{\Delta}'_{-\ell}$ defined analogously but using $\varepsilon'_{-\ell}$, it follows from \eqref{iterative application of the multivariate representation} that
\begin{align*}
    X^{\otimes2,[p]}_{0,\mathrm{d}} - X^{\otimes2,[p], (m)}_{0,\mathrm{d}}
    &= \sum_{\ell>m}\Big[\boldsymbol{\Psi}_{\!0,\ell}(\boldsymbol{\Delta}_{-\ell}) - \boldsymbol{\Psi}^{(m)}_{\!0,\ell}(\boldsymbol{\Delta}'_{-\ell})\Big],\quad m\in\mathbb{N}.
\end{align*}
Since $\boldsymbol{\Psi}_{\!0,\ell}$ and $\boldsymbol{\Delta}_{-\ell}$ are independent, and as $\boldsymbol{\Psi}_{\!0,\ell}(\boldsymbol{\Delta}_{-\ell}) \stackrel{d}{=} \boldsymbol{\Psi}^{(m)}_{\!0,\ell}(\boldsymbol{\Delta}'_{-\ell})$ and $\boldsymbol{\Delta}_{-\ell}\stackrel{d}{=} \boldsymbol{\Delta}_0$ for any $\ell,m$, we obtain for any $\nu>0$,
\begin{align}
    \Exp\!\big\|\boldsymbol{\Psi}_{\!0,\ell}(\boldsymbol{\Delta}_{-\ell}) - \boldsymbol{\Psi}^{(m)}_{\!0,\ell}(\boldsymbol{\Delta}'_{-\ell})\big\|^\nu_\mathcal{S}
    &\leq 2c_{2,\nu}\Exp\tau^\nu(\boldsymbol{\Psi}_{\!0,\ell})\Exp\!\|\boldsymbol{\Delta}_0\|^\nu_\mathcal{S}\notag\\
    &\leq 2c_{2,\nu}\|\Delta\|^\nu_{\mathcal{L}}\Exp\!\|\varepsilon_0\|^{2\nu}\big(\Exp\tau^\nu(\boldsymbol{\Psi}_{\!0})\big)^{\ell\bmod n}\xi^{\lfloor \ell/n\rfloor}_{n,\nu}.\notag
\end{align}
Further, according to $|\xi_{n,\nu}|<1$, 
$$
    \sum_{\ell>m}\xi^{\lfloor \ell/n\rfloor}_{n,\nu}
    \;\le\; n\sum_{\ell>m}(\xi^{1/n}_{n,\nu})^\ell
    \;\propto\; (\xi^{1/n}_{n,\nu})^m,
$$
and since $((\Exp\tau^\nu(\boldsymbol{\Psi}_{\!0}))^{\ell\bmod n})_\ell$ is uniformly bounded by some $C$, for $\nu\in(0,1]$,
\begin{align}
    \Exp\!\big\|X^{\otimes2,[p]}_{0,\mathrm{d}} - X^{\otimes2,[p],(m)}_{0,\mathrm{d}}\big\|^\nu_\mathcal{S}
    \;\le\; 2C\|\Delta\|^\nu_{\mathcal{L}}\Exp\!\|\varepsilon_0\|^{2\nu}\sum_{\ell>m}\xi^{\lfloor \ell/n\rfloor}_{n,\nu}
    \;\propto\; (\xi^{1/n}_{n,\nu})^m.\notag
\end{align}
Thus $(X^{\otimes2,[p]}_{k,\mathrm{d}})_k$ is $L^{\nu}$-$m$-approximable with geometrically decaying approximation errors for $\nu\in(0,1]$. The case $\nu>1$ follows analogously (cf.~proof of part (a)).

This property transfers to $(\Sigma_k)$ because $\Sigma_k=\Delta+\boldsymbol{\alpha}(X^{\otimes2,[p]}_{k-1}),$  $\Sigma^{(m)}_k\coloneqq\Delta+\boldsymbol{\alpha}(X^{\otimes2,[p],(m)}_{k-1})$ and the observations in part (a) lead to 
\begin{align*}
    \Exp\!\big\|\Sigma_0-\Sigma^{(m)}_0\big\|^\nu_\mathcal{S} &\le \|\boldsymbol{\alpha}\|^\nu_\mathcal{L}\Exp\!\big\|X^{\otimes2,[p]}_{0,\mathrm{d}} - X^{\otimes2,[p],(m)}_{0,\mathrm{d}}\big\|^\nu_\mathcal{S}.
\end{align*}
Further, it follows that
\begin{align*}
    \Exp\!\big\|\Sigma^{1/2}_0-(\Sigma^{1/2}_0)^{(m)}\big\|^{2\nu}_4 &\le \Exp\!\big\|\Sigma_0-\Sigma^{(m)}_0\big\|^{\nu}_\mathcal{S},
\end{align*}
where $\|\cdot\|_4$ is the norm of the Schatten class operators of order 4, and with $X^{(m)}_k\coloneqq(\Sigma^{1/2}_k)^{(m)}(\varepsilon_k),$
\begin{align*}
    \Exp\!\big\|X_0-X^{(m)}_0\big\|^{2\nu}\! &\le \Exp\!\big\|\Sigma^{1/2}_0-(\Sigma^{1/2}_0)^{(m)}\big\|^{2\nu}_4\Exp\!\|\varepsilon_0\|^{2\nu},
\end{align*}
so the claims for $(\Sigma^{1/2}_k)$ and $(X_k)$ hold. Finally, the claim for $(X^{\otimes 2}_k)$ follows from, e.g., \citet[Lemma 2.1]{HoermannKokoszka2010}.
\end{proof}

\begin{proof}[\textbf{Proof of Proposition \ref{prop:weak stationarity}}]
Recall that $(X_k)$ is a CCC-op-ARCH$(p)$ process that is weakly stationary, i.e., $\Exp(X_0)=\Exp(X_k)$ for all $k$ and $\mathscr{C}_{X_k,X_\ell}=\mathscr{C}_{X_0,X_{\ell-k}}$ for all $k,\ell.$

Since $X_0=\Sigma^{1/2}_0(\varepsilon_0)$ and $\Sigma_0$ and $\varepsilon_0$ are independent, $(X_k)$ is centered. Moreover, independence of $\varepsilon_n$ and $\Sigma_m$ for $n>m$ implies $\mathscr{C}_{X_k,X_\ell}=0$ for $k\neq\ell$, and positive definiteness of $\Sigma_0$ yields $\Exp\!\|X_0\|^2=\Exp\langle\Sigma_0(\varepsilon_0),\varepsilon_0\rangle>0.$ Hence, $(X_k)$ is a WWN.

Since expectation commutes with bounded linear operators, weak stationarity gives
$$
    \Exp(\Sigma_k)
    \;=\; \Delta + \sum_{i=1}^p \alpha_i(\Exp(X^{\otimes 2}_{k-i}))
    \;=\; \Delta + \sum_{i=1}^p \alpha_i(\mathscr{C}_{\!\boldsymbol{X}}),\quad k\in\mathbb{Z},
$$
establishing $\boldsymbol{\mu}_{\boldsymbol{\Sigma}}=\Exp(\Sigma_0)=\Exp(\Sigma_k)$ for all $k$ and the representation \eqref{eq:Mean vol process ARCH}.

To obtain the explicit form of $\boldsymbol{\mu}_{\boldsymbol{\Sigma}}$, note that Assumption \ref{as:Assumptions of (eps_k) and assoc. covop} implies $\mathscr{C}_{\!\boldsymbol{X}}=\boldsymbol{\mu}_{\boldsymbol{\Sigma}}\mathscr{C}_{\boldsymbol{\varepsilon}}.$ Combined with \eqref{eq:diagARCH}, this gives $\alpha_i(\mathscr{C}_{\!\boldsymbol{X}})=\boldsymbol{\mu}_{\boldsymbol{\Sigma}}\alpha_i(\mathscr{C}_{\boldsymbol{\varepsilon}})$ for all $i.$ Inserting this into \eqref{eq:Mean vol process ARCH} and using \eqref{eq:product op and innovation cov op norm}, the resulting Neumann series converges and we obtain
\begin{align*}
    \mu_{\boldsymbol{\Sigma}}
    \;=\; \Delta + \mu_{\boldsymbol{\Sigma}}\sum_{i=1}^p \alpha_i(\mathscr{C}_{\boldsymbol{\varepsilon}})
    \;=\; \bigg(\mathbb{I}-\sum_{i=1}^p \alpha_i(\mathscr{C}_{\boldsymbol{\varepsilon}})\bigg)^{-1}(\Delta).
\end{align*}
Thus, all claims follow.
\end{proof}

\begin{proof}[\textbf{Proof of Theorem \ref{Th: Cons est alphap fin}}] From the definition of $\hat{\boldsymbol{\alpha}}_{\mathrm{d}},$ \eqref{eq:YW of diag alpha_p}, and $\tilde{\boldsymbol{\alpha}}_{\mathrm{d}} = \diag^\ast(\boldsymbol{\alpha}_{\mathrm{d}}),$ it follows
\begin{align}\label{eq: identity est errors diag finite}
    \hat{\boldsymbol{\alpha}}_{\mathrm{d}} - \tilde{\boldsymbol{\alpha}}_{\mathrm{d}} = \diag^\ast\!\Big(\hat{\mathscr{D}}_{\mathrm{d}}\big[\hat{\mathscr{C}}^{-1}_{\mathrm{d}} - \mathscr{C}^{-1}_{\mathrm{d}}\big]\Big) + \diag^\ast\! \Big((\hat{\mathscr{D}}_{\mathrm{d}} - \mathscr{D}_{\mathrm{d}})\mathscr{C}^{-1}_{\mathrm{d}}\Big) + \diag^\ast\!\Big(\mathscr{R}_{\mathrm{d}}\mathscr{C}^{-1}_{\mathrm{d}}\Big),
\end{align}
where the inverses are well-defined for sufficiently large $N$ by Assumption \ref{as: C_d and hat C_d finite injective}. Let $\|\cdot\|$ be the induced matrix norm. Note that $\mathscr{C}_{\mathrm{d}} = \diag^\ast\!\mathscr{C}\diag$, $\mathscr{D}_{\mathrm{d}} = \diag^\ast\!\mathscr{D}\diag$, and $\mathscr{R}_{\mathrm{d}} = \diag^\ast\!\mathscr{R}\diag$ (with analogous empirical versions). For the remainder $\mathscr{R}$ in \eqref{eq:Def remainders R_TNp}, it follows from $\boldsymbol{\alpha} = \boldsymbol{\alpha}_K$ in \eqref{eq:representation alpha_p, diag, fin}, the operator definitions, the H-S norm, and $\vartheta_{\!N} = \mathrm{O}(N^{-1/2})$, as in the proof of Proposition \ref{prop:asymptotic result Delta ARCH, inf}:
\begin{align*}
    \|\mathscr{R}\|_{\mathcal{N}} 
    &\leq 
    \Exp\!\big\|\tilde{X}^{\otimes2,[p]}_0\big\|_{\mathcal{S}}\,\bigg\|\Big[\boldsymbol{\alpha}_K\big(\tilde{X}^{\otimes2,[p]}_0\big)\Big]\bigg(\mathscr{C}^\ddagger_{\!\boldsymbol{\varepsilon}}\!\coprod^{e_K}_{e_1}\,-\;\mathbb{I}\bigg)\bigg\|_{\mathcal{S}}\\
    &\leq \vartheta_{\!N}a^{-1}_K \Exp\!\big\|\tilde{X}^{\otimes2,[p]}_0\big\|_{\mathcal{S}}\,\Bigg(\sum^K_{m=1}\,\bigg\|\sum^p_{i=1}\sum^K_{\ell=1}\,a_{i\ell}\langle \tilde{X}^{\otimes2}_{1-i}, e_\ell\otimes e_\ell\rangle_{\!\mathcal{S}}\,(e_\ell\otimes e_\ell)(e_m)\bigg\|^2\Bigg)^{\!\!1/2}\allowdisplaybreaks\\
    &= \vartheta_{\!N}a^{-1}_{\!K} \Exp\!\big\|\tilde{X}^{\otimes2,[p]}_0\big\|_{\mathcal{S}}\big\|\boldsymbol{\alpha}_K\big(\tilde{X}^{\otimes2,[p]}_0\big)\big\|_{\mathcal{S}}\\
    &\leq \vartheta_{\!N}a^{-1}_{\!K}\|\boldsymbol{\alpha}_K\|_\mathcal{S} \Exp\!\big\|\tilde{X}^{\otimes2,[p]}_0\big\|^2_{\mathcal{S}}\\
    &= \mathrm{O}(N^{-1/2})\,.
\end{align*}
Subsequently, by the definition of $\mathscr{R}_{\mathrm{d}},$ and due to elementary conversions, we have
\begin{align*}
    \|\mathscr{R}_{\mathrm{d}}\| \;\leq\; \|\!\diag\!\|^2_\mathcal{L}\,\|\mathscr{R}\|_\mathcal{N} \;=\; \mathrm{O}(N^{-1/2}).
\end{align*}
Hence, due to $\|\hat{\mathscr{C}} - \mathscr{C}\| = \mathrm{O}_{\Prob}(N^{-1/2})$ and $\|\hat{\mathscr{D}} - \mathscr{D}\| = \mathrm{O}_{\Prob}(N^{-1/2})$ for fixed $K$ by Lemma~\ref{Lemma: Cons lagged cross-cov non-centered}, and the definition of $\mathscr{C}_{\mathrm{d}}, \mathscr{D}_{\mathrm{d}}$, along with their empirical counterparts, it holds
$$
    \max\Big\{\,\|\hat{\mathscr{D}}_{\mathrm{d}} - \mathscr{D}_{\mathrm{d}}\|, \,\|\hat{\mathscr{C}}_{\mathrm{d}} - \mathscr{C}_{\mathrm{d}}\|, \,\|\mathscr{R}_{\mathrm{d}}\|\Big\} = \mathrm{O}_{\Prob}(N^{-1/2}).
$$
Further, with $E_N\coloneqq \mathscr{C}^{-1}_{\mathrm{d}}(\hat{\mathscr{C}}_{\mathrm{d}} - \mathscr{C}_{\mathrm{d}}),$ where $\mathscr{C}_{\mathrm{d}}$ is non-singular, it holds $\hat{\mathscr{C}}_{\mathrm{d}} = \mathscr{C}_{\mathrm{d}}(\mathbb{I} + E_N).$ The Neumann series 
$$
    \big(\mathbb{I} + E_N\big)^{-1} = \sum^\infty_{k=0}\,(-E_N)^k,
$$
with $E^0_N \coloneqq \mathbb{I},$ converges (for $k\to\infty$) with high probability as $\|E_N\| = \mathrm{O}_{\Prob}(N^{-1/2}).$ Thus, due to sub-multiplicativity of the induced matrix norm and the triangle inequality, it holds 
\begin{align*}
    \big\|\hat{\mathscr{C}}^{-1}_{\mathrm{d}} - \mathscr{C}^{-1}_{\mathrm{d}}\big\| \;=\; \big\|\big((\mathbb{I} + E_N)^{-1} - \mathbb{I}\big)\mathscr{C}^{-1}_{\mathrm{d}}\big\| \;\leq\; \|E_N\|\big\|\mathscr{C}^{-1}_{\mathrm{d}}\big\|\sum^\infty_{k=0}\|E_N\|^k \;=\; \mathrm{O}_{\Prob}(N^{-1/2})\,.
\end{align*}

Combining all above with the identity \eqref{eq: identity est errors diag finite}, together with the triangle inequality and sub-multiplicativity, the claim follows for the induced matrix norm. Further, as all norms on finite-dimensional spaces are equivalent, the result extends to any matrix norm.
\end{proof}

\begin{proof}[\textbf{Proof of Theorem \ref{theo:Cons est alphap infinite diag}}] Recall that $\hat{\mathscr{D}}_{\mathrm{d}} = \diag^\ast\hat{\mathscr{D}}\,\diag$ and $\hat{\mathscr{C}}_{\mathrm{d}} = \diag^\ast\hat{\mathscr{C}}\,\diag$, where $\hat{\mathscr{D}}$ and $\hat{\mathscr{C}}$ are given in \eqref{eq:est lag-0-cov-op} and \eqref{eq:est lag-1-cross-cov-op}, respectively. Also, $\hat{\mathscr{C}}^\dagger_{\mathrm{d}} = (\hat{\mathscr{C}}_{\mathrm{d}} + \vartheta_{\!N}\mathbb{I})^{-1}$ with $\vartheta_{\!N} \to 0$ and $\vartheta_{\!N} > 0$. The eigenpairs $(\hat{\lambda}_{j,\mathrm{d}}, \hat{c}_{j,\mathrm{d}})$ and $(\lambda_{j,\mathrm{d}}, c_{j,\mathrm{d}})$ correspond to $\hat{\mathscr{C}}_{\mathrm{d}}$ and $\mathscr{C}_{\mathrm{d}}$, respectively, with $\hat{\lambda}_{1,\mathrm{d}} > \dots > \hat{\lambda}_{K,\mathrm{d}} > \hat{\lambda}_{K+1,\mathrm{d}} \geq 0$ by Assumption \ref{as: Dim eigenspaces op-ARCH, diag, inf}~(b), where all $\lambda_{j,\mathrm{d}}$ are positive due to injectivity of $\mathscr{C}_{\mathrm{d}}$. Define $\mathscr{C}^\dagger_{\mathrm{d}} \coloneqq (\mathscr{C}_{\mathrm{d}} + \vartheta_{\!N}\mathbb{I})^{-1}$, $\mathscr{C}^\ddagger_{\mathrm{d}} \coloneqq \mathscr{C}_{\mathrm{d}}\mathscr{C}^\dagger_{\mathrm{d}}$, and let $\boldsymbol{\alpha}^{K}_p$ be the finite-dimensional diagonal operator in \eqref{eq:representation alpha_p, diag, fin}. Moreover, $\boldsymbol{\alpha}_{\mathrm{d},K}$ denotes $\boldsymbol{\alpha}_{\mathrm{d}}$ with $a_{ij} = 0$ for $j > K$, and set $\tilde{\boldsymbol{\alpha}}_{\mathrm{d},K} \coloneqq \diag^\ast(\boldsymbol{\alpha}_{\mathrm{d},K})$. Then, using definitions \eqref{eq:Def bold alpha_p diag}–\eqref{eq:estimator for all coefficients in diag, infinite}, the YW-type equation \eqref{eq:YW of diag alpha_p}, $\tilde{\boldsymbol{\alpha}}_{\mathrm{d}} = \diag^\ast(\boldsymbol{\alpha}_{\mathrm{d}})$, and standard conversions, we have
\begin{align*}
    \hat{\boldsymbol{\alpha}}_{\mathrm{d}} - \tilde{\boldsymbol{\alpha}}_{\mathrm{d}} 
    &= \diag^\ast\!\Bigg(\big(\hat{\mathscr{D}}_{\mathrm{d}} - \mathscr{D}_{\mathrm{d}}\big)\hat{\mathscr{C}}^\dagger_{\mathrm{d}}\!\coprod^{\hat{c}_{K,\mathrm{d}}}_{\hat{c}_{1,\mathrm{d}}}\Bigg) \,+\, \diag^\ast\Bigg(\mathscr{D}_{\mathrm{d}}\bigg(\hat{\mathscr{C}}^\dagger_{\mathrm{d}}\!\coprod^{\hat{c}_{K,\mathrm{d}}}_{\hat{c}_{1,\mathrm{d}}} - \,\mathscr{C}^\dagger_{\mathrm{d}}\!\coprod^{c_{K,\mathrm{d}}}_{c_{1,\mathrm{d}}}\bigg)\Bigg)\\
    &\phantom{=~} + \diag^\ast\!\Bigg(\mathscr{R}_{\mathrm{d}}\mathscr{C}^\dagger_{\mathrm{d}}\!\coprod^{c_{K,\mathrm{d}}}_{c_{1,\mathrm{d}}}\Bigg) \,+\, \diag^\ast\!\Bigg(\boldsymbol{\alpha}_{\mathrm{d},K}\bigg(\mathscr{C}^\ddagger_{\mathrm{d}}\!\coprod^{c_{K,\mathrm{d}}}_{c_{1,\mathrm{d}}} - \mathbb{I}\bigg)\Bigg) + \,\tilde{\boldsymbol{\alpha}}^K_{\mathrm{d}} - \tilde{\boldsymbol{\alpha}}_{\mathrm{d}}\,, 
\end{align*}
where $\mathscr{R}^K_{\mathrm{d}}$ denotes the remainder from Theorem \ref{Th: Cons est alphap fin}. By the triangle inequality, the operator-valued H\"older's inequality, and since $\diag^\ast$ is a bounded linear operator with $\|\diag^\ast\!\|_{\mathcal{L}} = 1,$ 
\begin{align}\label{initial inequality operators of linear processes} 
    \begin{split}  
        \big\|\hat{\boldsymbol{\alpha}}_{\mathrm{d}} - \tilde{\boldsymbol{\alpha}}_{\mathrm{d}}\big\|_{\mathcal{S}} 
        &\leq \|\hat{\mathscr{D}}_{\mathrm{d}} - \mathscr{D}_{\mathrm{d}}\|_\mathcal{S}\,\bigg\|\hat{\mathscr{C}}^\dagger_{\mathrm{d}}\!\coprod^{\hat{c}_{K,\mathrm{d}}}_{\hat{c}_{1,\mathrm{d}}}\bigg\|_\mathcal{L} + \,\|\mathscr{D}_{\mathrm{d}}\|_\mathcal{S}\,\bigg\|\hat{\mathscr{C}}^\dagger_{\mathrm{d}}\!\coprod^{\hat{c}_{K,\mathrm{d}}}_{\hat{c}_{1,\mathrm{d}}}\, - \;\mathscr{C}^\dagger_{\mathrm{d}}\!\coprod^{c_{K,\mathrm{d}}}_{c_{1,\mathrm{d}}}\bigg\|_\mathcal{L}\\
        &\phantom{\leq~} + \big\|\mathscr{R}^K_{\mathrm{d}}\big\|_\mathcal{S}\,\bigg\|\mathscr{C}^\dagger_{\mathrm{d}}\!\coprod^{c_{K,\mathrm{d}}}_{c_{1,\mathrm{d}}}\bigg\|_\mathcal{L} + \,\bigg\|\,\boldsymbol{\alpha}_{\mathrm{d},K}\bigg(\mathscr{C}^\ddagger_{\mathrm{d}}\!\coprod^{c_{K,\mathrm{d}}}_{c_{1,\mathrm{d}}} - \;\mathbb{I}\bigg)\bigg\|_{\mathcal{S}} + \,\big\|\tilde{\boldsymbol{\alpha}}_{\mathrm{d}} - \tilde{\boldsymbol{\alpha}}^K_{\mathrm{d}}\big\|_{\mathcal{S}}\,.
    \end{split}
\end{align}

In the following, we analyze each quantity in \eqref{initial inequality operators of linear processes}, where several arguments are inspired by proofs in \cite{KuehnertRiceAue2026}. First, from the definition of the operator norm, our pseudoinverses and our projection operators, it follows for any $K,N:$ 
\begin{gather}
    \bigg\|\hat{\mathscr{C}}^\dagger_{\mathrm{d}}\!\coprod^{\hat{c}_{K,\mathrm{d}}}_{\hat{c}_{1,\mathrm{d}}}\bigg\|_\mathcal{L} = (\hat{\lambda}_{K,\mathrm{d}} + \vartheta_{\!N})^{-1}, \qquad \bigg\|\mathscr{C}^\dagger_{\mathrm{d}}\!\coprod^{c_{K,\mathrm{d}}}_{c_{1,\mathrm{d}}}\bigg\|_\mathcal{L} = (\lambda_{K,\mathrm{d}} + \vartheta_{\!N})^{-1}\,.\label{Tikhonov est and pop proj in inv}
\end{gather}
By the definition of $\mathscr{D}_{\mathrm{d}},$ due to elementary conversions, and $\|\diag\|_\mathcal{L} = 1,$ it holds 
\begin{align}\label{eq:upper bound D}
    \|\mathscr{D}_{\mathrm{d}}\|_{\mathcal{S}} \;\leq\; \|\mathscr{D}_{\mathrm{d}}\|_{\mathcal{N}} \;\leq\; \|\mathscr{D}\|_{\mathcal{N}} \;\leq\; \bigg\|\mathscr{C}^\dagger_{\!\boldsymbol{\varepsilon}}\!\coprod_{e_1}^{e_K}\bigg\|_{\mathcal{L}}\Exp\!\big\|X^{\otimes2,[p]}_0\big\|\big\|X^{\otimes2}_1\big\| \;=\; \mathrm{O}(a^{-2}_K)\,.
\end{align}

\noindent Furthermore, with $K_b$ in the definition of $\Lambda_{K,\mathrm{d}}=(\lambda_{K,\mathrm{d}}-\lambda_{K_b,\mathrm{d}})^{-1}$ in \eqref{eq:Def Lambda_K}, we get 
\begin{align}\label{eq:decomposition est Tikh projected}
    \bigg\|\hat{\mathscr{C}}^\dagger_{\mathrm{d}}\!\coprod^{\hat{c}_{K,\mathrm{d}}}_{\hat{c}_{1,\mathrm{d}}}\, - \;\mathscr{C}^\dagger_{\mathrm{d}}\!\coprod^{c_{K,\mathrm{d}}}_{c_{1,\mathrm{d}}}\bigg\|_\mathcal{L}
    &\leq \,\bigg\|\hat{\mathscr{C}}^\dagger_{\mathrm{d}}\!\!\coprod^{\hat{c}_{K_b-1,\mathrm{d}}}_{\hat{c}_{1,\mathrm{d}}} - \;\mathscr{C}^\dagger_{\mathrm{d}}\!\!\coprod^{c_{K_b-1,\mathrm{d}}}_{c_{1,\mathrm{d}}}\bigg\|_\mathcal{L} + \,\bigg\|\hat{\mathscr{C}}^\dagger_{\mathrm{d}}\!\!\coprod^{\hat{c}_{K_b-1,\mathrm{d}}}_{\hat{c}_{K+1,\mathrm{d}}} - \;\mathscr{C}^\dagger_{\mathrm{d}}\!\!\coprod^{c_{K_b-1,\mathrm{d}}}_{c_{K+1,\mathrm{d}}}\bigg\|_\mathcal{L}\,.
\end{align}
Moreover, 
\begin{align*}
    \begin{split}
         \bigg\|\,\hat{\mathscr{C}}^\dagger_{\mathrm{d}}\!\!\coprod^{\hat{c}_{K_b-1,\mathrm{d}}}_{\hat{c}_{1,\mathrm{d}}} - \;\mathscr{C}^\dagger_{\mathrm{d}}\!\!\coprod^{c_{K_b-1,\mathrm{d}}}_{c_{1,\mathrm{d}}}\bigg\|_\mathcal{L} 
         &\leq \,\bigg\|\sum^{K_b-1}_{i=1}\,(\hat{\lambda}_{i,\mathrm{d}} +\vartheta_{\!N})^{-1}\Big[(\hat{c}_{i,\mathrm{d}}\otimes \hat{c}_{i,\mathrm{d}}) - (c_{i,\mathrm{d}}\otimes c_{i,\mathrm{d}})\Big]\bigg\|_\mathcal{L}\\
         &\qquad\;\, + \,\bigg\|\sum^{K_b-1}_{i=1}\,\Big[(\hat{\lambda}_{i,\mathrm{d}} +\vartheta_{\!N})^{-1} -(\lambda_{i,\mathrm{d}} +\vartheta_{\!N})^{-1}\Big](c_{i,\mathrm{d}}\otimes c_{i,\mathrm{d}})\bigg\|_\mathcal{L}.  
    \end{split}
\end{align*}
\noindent Further, by \citet[][Lemmas 3.1--3.2]{Reimherr2015}, $\|\cdot\|_\mathcal{L} \leq \|\cdot\|_\mathcal{S},$ $(\sum_k a_k)^{1/2} \leq \sum_k a_k^{1/2}$ for non-negative $a_k$, $\Lambda_{K,\mathrm{d}}=\Lambda_{K_b-1,\mathrm{d}},$ $K_b \leq K + \xi$ by Assumption \ref{as: Dim eigenspaces op-ARCH, diag, inf}~(a), and basic conversions, with the empirical reciprocal eigengaps $\hat{\Lambda}_{K,\mathrm{d}} = (\hat{\lambda}_{K,\mathrm{d}} - \hat{\lambda}_{K+1,\mathrm{d}})^{-1}$ in \eqref{eq:Def hat Lambda_K}, it holds
\begin{align}
    \bigg\|\,\hat{\mathscr{C}}^\dagger_{\mathrm{d}}\!\!\coprod^{\hat{c}_{K_b-1,\mathrm{d}}}_{\hat{c}_{1,\mathrm{d}}} - \;\mathscr{C}^\dagger_{\mathrm{d}}\!\!\coprod^{c_{K_b-1,\mathrm{d}}}_{c_{1,\mathrm{d}}}\bigg\|_\mathcal{L} 
    &\leq 2(\hat{\lambda}_{K_b-1,\mathrm{d}}+\vartheta_{\!N})^{-1}K^{1/2}_b\,\|\hat{\mathscr{C}}_{\mathrm{d}} - \mathscr{C}_{\mathrm{d}}\|_\mathcal{L}\notag\\
    &\qquad\;\,\times \Big(\hat{\Lambda}_{K_b-1,\mathrm{d}} + \Lambda_{K_b-1,\mathrm{d}} + (\lambda_{K_b-1,\mathrm{d}} + \vartheta_{\!N})^{-1} + 1 \Big)\notag\\
    &\leq 2\hat{\Lambda}_{K_b-1, \mathrm{d}}\,(K+\xi)^{1/2}\,\|\hat{\mathscr{C}}_{\mathrm{d}} - \mathscr{C}_{\mathrm{d}}\|_\mathcal{L}\,\big(\hat{\Lambda}_{K_b-1, \mathrm{d}} + 2\Lambda_{K,\mathrm{d}} + 1 \big).\label{eq:est error Tikhonov proj, main theorem, main part}
\end{align}
Moreover, analogous arguments and $K_b-K \leq \xi$ yield 
\begin{align}\label{eq:est error Tikhonov proj, main theorem, remainder}
    \bigg\|\,\hat{\mathscr{C}}^\dagger_{\mathrm{d}}\!\!\coprod^{\hat{c}_{K_b-1,\mathrm{d}}}_{\hat{c}_{K+1,\mathrm{d}}} - \;\mathscr{C}^\dagger_{\mathrm{d}}\!\!\coprod^{c_{K_b-1,\mathrm{d}}}_{c_{K+1,\mathrm{d}}}\bigg\|_\mathcal{L} &\leq 2\hat{\Lambda}_{K_b-1, \mathrm{d}}\,\xi^{1/2}\,\|\hat{\mathscr{C}}_{\mathrm{d}} - \mathscr{C}_{\mathrm{d}}\|_\mathcal{L}\big(\hat{\Lambda}_{K_b-1, \mathrm{d}} + 2\Lambda_{K,\mathrm{d}} + 1 \big).
\end{align}
Notice that $\hat{\Lambda}_{K_b-1, \mathrm{d}}$ is well-defined due to $K_b-1 \leq K+\xi$ and Assumption \ref{as: Dim eigenspaces op-ARCH, diag, inf}~(b). 

The remainder $\mathscr{R}^K_{\mathrm{d}}$ is identical to $\mathscr{R}_{\mathrm{d}}$ in the proof of Theorem \ref{Th: Cons est alphap fin}, where we deduced
\begin{align}\label{eq:inequality remainder R^K_d}
    \big\|\mathscr{R}^K_{\mathrm{d}}\big\|_\mathcal{S} \leq \vartheta_{\!N}a^{-1}_{\!K}\|\!\diag\!\|^2_\mathcal{L}\,\|\boldsymbol{\alpha}_K\|_\mathcal{S} \Exp\!\big\|\tilde{X}^{\otimes2,[p]}_0\big\|^2_{\mathcal{S}} = \mathrm{O}\big(\vartheta_{\!N}a^{-1}_{\!K}\big)\,.
\end{align}

Similar arguments as in the proof of Theorem \ref{Th: Cons est alphap fin} and the definition of $\boldsymbol{\alpha}_{\mathrm{d},K}$ yield 
\begin{align}
    \bigg\|\boldsymbol{\alpha}_{\mathrm{d},K}\bigg(\mathscr{C}^\ddagger_{\mathrm{d}}\!\coprod^{c_{K,\mathrm{d}}}_{c_{1,\mathrm{d}}} - \;\mathbb{I}\bigg)\Big\|_{\mathcal{S}} 
    &= \bigg(\sum^K_{m=1} \big(\lambda_{m,\mathrm{d}}(\lambda_{m,\mathrm{d}} + \vartheta_{\!N})^{-1} - 1\big)^2\,\big\|\boldsymbol{\alpha}_{\mathrm{d},K}(c_{m, \mathrm{d}})\big\|^2\bigg)^{\!1/2}\notag\\
    &\leq \vartheta_{\!N}\lambda^{-1}_{K,\mathrm{d}}\bigg(\sum^K_{m=1}\,\big\|\boldsymbol{\alpha}_{\mathrm{d},K}(c_{m, \mathrm{d}})\big\|^2\bigg)^{\!1/2}\notag\\
    &\leq \vartheta_{\!N}\lambda^{-1}_{K,\mathrm{d}}\,\big\|\boldsymbol{\alpha}_{\mathrm{d},K}\big\|_\mathcal{S} = \mathrm{O}\big(\vartheta_{\!N}\lambda^{-1}_{K,\mathrm{d}}\big)\,.\label{eq:norm of fin alpha p d on diff pseudo id to id}
\end{align}

Furthermore, due to the definition of the H-S norm and Assumption \ref{eq: Sob cond op-ARCH, infinite diag}, it holds
\begin{align}\label{eq:est error of diag bold alphap and finite-dim proj}
    \big\|\tilde{\boldsymbol{\alpha}}_{\mathrm{d}} - \tilde{\boldsymbol{\alpha}}^K_{\mathrm{d}}\big\|^2_{\mathcal{S}} \;=\; \sum^p_{i=1}\sum_{\ell>K}\,a^2_{i\ell\ell} \;\leq\; K^{-2\gamma}\sum^p_{i=1}\sum_{\ell>K}\,a^2_{i\ell\ell}(1+\ell^{2\gamma}) \;=\; \mathrm{O}(K^{-2\gamma})\,.
\end{align}

Overall, by combining \eqref{Tikhonov est and pop proj in inv}--\eqref{eq:est error of diag bold alphap and finite-dim proj} with \eqref{initial inequality operators of linear processes}, and using that $(\lambda_{K,\mathrm{d}} + \vartheta_{\!N})^{-1} \leq \Lambda_{K,\mathrm{d}}$ and $(\hat{\lambda}_{K,\mathrm{d}} + \vartheta_{\!N})^{-1} \leq \hat{\Lambda}_{K,\mathrm{d}} = \mathrm{O}_{\Prob}(\Lambda_{K,\mathrm{d}}),$ together with Lemma \ref{Lemma: Cons lagged cross-cov non-centered}, it holds
\begin{align*}
    \big\|\hat{\boldsymbol{\alpha}}_{\mathrm{d}} - \tilde{\boldsymbol{\alpha}}_{\mathrm{d}}\big\|_{\mathcal{S}} \leq \mathrm{O}_{\Prob}\big(K^{1/2}a^{-2}_K\Lambda^2_{K,\mathrm{d}}N^{-1/2}\big) + \mathrm{O}\big(\vartheta_{\!N}a^{-1}_{\!K}\big) + \mathrm{O}\big(\vartheta_{\!N}\lambda^{-1}_{K,\mathrm{d}}\big) + \mathrm{O}(K^{-\gamma})\,. 
\end{align*}
Thus, as $K^{\gamma+1/2}a^{-2}_K\Lambda^2_{K,\mathrm{d}} = \mathrm{O}(N^{1/2})$ and $\vartheta_{\!N} = \mathrm{O}(\min(a_K,\lambda_{K, \mathrm{d}})K^{-\gamma}),$ the claim is proved. 
\end{proof}

\begin{proof}[\textbf{Proof of Proposition \ref{prop:asymptotic result Delta ARCH, inf}}] First, let $\Delta_K \coloneqq \sum_{i=1}^K d_i (e_i \otimes e_i)$ with $K$ as in the definition of the estimator $\hat{\Delta}$ in Eq.~\eqref{eq:Esimator Delta ARCH}. Then, 
\begin{align}
    \|\hat{\Delta} - \Delta\|_{\mathcal{S}} &\leq \bigg\|\hat{\Delta} - \Delta_K\mathscr{C}^\ddagger_{\!\boldsymbol{\varepsilon}}\!\coprod^{e_K}_{e_1}\bigg\|_\mathcal{S} + \bigg\|\Delta_K\bigg(\mathscr{C}^\ddagger_{\!\boldsymbol{\varepsilon}}\!\coprod^{e_K}_{e_1} - \mathbb{I}\bigg)\bigg\|_\mathcal{S} + \|\Delta - \Delta_K\|_{\mathcal{S}}\,.\label{eq:inequality proof of est error hat Delta - Delta}
\end{align}
Moreover, with 
$\mathscr{C}^\dagger_{\!\boldsymbol{\varepsilon}} = (\mathscr{C}_{\boldsymbol{\varepsilon}} + \vartheta_{\!N}\mathbb{I})^{-1}$ and $\mathscr{C}^\ddagger_{\!\boldsymbol{\varepsilon}} \coloneqq \mathscr{C}_{\boldsymbol{\varepsilon}}\mathscr{C}^\dagger_{\!\boldsymbol{\varepsilon}}$, we have
$$
    \bigg\|\mathscr{C}^\dagger_{\!\boldsymbol{\varepsilon}}\coprod^{e_K}_{e_1}\bigg\|_\mathcal{L} = \sup_{1 \leq j \leq K} (a_j + \vartheta_{\!N})^{-1} \leq a^{-1}_K, \qquad \bigg\|\mathscr{C}^\ddagger_{\!\boldsymbol{\varepsilon}}\coprod^{e_K}_{e_1}\bigg\|_\mathcal{L} = \sup_{1 \leq j \leq K} a_j(a_j + \vartheta_{\!N})^{-1} \leq 1.
$$
For the first term in \eqref{eq:inequality proof of est error hat Delta - Delta}, using \eqref{eq:Identity for Delta in ARCH}, the definition of $\hat{\Delta}$ in \eqref{eq:Esimator Delta ARCH}, the triangle inequality, and the operator-valued H\"older inequality, we obtain
\begin{align}
    \bigg\|\hat{\Delta} - \Delta\mathscr{C}^\ddagger_{\!\boldsymbol{\varepsilon}}\!\coprod^{e_K}_{e_1}\bigg\|_\mathcal{S} 
    &\leq \|\hat{\mathscr{C}}_{\!\boldsymbol{X}} - \mathscr{C}_{\!\boldsymbol{X}}\|_\mathcal{S}\,\bigg\|\mathscr{C}^\dagger_{\!\boldsymbol{\varepsilon}}\coprod^{e_K}_{e_1}\bigg\|_\mathcal{L} 
    + \big\|\boldsymbol{\hat{\alpha}}(\hat{m}_p) - \boldsymbol{\alpha}(m_p)\big\|_\mathcal{S}\,\bigg\|\mathscr{C}^\ddagger_{\!\boldsymbol{\varepsilon}}\coprod^{e_K}_{e_1}\bigg\|_\mathcal{L}\notag\\
    &\leq a^{-1}_K\|\hat{\mathscr{C}}_{\!\boldsymbol{X}} - \mathscr{C}_{\!\boldsymbol{X}}\|_\mathcal{S} 
    + \|\boldsymbol{\hat{\alpha}} - \boldsymbol{\alpha}\|_\mathcal{S}\|\hat{m}_p\|_\mathcal{S} 
    + \|\boldsymbol{\alpha}\|_\mathcal{S}\|\hat{m}_p - m_p\|_\mathcal{S}\notag\\
    &= \mathrm{O}_{\Prob} \big(a^{-1}_KN^{-1/2} \big) + \mathrm{O}_{\Prob}\big(\|\boldsymbol{\hat{\alpha}} - \boldsymbol{\alpha}\|_\mathcal{S} \big).\label{eq:1st term in upper bound for hat Delta - Delta}
\end{align}

For the second term in \eqref{eq:inequality proof of est error hat Delta - Delta}, noting that $(e_i)$ is a complete orthonormal system, we have
\begin{align}
    \bigg\|\Delta_K\bigg(\mathscr{C}^\ddagger_{\!\boldsymbol{\varepsilon}}\!\coprod^{e_K}_{e_1} - \mathbb{I}\bigg)\bigg\|^2_\mathcal{S} 
    &= \sum^\infty_{j=1}\bigg\|\sum^K_{i=1}d_i(e_i\otimes e_i)\bigg(\mathscr{C}^\ddagger_{\!\boldsymbol{\varepsilon}}\!\coprod^{e_K}_{e_1}-\mathbb{I}\bigg)(e_j)\bigg\|^2\notag\\
    &= \sum^K_{j=1} \Big(a_j(a_j + \vartheta_{\!N})^{-1} - 1\Big)^{\!2}\bigg\|\sum^K_{i=1}d_i(e_i\otimes e_i)(e_j)\bigg\|^2\allowdisplaybreaks\notag\\
    &\leq \vartheta^2_{\!N}a^{-2}_{\!K}\sum^K_{j=1}\bigg\|\sum^K_{i=1}d_i(e_i\otimes e_i)(e_j)\bigg\|^2\allowdisplaybreaks\notag\\
    &= \vartheta^2_{\!N}a^{-2}_{\!K}\|\Delta_K\|^2_\mathcal{S}.\label{eq:2nd term in upper bound for hat Delta - Delta}
\end{align}

For the last term in \eqref{eq:inequality proof of est error hat Delta - Delta}, condition \eqref{eq:Sobolev cond Delta} yields
\begin{align}
    \|\Delta - \Delta_K\|^2_{\mathcal{S}} \;=\; \sum_{i>K}d^2_i 
    \;\leq\; K^{-2\delta}\sum_{i>K}d^2_i(1+i^{2\delta}) 
    \;=\; \mathrm{O}(K^{-2\delta}).\label{eq:3rd term in upper bound for hat Delta - Delta}
\end{align}

Finally, substituting \eqref{eq:1st term in upper bound for hat Delta - Delta}--\eqref{eq:3rd term in upper bound for hat Delta - Delta} into \eqref{eq:inequality proof of est error hat Delta - Delta}, and using $\vartheta_{\!N} = \mathrm{O}(a_KK^{-\delta}),$ where $K=K_{\!N} \to \infty,$ together with $\sup_K\|\Delta_K\|_\mathcal{S}\leq \|\Delta\|_\mathcal{S}<\infty$, the claim follows.
\end{proof}

\end{document}